\documentclass[a4paper,11pt]{article}

\usepackage{amsthm}
\usepackage{amsmath}
\usepackage{amssymb}
\usepackage{amsfonts}
\usepackage{amstext}
\usepackage{xspace}
\usepackage{graphicx}
\usepackage{mathrsfs} 
\usepackage{fullpage}

\usepackage{vmargin}
\setmarginsrb{1in}{1in}{1in}{1in}{0pt}{0pt}{0pt}{6mm}

 \newtheorem{theorem}{Theorem}[section]
 \newtheorem{lemma}{Lemma}[section]
 
 \newtheorem{corollary}{Corollary}
 \newtheorem{definition}{Definition}[section]
 
 \newtheorem{proposition}{Proposition}[section]

\newcommand{\tw}{{\mathbf{tw}}}

\newcommand{\pw}{\mbox{\bf pw}}
\newcommand{\h}[1]{\end{document}}

\usepackage{boxedminipage}
 
 \usepackage{todonotes}

\newcommand{\cS}{{\mathcal{S}}}
\newcommand{\cO}{{\mathcal{O}}}

\newcommand{\defparproblem}[4]{
  \vspace{1mm}
\noindent\fbox{
  \begin{minipage}{0.96\textwidth}
  \begin{tabular*}{\textwidth}{@{\extracolsep{\fill}}lr} #1  & {\bf{Parameter:}} #3 \\ \end{tabular*}
  {\bf{Input:}} #2  \\
  {\bf{Question:}} #4
  \end{minipage}
  }
  \vspace{1mm}
}

\usepackage{bold-extra}
\usepackage{lmodern}
\usepackage[T1]{fontenc}
\usepackage{textcomp}

\newcommand{\I}{\cal{I}}

\newcommand{\mat}{$M=(E,{\cal I})$}

\newcommand{\whnd}[1]{\widehat{#1}}

\newcommand{\rep}[2] {$\widehat{{\cal #1}} \subseteq_{rep}^{#2} {\cal #1}$}

\newcommand{\lminrep}[2] {$\widehat{#1} \subseteq_{minrep}^{#2} #1$}

\newcommand{\minrep}[2] {$\widehat{{\cal #1}} \subseteq_{minrep}^{#2} {\cal #1}$}
\newcommand{\maxrep}[2] {$\widehat{{\cal #1}} \subseteq_{maxrep}^{#2} {\cal #1}$}

\newcommand{\bnoml}[2]{  $\binom{{#1}}{{#2}}$}

\newcommand{\rank}[1]{$\mbox{\sf rank}(#1)$}
\newcommand{\wf}{${w}:{\cal S} \rightarrow \mathbb{N}$}

\newcommand{\tgem}{$\cO\left({p+q \choose p} t p^\omega + t {p+q \choose q} ^{\omega-1} \right)$}
\newcommand{\repmat}[1]{$A_{#1}$}

\title{Representative Sets of Product Families}
\author{
{\large\sc Fedor V. Fomin\thanks{University of Bergen, Norway. \texttt{\{fomin|daniello\}@ii.uib.no}}}\addtocounter{footnote}{-1}
\and {\large \sc Daniel Lokshtanov}\footnotemark 
\and {\large \sc Fahad Panolan}\thanks{Institute of Mathematical Sciences, India. \texttt{\{fahad|saket\}@imsc.res.in}} \addtocounter{footnote}{-1}
\and {\large\sc Saket Saurabh\footnotemark \addtocounter{footnote}{-2}$~$\footnotemark}
}

\date{}
\begin{document}
\maketitle

\thispagestyle{empty}
 \begin{abstract}
%\input{abstract.tex}
%!TEX root = main.tex

A subfamily ${\cal F}'$ of a set family ${\cal F}$ is said to $q$-{\em represent}  ${\cal F}$ if for every $A \in {\cal F}$ and $B$ of size $q$ such that $A \cap B = \emptyset$ there exists a set $A' \in {\cal F}'$ such that $A' \cap B = \emptyset$. In a recent paper [SODA 2014] three of the authors gave an algorithm that given as input a family ${\cal F}$ of sets of size $p$ together with an integer $q$, efficiently computes a $q$-representative family ${\cal F'}$ of ${\cal F}$ of size approximately ${p+q \choose p}$, and demonstrated several applications of this algorithm. In this paper, we consider the efficient computation of $q$-representative sets for {\em product} families ${\cal F}$. A family ${\cal F}$ is a product family if there exist families ${\cal A}$ and ${\cal B}$ such that ${\cal F} = \{A \cup B~:~A \in {\cal A}, B \in {\cal B}, A \cap B = \emptyset\}$.  Our main technical contribution is an algorithm which given ${\cal A}$, ${\cal B}$ and $q$ computes a $q$-representative family ${\cal F}'$ of ${\cal F}$. The running time of our algorithm is {\em sublinear} in $|{\cal F}|$ for many choices of  ${\cal A}$,  ${\cal B}$ and $q$ which occur naturally in several dynamic programming algorithms. We also give an algorithm for the computation of $q$-representative sets for product families ${\cal F}$ in the more general setting where $q$-representation also involves independence in a matroid in addition to disjointness. This algorithm considerably outperforms the naive approach where one first computes ${\cal F}$ from ${\cal A}$ and ${\cal B}$, and then computes the $q$-representative family ${\cal F}'$ from ${\cal F}$.

We give two applications of our new algorithms for computing $q$-representative sets for product families. The first is a $3.8408^kn^{\cO(1)}$ deterministic algorithm for the \textsc{Multilinear Monomial Detection} (\textsc{$k$-MlD}) problem. The second is a significant improvement of deterministic dynamic programming algorithms for ``connectivity problems'' on graphs of bounded treewidth.
 \end{abstract}
% \tableofcontents
\newpage
\setcounter{page}{1}

%\newpage
\section{Introduction}
%!TEX root = main.tex

Let \mat{} be a matroid and let  ${\cal S}=\{S_1, \dots, S_t\}$ be a family of subsets of $E$ of size  $p$. A subfamily $\widehat{\cal{S}}\subseteq \cal S$ is {\em $q$-representative} for $\cal S$ if  for every set $Y\subseteq  E$ of size at most $q$, if there is a set $X \in \cal S$ disjoint from $Y$ with $X\cup Y \in \I$, then there is a set $\widehat{X} \in \widehat{\cal S}$ disjoint from $Y$  with $\widehat{X} \cup  Y \in \I$. By the classical result of Lov{\'a}sz \cite{Lovasz77}, there exists a representative family  \rep{S}{q} with at most  $\binom{p+q}{p}$ sets. 
However,  it is a very non-trivial question how to construct  such a representative family  efficiently. 
It appeared already in the 1980's that representative families can be extremely useful in dynamic programming algorithms and that faster computation of representative families leads to more efficient algorithms. 

Recently, three of the authors in \cite{FominLS13} showed that a $q$-representative family  with at most  \bnoml{p+q}{p} sets can be found in  \tgem \, operations over the field representing the matroid. Here, $\omega<2.373$ is the matrix multiplication exponent.
For the special case of uniform matroids on $n$ elements,  a faster algorithm computing a representative family in time  $\cO((\frac{p+q}{q})^q \cdot 2^{o(p+q)}\cdot t \cdot \log{n})$ was given. The results of Fomin et al.~\cite{FominLS13} improved over previous work by Monien~\cite{Monien85} and Marx~\cite{Marx:2006ys,Marx09}, and led to the fastest known deterministic parameterized algorithms for   \textsc{$k$-Path},  \textsc{$k$-Tree}, and more generally, for  \textsc{$k$-Subgraph Isomorphism}, where the $k$-vertex pattern graph is of constant treewidth \cite{FominLS13}.

All currently known algorithms that use fast computation of representative sets as a subroutine are based on dynamic programming. It is therefore very tempting to ask whether it is possible to compute representative sets faster for families that arise naturally in dynamic programs, than for general families. A class of families which often arises in dynamic programs is the class of {\em product} families; a family ${\cal F}$ is the {\em product} of ${\cal A}$ and ${\cal B}$ if  ${\cal F} = \{A \cup B~:~A \in {\cal A}, B \in {\cal B} \wedge A \cap B = \emptyset\}$. Product families naturally appear in dynamic programs where sets represent partial solutions and two partial solutions can be combined if they are disjoint. For an example, in the $k$-{\sc Path} problem partial solutions are vertex sets of paths starting at a particular root vertex $v$, and two such paths may be combined to a longer path if and only if they are disjoint (except for overlapping at $v$). Many other examples exist---essentially product families can be thought of as a {\em subset convolution}~\cite{BellmanK62,BellmanK62a}, and the wide applicability of the fast subset convolution technique of Bjorklund et al~\cite{BjorklundHKK07}  is largely due to the frequent demand to compute product families in dynamic programs.

\medskip\noindent
\textbf{Our results.}
Our main technical contributions are two algorithms for the computation of representative sets for product  families, one for uniform, and one for linear matroids. 
For uniform matroids we give an algorithm which given an integer $q$ and  families ${\cal A}$, ${\cal B}$ of sets of sizes $p_1$ and $p_2$ over the ground set of size $n$,   computes a $q$-representative family ${\cal F}'$ of ${\cal F}$. The running time of our algorithm is {\em sublinear} in $|{\cal F}|$ for many choices of  ${\cal A}$,  ${\cal B}$ and $q$ which occur naturally in several dynamic programming algorithms. 
%For any choice of 
For example, let $q$, $p_1$, $p_2$ be integers. Let  $k=q+p_1+p_2$ and suppose that we have  families 
${\cal A}$ and  ${\cal B}$, which are $(k-p_1)$ and $(k-p_2)$-representative families. Then the sizes of these families are roughly 
$|{\cal A}|={{k}\choose{p_1}}$ and $|{\cal B}|={{k}\choose{p_2}}$.  In particular, when $p_1=p_2=\lceil k/2\rceil$  both families are of size roughly $2^k$, and thus the cardinality of ${\cal F}$ is approximately $4^k$.  On the other hand, for any choice of $p_1$,  $p_2$, and $k$, our algorithm outputs a $(k-p_1-p_2)$-representative family of  ${\cal F}$  of size roughly ${{k}\choose{p_1+p_2}}$ in time $3.8408^k n^{\cO(1)}$. For many choices of $p_1$, $p_2$ and $q$ our algorithm runs significantly faster than  $3.8408^k n^{\cO(1)}$. The expression capturing the running time dependence on $p_1$, $p_2$ and $q$ can be found in 
%The running time of our algorithm is expressed as a function of $p_1$, $p_2$ and $q$, and can be tuned accordingly to the choices of these parameters, see 
Theorem~\ref{thm:product_uniform} and Corollary~\ref{cor:product_uniform}.
% for precise statements.

Our second algorithm is for computing representative families of product families, when the universe is also enriched with a linear matroid. More formally, let \mat{} be a matroid and let ${\cal A},  {\cal B}\subseteq {\cal I}$. Then  let
 % be families of independent sets of sizes $p_1$ and $p_2$ respectively. 
$
 {\cal F} ={\cal A}\bullet {\cal B}=\{A \cup B~: A \cup B\in {\cal I},  A\in  {\cal A}, B \in {\cal B} \text{ and } A \cap B = \emptyset\}.
$
Just as for uniform matroids, a naive approach for computing a representative familiy of ${\cal F} $ would be to compute the product ${\cal A}\bullet {\cal B}$ first and then compute a representative family of the product. The fastest currently known algorithm for computing a representative family is by Fomin et al.~\cite{FominLS13} and has running time approximately ${p+q \choose p}^{\omega-1}|{\cal F}|$. We give an algorithm that significantly outperforms the naive approach. An appealing feature of our algorithm is that it works by reducing the computation of a representative family for ${\cal F}$ to the computation of represesentative families for many smaller families. Thus an improved algorithm for the computation of representative sets for general families will automatically accelerate our algorithm for product families as well. The expression of the running time of our algorithm can be found in Theorem~\ref{thm:repset_product}.

%
%A crucial step in every dynamic programming algorithms is the step where where partial solutions to smaller subproblems are combined to form partial solutions to larger subproblems. In many dynamic programs over subsets two partial solutions can be combined to a partial solution of a larger subproblem if and only if they are disjoint. In this case this combination step can be expressed as 
% a convolution of two functions and the step consists of evaluating the values of the sets
% ${\cal A}\bullet {\cal B}=\{S_1\cup S_2~|~S_1\in {\cal A}, S_2\in {\cal B}, S_1 \cap S_2=\emptyset\}$.
%  Usually this is the most time-consuming step of dynamic programming and there is a lot of study of different convolutions occurring in dynamic programs starting from 
%  the work of Bellman and Karush in the early 1960's
%\cite{BellmanK62,BellmanK62a}, [FFT here? D.K. Maslen, D.N. Rockmore, Generalized FFTs�a survey
%of some recent results, in: L. Finkelstein, W.M. Kantor (Eds.),
%Groups and Computation, II
%, American Mathematical Society, Providence, RI, 1997, pp.
%183�237]. The recent applications of efficient algorithms for subset convolution in parameterized algorithms can be found in  the works of  Bj\"{o}rklund  et al. 
% \cite{BjorklundHKK07}\todo[inline]{ [MORE HERE?]}

% 
% Our main technical contributions are two algorithms computing representative families of product ${\cal A}\bullet {\cal B}$ much faster.  
% \todo[inline]{[HERE. TRY TO EXPLAIN THE STATEMENTS OF TWO THEOREMS?]}

\medskip\noindent
\textbf{Applications.} 
 Our first application is a deterministic algorithm for the following parameterized version of multilinear monomial testing.

\defparproblem{{\sc Multilinear Monomial Detection ($k$-MlD)}}{An arithmetic circuit $C$ over $\Bbb{Z}^+$ representing a polynomial $P(X)$ over $\Bbb{Z}^+$.}{$k$}{Does  $P (X)$ construed as a sum of monomials contain a multilinear monomial of degree $k$?}

This is the  central problem in the algebraic approach of Koutis and Williams for designing fast parameterized algorithms \cite{Koutis08,Koutis12,KW09,Williams09}.  The idea behind  the approach is to translate a given  problem into the language of algebra by reducing it to the problem of deciding whether a constructed polynomial has a multilinear monomial of degree $k$. As it is mentioned implicitly by Koutis in \cite{Koutis08}, \textsc{$k$-MlD} can be solved in time $(2e)^kn^{\cO(1)}$, where $n$ is the input length,  by making use of  color coding. The color coding technique of Alon, Yuster and Zwick \cite{AlonYZ} is a fundamental and widely used technique in the design of parameterized algorithms. It appeared that most of the  problems solvable by making use of color coding can be reduced to a multilinear monomial testing.  
 Williams  \cite{Williams09} gave a \emph{randomized}  algorithm solving  \textsc{$k$-MlD}  in time $2^kn^{\cO(1)}$.  
The  algorithms based on the algebraic method of Koutis-Williams provide a dramatic  improvement for a number of fundamental problems \cite{bjrklund_et_al:LIPIcs:2013:3919,BjHuKK10,FominLRS12,GuillemotS13,Koutis08,Koutis12,KW09,Williams09}. 
 
% 
%even faster algorithms can be obtained by making use of the algebraic method. For example, for \textsc{$k$-Path}, direct  reduction to \textsc{$k$-MlD} solves the problem $2^kn^{O(1)}$

%
% with additional work, it is possible 

The advantage of the algebraic approach over color coding is that for a number of parameterized problems, the algorithms based on this approach have much better exponential dependence on the parameter.  On the other hand color coding based algorithms admit direct derandomization~\cite{AlonYZ} and are able to handle integer weights with running time overhead poly-logarithmic in the weights. Obtaining deterministic algorithms matching the running times of the algebraic methods, but sharing these nice features of color coding remain a challenging open problem. 
%On the other hand,   these algorithms do not possess the following nice properties of color coding: the possibility of being derandomized without a significant loss in the running time and  handling problems with weights.  
  Our deterministic algorithm for \textsc{$k$-MlD} is the first non-trivial step towards resolving this problem. In fact, our algorithm solves a weighted version of \textsc{$k$-MlD}, where the elements of $X$ are assigned weights and the task is to find a $k$-multilinear term with minimum weight. The running time of our deterministic algorithm is 
  $\cO(3.8408^{k}2^{o(k)}s(C) n \log{W}\log^2{n})$, where $s( C)$ is the size of the circuit and $W$ is the maximum weight of an element from $X$.  We also provide an algorithm for a more general version of multilinear monomial testing, where   variables of  a monomial   should form  an  independent set  of a linear matroid.  

%The advantage of the algebraic approach over color coding is that for a number of parameterized problems, the algorithms based on this approach have much better exponential dependence on the parameter. On the other hand, 
%  these algorithms do not possess the following nice properties of color coding: the possibility of being derandomized without a significant loss in the running time and  handling problems with weights.  Our deterministic algorithm for \textsc{$k$-MlD} is the first non-trivial step towards resolving these drawbacks. In fact, our algorithm solves a weighted version of \textsc{$k$-MlD}, where the elements of $X$ are assigned weights and the task is to find a $k$-multilinear term with minimum weight. The running time of our deterministic algorithm is 
%  $\cO(3.8408^{k}2^{o(k)}s(C) n \log{W}\log^2{n})$, where $s( C)$ is the size of the circuit and $W$ is the maximum weight of an element from $X$.  We also provide an algorithm for a more general version of multilinear monomial testing, where   variables of  a monomial   should form  an  independent set  of a linear matroid.  
%  

%%Pipelined with known reductions 
%%\todo[inline]{[WHAT WE SAY HERE AND HOW?]} 
%\todo[inline]{Do we mention the work of  Chen and Chen for monomial testing  Chen and Chen \cite{ChenC13}
% $O^*(5.44^ks^4(n))$ and 
% Chen \cite{Chen13a}? Zehavi \cite{Zehavi13}?}
% %\todo[inline]{Weights. Does Williams approach gives }

  The second application of our fast computation of representative families is for dynamic programming algorithms on graph of bounded treewidth.  
  It is well known that many intractable problems can be solved efficiently when   the input graph has bounded treewidth. Moreover, many fundamental problems like 
\textsc{Maximum Independent Set} or \textsc{Minimum Dominating Set} can be solved in  time $2^{\cO(t)}n$.
  On the other hand, it was believed until very recently  that for some ``connectivity''  problems such as  \textsc{Hamiltonian Cycle} or \textsc{Steiner Tree}  no such algorithm exists. 
In their  breakthrough paper, Cygan et al.  
\cite{Cygan11} introduced a new algorithmic framework called Cut\&Count and used it to obtain  $2^{\cO(t)}n^{\cO(1)}$   time  Monte Carlo algorithms for a number of   connectivity problems.  
Recently, Bodlaender et al.  \cite{BodlaenderCK12} obtained the first  deterministic single-exponential algorithms for these problems using two novel approaches. One of the approaches of Bodlaender et al. is  based on     rank estimations in specific matrices and the second based on matrix-tree theorem and computation of determinants.
In~\cite{FominLS13}, Fomin et al. used efficient algorithms for computing representative families of linear matroids to provide yet another approach for single-exponential algorithms on graphs of bounded treewdith. 

It is interesting to note that  for a number of  connectivity problems such as \textsc{Steiner Tree}  or  \textsc{Feedback Vertex Set}  
 the ``bottleneck'' of treewidth based dynamic programming algorithms is the \emph{join} operation. For example, as it was shown by Bodlaender et al.  in~\cite{BodlaenderCK12}, \textsc{Feedback Vertex Set} and \textsc{Steiner Tree} can be solved in time  $\cO\left((1+2^\omega)^{\pw} \pw^{\cO(1)}  n  \right)$ and $\cO\left((1+2^{\omega+1})^{\tw} \tw^{\cO(1)}  n  \right)$, where $\pw$ and $\tw$ are the pathwidth and the treewidth of the input graph. 
%They  also obtained another algorithms with running time 
% $\cO\left(10^{\tw} \tw^{\cO(1)}  n^3  \right)$.
  The reason for  the difference in the  exponents of these two algorithms is   due to the cost of  the join operation, which is required for treewidth and does not occur for pathwidth.   For many computational problems on graphs of bounded treewidth in the join 
nodes of the decomposition, the family of partial solutions is the product of the families of its children, and we wish to store a representative set (for a graphic matroid) for this product family. Here our second algorithm comes into play. By making use of this algorithm one can obtain faster deterministic algorithms for many connectivity problems. We exemplify this by providing algorithms with running time 
 $\cO\left(( 1+ 2^{\omega-1}\cdot 3)^{\tw}{\tw}^{O(1)}n \right)$ for \textsc{Feedback Vertex Set} and \textsc{Steiner Tree}.

\medskip
\noindent
{\bf Our methods.} 
The engine behind our algorithm for the computation of representative sets of product families is a new construction of pseudorandom coloring families. A {\em coloring} of a universe $U$ is simply a function $f : U \rightarrow \{red, blue\}$. Consider a pair of disjoint sets $A$ and $B$, with $|A|=p$ and $|B|=q$. A random coloring which colors each element in $U$ red with probability $\frac{p}{p+q}$ and blue with probability $\frac{q}{p+q}$ will color $A$ red and $B$ blue with probability roughly 
$\frac{1}{{p+q \choose p}}$. Thus a family of slightly more than ${p+q \choose p}$ such random colorings will contain, with high probability, for each pair of disjoint sets $A$ and $B$, with $|A|=p$ and $|B|=q$ a function which colors $A$ red and $B$ blue. The fast computation of representative sets of Fomin et al.~\cite{FominLS13} deterministically constructs a collection of colorings which mimics this property of random coloring families. The colorings in the family are used to {\em witness disjointedness}, since a coloring which colors $A$ red and $B$ blue certifies that $A$ and $B$ are disjoint. In our setting we can use such coloring families both for witnessing disjointedness in the computation of representive sets, and in the computation of ${\cal F} = {\cal A} \bullet {\cal B}$. After all, each set in ${\cal F}$ is the disjoint union of a set in ${\cal A}$ and a set in ${\cal B}$. In order to make this idea work we need to make a deterministic construction of coloring familes which mimics even more properties of random colorings than the construction from~\cite{FominLS13}. We believe that the new construction of coloring families will find applications beyond our algorithm. We demonstrate this by showing how the new construction can be used to speed-up the deterministic algorithm for \textsc{$k$-Path} of Fomin et al.~\cite{FominLS13}  from 
$\cO(2.851^{k} n \log^2 {n})$ to 
  $\cO(2.619^k n \log^2 {n})$.

For linear matroids, our algorithm computes a representative family ${\cal F'}$ of  ${\cal F} = {\cal A} \bullet {\cal B}$ as follows. First the family ${\cal F}$ is broken up into many smaller families ${\cal F}_1, \ldots,  {\cal F}_t$, then a representative family ${\cal F}_i'$ is computed for each ${\cal F}_i$. Finally ${\cal F}'$ is obtained by computing a representative family of $\bigcup_i {\cal F}_i'$ using the algorithm of Fomin et al~\cite{FominLS13} for computing representative families. The speedup over the naive method is due to the fact that (a) $\bigcup_i {\cal F}_i'$ is much smaller than ${\cal F}$ and (b) that each ${\cal F}_i$ has a certain structure which ensures better upper bounds on the size of ${\cal F}_i'$, and allows ${\cal F}_i'$ to be computed faster.

% and run the above algorithm. 

%%%%%%%%%%%%%%%%%%%%%%%%%%%%%%%%%%%%%%%%%%%%%include last
 \section{Preliminaries}\label{sec:prelim}
 %!TEX root = main.tex

In this section we give various definitions which we make use of in the paper. 

\medskip

\noindent 
{\bf Graphs.} Let~$G$ be a graph with vertex set $V(G)$ and edge set $E(G)$. A graph~$G'$ is a 
\emph{subgraph} of~$G$ if~$V(G') \subseteq V(G)$ and~$E(G') \subseteq E(G)$. 
The subgraph~$G'$ is called an \emph{induced subgraph} of~$G$ if~$E(G')
 = \{ uv \in E(G) \mid u,v \in V(G')\}$, in this case, $G'$~is also called the subgraph \emph{induced by~$V(G')$} 
and denoted by~$G[V(G')]$. For a vertex set $S$, by $G \setminus S$ we denote $G[V(G) \setminus S]$, 
and by $E(S)$ we denote the edge set $E(G[S])$.  
For an edge set $E'$, we denote $G\setminus E'$ to represent the graph with vertex set $V(G)$ and edge set 
$E(G)\setminus E'$. 
% For any vertex set $S$, by $E()$ 
% By $N(u)$ we denote (open) neighborhood of $u$, that is, the set of all vertices adjacent to $u$. 
% Similarly, by $N[u]=N(u) \cup \{u\}$ we define the closed neighborhood.  
% The degree of a vertex $v$ in $G$ is $|N_G(v)|$ and is denoted by $d(v)$. 
% For a subset $S \subseteq V(G)$, we define $N[S]=\cup_{v\in S} N[v]$ and $N(S) = N[S] \setminus S$.  
% By the length of the path we mean the number of edges in it. 
 
%  \medskip 
%  
%  \noindent
%  {\bf Digraphs.} Let $D$ be a digraph. By $V(D)$ and $A(D)$ we represent the vertex set and arc set of $D$, respectively. Given a
% subset $V'\subseteq V(D)$ of a digraph $D$, let $D[V']$ denote the
% digraph induced by $V'$.  A digraph $D$ is
% {\em strong} if for every pair $x,y$ of vertices there
% are directed paths from $x$ to $y$ and from $y$ to $x.$ A maximal
% strongly connected subdigraph of $D$ is called a {\em strong
% component}. A vertex $u$ of $D$ is an {\em in-neighbor} ({\em
% out-neighbor}) of a vertex $v$ if $uv\in A(D)$ ($vu\in A(D)$,
% respectively). The {\em in-degree} $d^-(v)$ ({\em out-degree}
% $d^+(v)$) of a vertex $v$ is the number of its in-neighbors
% (out-neighbors). We denote the set of in-neighbors and out-neighbors of a vertex $v$  by $N^{-}(v)$  and $N^{+}(v)$ correspondingly. 
%  A {\em closed directed walk} in a digraph $D$ is a sequence $v_0v_1\cdots v_\ell$   of vertices of $D$, not necessarily distinct, 
% such that %each of $v_i\in V(D)$, 
% $v_0=v_\ell$ and for every $0\leq i\leq \ell-1$,   $v_iv_{i+1}\in A(D)$. 
% 
% 
% 

\medskip
\noindent
{\bf Sets, Functions and Constants.}  
Let $[n]=\{0,\ldots,n-1\}$ and ${[n] \choose i}=\{X~|~X\subseteq [n],~|X|=i\}$. Furthermore for any ground set $U$, 
we  use $2^{U}$  to denote the family of all subsets of $U$. 
% For a ground set $U$, 
We call a function $f~:~2^{U}\rightarrow \mathbb{N}$, {\em additive} if for any subsets $X$ and $Y$ of $U$ we have 
that $f(X)+f(Y)=f(X\cup Y)-f(X\cap Y)$. 

%  \begin{definition}
% Given two families of sets ${\cal L}_1$ and ${\cal L}_2$, we define 
% $${\cal L}_1 \bullet {\cal L}_2=\{X \cup Y~|~X \in {\cal L}_1 \mbox{ and } Y \in {\cal L}_2 \mbox{ and } X \cap Y = \emptyset\}.$$ 
% % Let ${\cal L}_1, \ldots ,{\cal L}_r$ be $r$ families. Then 
% % \[\prod^{\bullet}_{i\in [r]} {\cal L}_i = {\cal L}_1\bullet \cdots \bullet {\cal L}_r.\]
% \end{definition}
% \begin{definition}
% For two families ${\cal A}$ and  ${\cal B}$ over (subsets of) $U$, we define  
% $${\cal A} \circ {\cal B} = \{A \cup B ~:~A \in {\cal A} \wedge B \in {\cal B}\}.$$ 
% \end{definition}

% \begin{definition}
% For a family ${\cal A}$ over (subsets of) $U$ and  set $X$, we define  
% $${\cal A} \oplus X = \{A \cup X ~:~A \in {\cal A}\}.$$ 
% \end{definition}
A monomial $Z=x_1^{s_1}\cdots x_n^{s_n}$ of a polynomial $P(x_1,\ldots,x_n)$ is called {\em multilinear} if 
$s_i\in \{0,1\}$ for all $i\in \{1,\ldots,n\}$. We say a monomial $Z=x_1^{s_1}\cdots x_n^{s_n}$ as 
{\em $k$-multilinear term}, if $Z$ is multilinear and $\sum_{i=1}^n s_i=k$.
Throughout the paper we use $\omega$ to denote the matrix multiplication exponent. The current best known bound on 
$\omega<2.373$~\cite{Williams12}.    
% We use $e$ to denote the base of natural logarithm. 

% \subsection{Randomized Algorithms}
% We follow the same notion of randomized algorithms as described in~\cite[Section~2.3]{Marx09}. That is, 
% some of the algorithms presented in this paper are randomized, which means that they can produce incorrect answer, but the probability of doing 
% so is small. We assume that the algorithm has an integer parameter $P$ given in unary, and the probability of incorrect answer is 
% $2^{-P}$. 

\subsection{Matroids and Representative Family}
In the next few subsections we give definitions related to matroids and representative family. 
For a broader overview on matroids we refer to~\cite{oxley2006matroid}. 
\begin{definition}
A pair \mat, where $E$ is a ground set and $\cal I$ is a family of subsets (called independent sets) of $E$, is a {\em matroid} if it satisfies the following conditions:
 \begin{enumerate}
 \item[\rm (I1)]  $\phi \in \cal I$. 
 \item[\rm (I2)]  If $A' \subseteq A $ and $A\in \cal I$ then $A' \in  \cal I$. 
 \item[\rm (I3)] If $A, B  \in \cal I$  and $ |A| < |B| $, then $\exists ~e \in  (B \setminus A) $  such that $A\cup\{e\} \in \cal I$.
 \end{enumerate}
\end{definition}
The axiom (I2) is also called the hereditary property and a pair $(E,\cal I)$  satisfying  only (I2) is called hereditary family.  
An inclusion wise maximal set of $\cal I$ is called a {\em basis} of the matroid. Using axiom (I3) it is easy to show that all the bases of a matroid  
have the same size. This size is called the {\em rank} of the matroid $M$, and is denoted by \rank{M}. 
The {\em uniform matroids} are among the simplest examples of matroids. A pair \mat{} over an $n$-element ground 
set $E$, is called a uniform matroid if the family of independent sets  is given by 
${\cal I}=\{A\subseteq E~|~|A|\leq k\}$, where $k$ is some constant. This matroid is also denoted as $U_{n,k}$.
 \subsection{Linear Matroids and Representable Matroids} 
 Let $A$ be a matrix over an arbitrary field $\mathbb F$ and let $E$ be the set of columns of $A$. Given $A$ we define the matroid 
 \mat{} as follows. A set $X \subseteq E$ is independent (that is $X\in \cal I$) if the corresponding columns are linearly independent over $\mathbb F$.  
The matroids that can be defined by such a construction are called {\em linear matroids}, and if a matroid can be defined by a matrix $A$ over a 
field $\mathbb F$, then we say that the matroid is representable over $\mathbb F$. 
That is, a matroid \mat{} of rank $d$ is representable over a field $\mathbb F$ if there exist vectors in 
$\mathbb{F}^d$ correspond to the elements such that linearly independent sets of vectors 
correspond to independent sets of the matroid. A matroid \mat{}  is called {\em representable} or {\em linear} if 
it is representable over some field $\mathbb F$. 

 \subsection{Graphic Matroids}  
 Given a graph $G$, a graphic matroid \mat{} is defined by taking elements as edges of $G$ (that is $E=E(G)$) and $F\subseteq E(G)$ is in $\cal I$ if it forms a spanning forest in the  graph $G$.  The graphic matroid is representable over any field of size at least $2$. 
 Consider the matrix $A_M$  with a row for each vertex $i \in V(G)$ and a column for each edge $e = ij\in E(G)$. In the column corresponding to 
 $e=ij$, all entries are $0$, except for a $1$ in $i$ or $j$ (arbitrarily) and a $-1$ in the other. This is a representation over reals. To obtain a representation over a field $\mathbb F$, one simply needs to take the representation given above over reals and simply replace all $-1$ 
 by the additive inverse of $1$ 

 \begin{proposition}[\cite{oxley2006matroid}]
\label{prop:graphicrep}
Graphic matroids are representable over any field of size at least $2$.
\end{proposition}
 
% \subsection{Truncation of a Matroid.} The {\em $t$-truncation} of a matroid \mat{}  is a matroid $M'=(E,{\cal I}')$ such that $S\subseteq  E$ is independent in $M'$  if and only if $|S| \leq t$ and $S$ is independent in $M$ (that is $S\in \cal I$).  
% 
% \begin{proposition}[{\cite[Proposition 3.7]{Marx09}}]
% \label{prop:truncationrep}
% Given a matroid $M$ with a representation $A$ over a finite field $\mathbb F$ and an integer $t$,  a representation of the $t$-truncation $M'$ 
% can be found in randomized polynomial time.
% \end{proposition}

\subsection{Representative Family}
%!TEX root = main.tex

In this section we define $q$-representative family of a given family and state Theorems~\cite{FominLS13} regarding 
its compuation.  
% We start with the definition of a {\em $q$-representative family}.  
\begin{definition}[{\bf $q$-Representative Family}~\cite{FominLS13}]
Given a matroid  \mat{} and a family $\cal S$ of subsets of $E$, we say that a subfamily 
$\widehat{\cal{S}}\subseteq \cal S$ is {\em $q$-representative} for $\cal S$ if the following holds: for every 
set $Y\subseteq  E$ of size at most $q$, if there is a set $X \in \cal S$ disjoint from $Y$ with $X\cup Y \in \I$, 
then there is a set $\whnd{X} \in \whnd{\cal S}$ disjoint from $Y$ with $\whnd{X} \cup  Y \in \I$.  
If $\widehat{\cal S} \subseteq {\cal S}$ is $q$-representative for ${\cal S}$ we write \rep{S}{q}. 
\end{definition}

In other words if some independent set in $\cal S$ can be extended to a larger independent set by $q$ new elements, 
then there is a set in $\widehat{\cal S}$ that can be extended by the same $q$ elements. A weighted variant of 
$q$-representative families  is defined as follows. It is useful for solving problems where we are looking for 
objects of maximum or minimum weight.
\begin{definition}[{\bf Min/Max $q$-Representative Family}~\cite{FominLS13}]
Given a matroid  \mat{}, a family $\cal S$ of subsets of $E$ and a non-negative weight function \wf{} we say that a 
subfamily $\widehat{\cal{S}}\subseteq \cal S$ is {\em min $q$-representative} ({\em max $q$-representative}) for 
$\cal S$ if the following holds: for every set $Y\subseteq  E$ of size at most $q$, if there is a set 
$X \in \cal S$ disjoint from $Y$ with $X\cup Y \in \I$, then there is a set $\whnd{X} \in \whnd{\cal S}$ disjoint 
from $Y$ with 
\begin{enumerate}
\item  $\whnd{X} \cup  Y \in \I$; and 
\item $w(\whnd{X} )\leq w(X)$ ($w(\whnd{X} )\geq w(X)$).  
\end{enumerate}
We use  \minrep{S}{q} (\maxrep{S}{q}) to denote a  min $q$-representative (max $q$-representative) family for 
$\cal S$. 
\end{definition}
 \begin{definition}
Given two families of independent sets ${\cal L}_1$ and ${\cal L}_2$ of a matroid \mat{}, we define 
$${\cal L}_1 \bullet {\cal L}_2 = \{X \cup Y ~|~ X\in {\cal L}_1 \wedge Y\in {\cal L}_2 \wedge X \cap Y = \emptyset\wedge X \cup Y\in {\cal I}\}.$$ 
\end{definition}
For normal set families ${\cal A}$ and ${\cal B}$ (in uniform matroid), note that 
${\cal A} \bullet {\cal B} = \{X \cup Y ~|~ X\in {\cal A} \wedge Y\in {\cal B} \wedge X \cap Y = \emptyset\}.$ 
% We will use ${\cal A}\sqcup {\cal B}$ for normal set family and ....

We say that a family  $\cS = \{S_1,\ldots, S_t\}$ of independent sets is a {\em $p$-family} if each set in $\cal S $ 
is of size $p$. 
We state three lemmata providing basic results about representative family. 
These lemmata works for weighted variant representative family.   
% These lemmata will be used in Section~\ref{section:application}, where we provide algorithmic applications of 
% representative families. We proof them for unweighted representative families but they can be easily modified 
% to work for weighted variant. 
%
\begin{lemma}[\cite{FominLS13}]
\label{lem:reptransitive}
Let \mat{} be a matroid and $\cal S$ be a family of subsets of $E$. If  ${\cal S}' \subseteq_{rep}^{q} {\cal S}$  and $\widehat{{\cal S}}\subseteq_{rep}^{q} {\cal S}'$, then   \rep{S}{q}.   
\end{lemma}
%
% \begin{proof}
% Let $Y\subseteq  E$ of size at most $q$ such that there is a set $X \in \cal S$ disjoint from $Y$ with $X\cup Y \in \I$. By the definition of $q$-representative family we have that  there is a set $X'\in {\cal S}'$ disjoint from $Y$ with $X' \cup  Y \in \I$. Now the fact that 
% $\widehat{{\cal S}}\subseteq_{rep}^{q} {\cal S}'$ yields that there exists a $\whnd{X} \in \whnd{\cal S}$ disjoint from $Y$ with $\whnd{X} \cup  Y \in \I$. 
% \end{proof}
%
\begin{lemma}[\cite{FominLS13}]
\label{lem:repunion}
Let \mat{} be a matroid and $\cal S$ be a family of subsets of $E$. If  ${\cal S}={\cal S}_1\cup \cdots \cup {\cal S}_\ell$ and 
$\widehat{\cal{S}}_i\subseteq_{rep}^q {\cal S}_i$,  then $\cup_{i=1}^\ell \widehat{\cal{S}}_i \subseteq_{rep}^{q} {\cal S}$.   
\end{lemma}
%
% \begin{proof}
% Let $Y\subseteq  E$ of size at most $q$ such that there is a set $X \in \cal S$ disjoint from $Y$ with $X\cup Y \in \I$. Since  ${\cal S}={\cal S}_1\cup \cdots \cup {\cal S}_\ell$, there exists an $i$ such that $X\in {\cal S}_i$. This implies that there exists a $\whnd{X} \in \whnd{\cal S}_i \subseteq \cup_{i=1}^\ell \widehat{\cal{S}}_i $ disjoint from $Y$ with $\whnd{X} \cup  Y \in \I$. 
% \end{proof}
%
\begin{lemma}[\cite{FominLS13}]
\label{lem:repconvolution}
Let \mat{} be a matroid of rank $k$ and ${\cal S}_1$ be a $p_1$-family of independent sets,  ${\cal S}_2$ be a $p_2$-family of independent sets,  $\widehat{\cal{S}}_1 \subseteq_{rep}^{k-p_1} {\cal S}_1$ and  $\widehat{\cal{S}}_2 \subseteq_{rep}^{k-p_2} {\cal S}_2$. Then $\widehat{\cal{S}}_1 \bullet \widehat{\cal{S}}_2 \subseteq_{rep}^{k-p_1-p_2} {\cal S}_1 \bullet {\cal S}_2$.
\end{lemma}
%
% \begin{proof} 
% Let $Y\subseteq  E$ of size at most $q = k-p_1-p_2$ such that there is a set $X \in {\cal S}_1\bullet {\cal S}_2$ disjoint from $Y$ with $X\cup Y \in \I$.  This implies that there exist $X_1\in {\cal S}_1$ and $X_2 \in {\cal S}_2$ such that $X_1\cup X_2=X$ and $X_1 \cap X_2=\emptyset$.  Since 
% $\widehat{\cal{S}}_1 \subseteq_{rep}^{k-p_1} {\cal S}_1$, we have that there exists a $\widehat{X}_1\in \widehat{\cal{S}}_1$ such that 
% $\widehat{X}_1 \cup X_2 \cup Y\in \cal I$ and $\widehat{X}_1 \cap (X_2 \cup Y) = \emptyset$. Now since $\widehat{\cal{S}}_2 \subseteq_{rep}^{k-p_2} {\cal S}_2$, we have that there exists a $\widehat{X}_2\in \widehat{\cal{S}}_2$ such that $\widehat{X}_1 \cup \widehat{X}_2 \cup Y\in \cal I$ and $\widehat{X}_2 \cap (\widehat{X}_1 \cup Y) = \emptyset$.  This shows that $\widehat{X}_1 \cup \widehat{X}_2  \in 
% \widehat{\cal{S}}_1 \bullet \widehat{\cal{S}}_2 $ and $\widehat{X}_1 \cup \widehat{X}_2 \cup Y\in \cal I$ thus $\widehat{\cal{S}}_1 \bullet \widehat{\cal{S}}_2 \subseteq_{rep}^{k-p_1-p_2} {\cal S}_1 \bullet {\cal S}_2$.
% \end{proof}
%

\begin{theorem}[\cite{FominLS13}]
\label{thm:repsetlovaszweighted}
Let \mat{}   be a linear matroid of rank $p+q=k$, $ \cS = \{S_1,\ldots, S_t\}$ be a $p$-family of independent sets and 
\wf{} be a non-negative weight function. Then there exists \minrep{S}{q} (\maxrep{S}{q}) of size \bnoml{p+q}{p}.  
Moreover, 
given a representation \repmat{M}  of $M$ over a field $ \mathbb{F}$, we can find  \minrep{S}{q} (\maxrep{S}{q}) of size at most   \bnoml{p+q}{p} in   \tgem \, operations over $ \mathbb{F}$. 
 \end{theorem}

% In Theorem~\ref{thm:repsetlovaszweighted}  we assumed that \rank{M}$=p+q$. However, one can obtain a similar result even when 
% \rank{M}$>p+q$ in lieu of randomness.  To do this we first need to compute the representation matrix of a $k$-restriction of \mat. 
% For that we make use of  Proposition~\ref{prop:truncationrep}. This step returns a representation of a $k$-restriction of \mat{} with a high probability.  Given this matrix, we apply Theorem~\ref{thm:repsetlovaszweighted} and arrive at  the following result. 
% 
% \begin{theorem}
% \label{thm:repsetlovaszrandomized}
% Let \mat{}   be a linear matroid and let $ \cS = \{S_1,\ldots, S_t\}$ be a $p$-family of independent sets. Then there 
% exists \rep{S}{q} of size \bnoml{p+q}{p}. 
% Furthermore, given a representation \repmat{M}  of $M$ over a field $ \mathbb{F}$, there is a randomized algorithm  computing  
% \rep{S}{q} in  \tgem \, operations over $ \mathbb{F}$. 
%  \end{theorem}
\begin{theorem}[\cite{FominLS13}]
\label{thm:repsetweighted} 
There is an algorithm that given a $p$-family ${\cal A}$ of sets over a universe $U$ of size
$n$, an integer $q$, and a non-negative weight function $w~:~{\cal A}\rightarrow {\mathbb N}$ with maximum value 
at most $W$, computes in time 
$\cO\left(|{\cal A}|\cdot \left(\frac{p+q}{q}\right)^q\cdot\log n +|{\cal A}|\cdot \log |{\cal A}|\cdot\log W\right)$ 
a subfamily $\widehat{\cal A}\subseteq {\cal A}$ such that 
$|\widehat{\cal A}|\leq {p+q \choose q}\cdot 2^{o(p+q)}\cdot\log n$ and \minrep{A}{q} (\maxrep{A}{q})
\end{theorem}

    \section{Representative set computation for product families}
    %!TEX root = main.tex

% We say that a family ${\cal F}$ {\em separates} a set $A$ from a set $B$ if there is an $F \in {\cal F}$ such that $A \subseteq F$ and $B \cap F = \emptyset$.

In this section we design a faster algorithm to find $q$-representative family for product families. Our main technical tool is a generalization of {\em $n$-$p$-$q$-separating collection} defined in~\cite{FominLS13} to compute  $q$-representative families of an arbitrary family. In fact we design a {\em family} of $n$-$p$-$q$-separating collections of various sizes governed by a parameter $0<x<1$.  The construction of generalized $n$-$p$-$q$-separating collection  is similar to the proof given in~\cite{FominLS13}. However, the new construction  requires some additional ideas and the proof is slightly more involved. Finally, we combine two $n$-$p$-$q$-separating collections obtained with different parameters to obtain the desired algorithm for product families. 

\subsection{Generalized $n$-$p$-$q$-separating collection}

We start with the formal definition of {\em generalized $n$-$p$-$q$-separating collection}.

\begin{definition}
\label{def:twincollection}
A generalized $n$-$p$-$q$-separating collection ${\cal C}$ is a tuple $({\cal F}, \chi, \chi')$, where ${\cal F}$ is a family of sets over a universe $U$ of size $n$, $\chi$ is a function from 
$\bigcup_{p'\leq p}{U\choose p'}$ to $2^{\cal F}$ and $\chi'$ is a function from $\bigcup_{q'\leq q}{U\choose q'}$ to $2^{\cal F}$ such that the following properties are satisfied
\begin{enumerate}
 \item for every $A\in \bigcup_{p'\leq p}{U\choose p'}$ and $F \in \chi(A)$, $A \subseteq F$,
 \item for every $B\in \bigcup_{q'\leq q}{U\choose q'}$ and $F \in \chi'(B)$, $F\cap B=\emptyset$, 
 \item for every pairwise disjoint sets $A_1\in {U \choose p_1},A_2\in {U \choose p_1},\cdots, A_r\in {U \choose p_r}$ and $B \in {U \choose q}$ such that $p_1+\cdots+p_r=p$, 
$\exists F\in \chi(A_1)\cap\chi(A_2)\ldots\chi(A_r)\cap \chi'(B)$.
%  \item for every $A_1 \in {U \choose p_1}$, $A_2\in{U\setminus A_1 \choose p_2}$ and $B \in {U\setminus (A_1\cup A_2) \choose q}$, $\exists F\in \chi(A_1)\cap \chi(A_2)\cap \chi'(B)$  
% such that $A_1,A_2\subseteq F$ and $F\cap B=\emptyset$.
\end{enumerate}
The size of  $({\cal F},\chi, \chi')$ is $|{\cal F}|$, the $(\chi,p')$-degree of $({\cal F},\chi,\chi')$ for $p'\leq p$ is $\max_{A \in {U \choose p'}} |\chi(A)|$, 
% the $p_2$-degree of $({\cal F},\chi, \chi')$ is $\max_{A \in {U \choose p_2}} |\chi(A)|$ 
and the $(\chi',q')$-degree of $({\cal F},\chi, \chi')$ for $q'\leq q$ is $\max_{B \in {U \choose q'}} |\chi'(B)|$. 
\end{definition}

A {\em construction} of generalized separating collections is a data structure, that given $n$, $p$ and $q$ initializes and outputs a family ${\cal F}$ of sets over the universe $U$ of size $n$. 
After the initialization one can query the data structure by giving it a set  $A \in \bigcup_{p'\leq p}{U \choose p'}$ or $B\in \bigcup_{q'\leq q}{U \choose q'}$, the data structure 
then outputs a family  $\chi(A) \subseteq 2^{\cal F}$ or $\chi'(B)\subseteq 2^{\cal F}$ respectively. Together the tuple ${\cal C}= ({\cal F},\chi, \chi')$ computed by the data structure 
should form a {\em generalized} $n$-$p$-$q$-{\em separating collection}.

We call the time the data structure takes to initialize and output ${\cal F}$ the {\em initialization time}. 
The {\em $(\chi,p')$-query time}, $p'\leq p$, of the data structure is the maximum time the data structure uses to compute $\chi(A)$ over all $A \in {U \choose p'}$. Similarly, the 
{\em $(\chi',q')$-query time}, $q'\leq q$, of the data structure is the maximum time the data structure uses to compute $\chi'(B)$ over all $B \in {U \choose q'}$.
% The initialization time and query time of the data structure and the size and degree of ${\cal C}$ are functions of  $n$, $p_1$, $p_2$ and $q$. The initialization time is denoted by 
% $\tau_I(n,p_1,p_2,q)$, the query time by $\tau_Q(n,p_1,p_2,q)$, the resulting size of ${\cal C}$ is denoted by $\zeta(n,p_1,p_2,q)$, while the degree of  ${\cal C}$ is denoted by $\Delta(n,p_1,p_2,q)$.
The initialization time of the data structure and the size of ${\cal C}$ are functions of  $n$, $p$ and $q$. The initialization time is denoted by 
$\tau_I(n,p,q)$, size of ${\cal C}$ is denoted by $\zeta(n,p,q)$. The $(\chi,p')$-query time and $(\chi,p')$-degree of 
$\cal C$, $p'\leq p$, are functions of 
$n,p',p,q$ and is denoted by ${Q_{(\chi,p')}}(n,p,q)$ and $\Delta_{(\chi,p')}(n,p,q)$ respectively. Similarly, the $(\chi',q')$-query time and $(\chi',q')$-degree of ${\cal C}$, $q'\leq q$,  are functions of 
$n,q',p,q$ and are denoted by ${Q_{(\chi',q')}}(n,p,q)$ and $\Delta_{(\chi',q')}(n,p,q)$ respectively.  
% The main technical component in the proof of Theorem~\ref{thm:fastRepUniform} is the following lemma.
We are now ready to state the main technical  tool of this subsection.

\begin{lemma}
\label{lem:twin_sep_coll_construction}
Given a constant $x$ such that $0<x<1$, there is a construction of generalized $n$-$p$-$q$- separating collection with the following parameters
\begin{itemize} 
\setlength\itemsep{-.7mm}
\item size, $\zeta(n,p,q) \leq 2^{\cO(\frac{p+q}{\log\log\log(p+q)})}\cdot \frac{1}{x^p(1-x)^q}\cdot (p+q)^{\cO(1)} \cdot \log n$
\item initialization time, $\tau_I(n,p,q) \leq  2^{\cO(\frac{p+q}{\log\log\log(p+q)})}\cdot \frac{1}{x^p(1-x)^q}\cdot (p+q)^{\cO(1)} \cdot n\log n$
\item $(\chi,p')$-degree, $\Delta_{(\chi,p')}(n,p,q) \leq  2^{\cO(\frac{p+q}{\log\log\log(p+q)})}\cdot \frac{1}{x^{p-p'}(1-x)^q}\cdot (p+q)^{\cO(1)} \cdot \log n$
\item $(\chi,p')$-query time, $Q_{(\chi,p')}(n,p,q) \leq  2^{\cO(\frac{p+q}{\log\log\log(p+q)})}\cdot \frac{1}{x^{p-p'}(1-x)^q}\cdot (p+q)^{\cO(1)} \cdot \log n$
\item $(\chi',q')$-degree, $\Delta_{(\chi',q')}(n,p,q) \leq  2^{\cO(\frac{p+q}{\log\log\log(p+q)})}\cdot \frac{1}{x^{p}(1-x)^{q-q'}}\cdot (p+q)^{\cO(1)} \cdot \log n$
\item $(\chi',q')$-query time, $Q_{(\chi',q')}(n,p,q) \leq  2^{\cO(\frac{p+q}{\log\log\log(p+q)})}\cdot \frac{1}{x^{p}(1-x)^{q-q'}}\cdot (p+q)^{\cO(1)} \cdot \log n$
\end{itemize}
\end{lemma}

We first give a road map to prove Lemma~\ref{lem:twin_sep_coll_construction}. 
The proof of  Lemma~\ref{lem:twin_sep_coll_construction} uses three auxiliary lemmata. 
\begin{enumerate}
\item[(a.)] {\bf Existential Proof (Lemma~\ref{lem:twin_sep_coll_brute_force}}). This lemma shows that there is indeed a 
generalized $n$-$p$-$q$-separating collection with the required sizes, degrees and query time. Essentially, it shows that if we form a family  ${\cal F}=\{F_1,\ldots,F_t\}$ of sets of $U$ such that each $F_i$ is a random subset of $U$ where each element  is inserted into $F_i$ with probability $x$, then ${\cal F}$ has the desired sizes, degrees and query time. Thus, this also gives a brute force algorithm to design the family $\cal F$ by just guessing the family of desired size and then checking whether it is indeed  a generalized $n$-$p$-$q$-separating collection. 
\item[(b.)]  {\bf Universe Reduction (Lemma~\ref{lem:twinreduceUniverse}).} The construction obtained in Lemma~\ref{lem:twin_sep_coll_brute_force} 
has only one drawback that the initialization time is much larger than claimed in Lemma~\ref{lem:twin_sep_coll_construction}. To overcome this lacuna, we do not apply the construction in Lemma~\ref{lem:twin_sep_coll_brute_force} directly. 
We first prove a Lemma~\ref{lem:twinreduceUniverse} which helps us in reducing the universe size to $(p+q)^2$. This is done using the  known construction of $k$-perfect hash families of size $(p+q)^{\cO(1)} \log n$. 
Lemma~\ref{lem:twinreduceUniverse} alone  can not reduce the universe size sufficiently, that we can apply the construction of Lemma~\ref{lem:twin_sep_coll_brute_force}. 
\item[(c.)] {\bf Splitting Lemma (Lemma~\ref{lem:splitSolution}).} We give a splitter type construction in Lemma~\ref{lem:splitSolution} that when applied with 
Lemma~\ref{lem:twinreduceUniverse} makes the universe and other parameters small enough that we can apply the construction given in 
Lemma~\ref{lem:twin_sep_coll_brute_force}. In this construction we consider all the ``consecutive partitions''  of the universe into $t$ parts, assume that the sets $A\cup B$, $A=\cup_{i=1}^r A_i$, are distributed uniformly into $t$ parts and then use this information to obtain a construction of generalized separating collections in each part and then take the product of these collections to obtain a collection for the original instance.
\end{enumerate}

We start with an existential proof.

\begin{lemma}\label{lem:twin_sep_coll_brute_force}
Given $0<x<1$, there is a construction of generalized $n$-$p$-$q$-separating collections with 
\begin{itemize}\setlength\itemsep{-.7mm}
\item size $\zeta(n,p,q) =\cO \left(\frac{1}{x^{p}(1-x)^q} \cdot (p^2+q^2+1)\log n \right)$, 
% \cO({p+q \choose q} \cdot (p+q)^{\cO(1)} \cdot \log n)$, 
\item initialization time $\tau_I(n,p,q) = \cO({2^n \choose \zeta(n,p,q)} \cdot \frac{1}{x^p(1-x)^q} \cdot n^{\cO(p+q)})$,
\item $(\chi,p')$-degree for $p'\leq p$, $\Delta_{(\chi,p')}(n,p,q) = \cO\left(\frac{1}{x^{p-p'}}\cdot\frac{(p^2+q^2+1)}{(1-x)^q} \cdot \log n\right)$
% \item $p_2$-degree $\Delta_{p_2}(n,p,q) = \cO\left(\frac{1}{x^{p_2}}\cdot\frac{(p+q+1)}{(1-x)^q} \cdot \log n\right)$, and 
\item $(\chi,p')$-query time ${Q_{(\chi,p')}}(n,p,q) = \cO(\frac{1}{x^{p}(1-x)^q} \cdot n^{\cO(1)}).$
\item $(\chi',q')$-degree $\Delta_{(\chi',q')}(n,p,q)=\cO\left(\frac{1}{x^{p}(1-x)^{q-q'}}\cdot(p^2+q^2+1) \cdot \log n\right)$
\item $(\chi',q')$-query time ${Q_{(\chi',q')}}(n,p,q)=\cO(\frac{1}{x^{p}(1-x)^q} \cdot n^{\cO(1)}).$
\end{itemize}
\end{lemma}

\begin{proof}
We start by giving a randomized algorithm that with positive probability constructs a generalized $n$-$p$-$q$-separating collection ${\cal C} = ({\cal F},\chi, \chi')$ with the desired size and degree parameters. 
We will then discuss how to deterministically compute such a ${\cal C}$ within the required time bound. Set $t = \frac{1}{x^p(1-x)^q} \cdot (p^2+q^2+1)\log n$ and construct the family 
${\cal F} = \{F_1, \ldots, F_t\}$ as follows. Each set $F_i$ is a random subset of $U$, where each element of $U$ is inserted into $F_i$ with probability $x$. Distinct elements are inserted (or not) into $F_i$ 
independently, and the construction of the different sets in ${\cal F}$ is also independent. For each  $A \in \bigcup_{p'\leq p}{U\choose p'}$ we set $\chi(A) = \{F \in {\cal F}~:~A \subseteq F\}$ and 
for each $B\in \bigcup_{q'\leq q}{U\choose q'}$ we set $\chi'(B)=\{F\in {\cal F}~:~F\cap B=\emptyset\}$.

The size of ${\cal F}$ is within the required bound by construction. We now argue that with positive probability 
$({\cal F},\chi, \chi')$ is indeed a generalized $n$-$p$-$q$-separating collection, and that the degrees of ${\cal C}$ 
is within the required bounds as well. For fixed sets $A \in {U \choose p}$, $B \in {U\setminus A \choose q}$, and integer $i \leq t$, we consider the probability that $A \subseteq F_i$ and $B \cap F_i = \emptyset$. 
This probability is $x^{p}(1-x)^q$. Since each $F_i$ is constructed independently from the other sets in ${\cal F}$, the probability that {\em no} $F_i$ satisfies $A \subseteq F_i$ and $B \cap F_i = \emptyset$ is
\begin{align*} \left(1 - x^p(1-x)^q\right)^t \leq e^{-(p^2+q^2+1)\log n} = \frac{1}{n^{p^2+q^2+1}}.\end{align*}
For a fixed $A_1,A_2,\ldots,A_r$ and $B$ (choices in condition $3$), the probability that no $F_i$ in $\chi(A_1)\cap\chi(A_2)\cap \cdots \cap\chi(A_r)\cap \chi'(B)$ is equal to the 
probability that no $F_i$ in $\chi(A_1\cup A_2\cdots \cup A_r)\cap \chi'(B)$  (since $\chi(A')$ contains all 
the sets in ${\cal F}$ that contains $A'$ and $\chi'(B)$ contains all the sets in ${\cal F}$ that are disjoint from $B$). 
Hence the probability that condition $3$ fails is upper bounded by 
$$Y\cdot\frac{1}{n^{p^2+q^2+1}}$$
where $Y$ is the number of choices for $A_1,\ldots,A_r$ and $B$ in condition $3$. We upper bound $Y$ as follows.
% We upperbound the number of choices for the sets $A_1,A_2,\ldots,A_r$ and $B$ in condition $3$. 
There are ${n \choose p}$ choices for $A_1\cup\cdots\cup A_r$ and ${n \choose q}$ choices for $B$. 
For each choice of $A_1\cup\cdots\cup A_r$ there are at most $r^p$ choices of making $A_1,\ldots,A_r$ with some of them being empty as well. Note that $r\leq p$. 
Therefore the number of possible choices of sets $A_1,A_2,\ldots,A_r$ and $B$ in condition $3$ is upper bounded by ${n\choose p}{n\choose q}p^p\leq n^{2p+q}\leq n^{p^2+q^2}$. 
% There are ${n \choose p}$ choices for $A=A_1\cup A_2 \ldots A_r \in  {U \choose p}$, and ${n-p \choose q}$ choices for $B \in {U\setminus A \choose q}$, 
% Therefore the union bound yields that the probability that there exists an 
% $A \in  {U \choose p}$ and $B \in {U \choose q}$ such that $\nexists F\in\chi(A)\cap\cap\chi'(B)$, 
% is at most  $\frac{1}{n^{p^2+q^2+1}} \cdot n^{p+q} \leq \frac{1}{n}$. Since $\chi(A')$ contains all 
% the sets in ${\cal F}$ that contains $A'$ and $\chi'(B)$ contains all the sets in ${\cal F}$ that are disjoint from $B$, 
Hence the probability that condition $3$ in Definition~\ref{def:twincollection} fails is at most $\frac{1}{n}$.

We also need to upper bound the maximum degree of ${\cal C}$. For every $A \in  {U \choose p'}$, $|\chi(A)|$ is a random variable. For a fixed $A \in  {U \choose p'}$ and $i \leq t$ the probability 
that $A \subseteq F_i$ is exactly $x^{p'}$. Hence $|\chi(A)|$ is the sum of $t$ independent  $0/1$-random variables that each take value $1$ with probability $x^{p'}$. Hence the expected value of $|\chi(A)|$ is 
$$E[|\chi(A)|] = t \cdot x^{p'} = \frac{1}{x^{p-p'}(1-x)^q}\cdot (p^2+q^2+1)\log n$$ 
% Similarly for every $A'\in {U\choose p_2}$,
% $$E[|\chi(A')|] = t \cdot x^{p_2} = \frac{1}{x^{p_1}(1-x)^q}\cdot (p+q+1)\log n$$ 
For every $B\in {U\choose q'}$, $|\chi'(B)|$ is also a random variable. For a fixed $B \in  {U \choose q'}$ and $i \leq t$ the probability that $A \cap F_i=\emptyset$ is exactly $(1-x)^{q'}$. 
Hence the expected value of $|\chi'(B)|$ is,
$$E[|\chi'(B)|] = t \cdot (1-x)^{q'} = \frac{1}{x^{p}(1-x)^{q-q'}}\cdot (p^2+q^2+1)\log n.$$
Standard Chernoff bounds~\cite[Theorem 4.4]{mitzenmacher2005probability} show that the probability that for any $A\in{U\choose p'}$, $|\chi(A)|$ 
is at least $6E[|\chi(A)|]$ is upper bounded by $2^{-6E[|\chi(A)|]} \leq \frac{1}{n^{p^2+q^2+1}}$. 
Similarly the probability that for any $B\in{U \choose q'}$, $|\chi'(B)|$ is at least $6E[|\chi'(B)|]$ is upper bounded by $2^{-6E[|\chi'(B)|]} \leq \frac{1}{n^{p^2+q^2+1}}$.  
There are  $\sum_{p'\leq p}{n \choose p'}\leq{n^{p^2}}$ choices for $A\in\bigcup_{p'\leq p}{U\choose p'}$ and $\sum_{q'\leq q}{n \choose q'}\leq{n^{q^2}}$ choices for $B\in\bigcup_{q'\leq q}{U\choose q'}$. 
Hence the union bound yields that the probability that there exists an $A\in\bigcup_{p'\leq p}{U\choose p'}$ such that $|\chi(A)|  > 6E[|\chi(A)|]$ or there exists $B\in \bigcup_{q'\leq q}{U\choose q'}$ such 
that $|\chi'(B)|  > 6E[|\chi'(B)|]$ is upper bounded by $\frac {1}{n}$. Thus ${\cal C}$ is a family of $n$-$p$-$q$-separating collections with the desired size and degree parameters with probability at least 
$1 - \frac{2}{n} > 0$. The degenerate case that  $1 - \frac{2}{n} \leq 0$ is handled 
by the family ${\cal F}$ containing all (at most four) subsets of $U$. 

To construct ${\cal F}$ within the stated initialization time bound, it is sufficient to try all families ${\cal F}$ of size $t$ and for each of the ${2^n \choose \zeta(n,p,q)}$ 
guesses, test whether it is indeed a family of $n$-$p$-$q$-separating collections in time $\cO(t \cdot n^{\cO(p+q)}) = \cO(\frac{1}{x^p(1-x)^q} \cdot n^{\cO(p+q)})$.

For the queries, we need to give an algorithm that given $A$, computes $\chi(A)$ (or $\chi'(A)$), under the assumption that ${\cal F}$ has already has been computed in the initialization step. 
This is easily done within the stated running time bound by going through every set $F \in {\cal F}$, checking whether $A \subseteq F$ (or $A\cap F=\emptyset$), and if so, inserting $F$ into $\chi(A)$ ($\chi'(A)$). 
This concludes the proof.
\end{proof}
We will now work towards improving the time bounds of Lemma~\ref{lem:twin_sep_coll_brute_force}. 
 To that end we will need a construction of {\em $k$-perfect hash functions} by Alon et al.~\cite{AlonYZ}
\begin{definition} 
A family of functions $f_1, \ldots, f_t$ from a universe $U$ of size $n$ to a universe of size $r$ is a $k$-perfect family of hash functions if for every set $S \subseteq U$ such that $|S|=k$ 
there exists an $i$ such that the restriction of $f_i$ to $S$ is injective.
\end{definition}
Alon et al.~\cite{AlonYZ} give very efficient constructions of $k$-perfect families of hash functions from a universe of size $n$ to a universe of size $k^2$.
\begin{proposition}[\cite{AlonYZ}]\label{prop:hashFun} 
For any universe $U$ of size $n$ there is a $k$-perfect family $f_1, \ldots, f_t$ of hash functions from $U$ to 
% $\{1,2,\ldots,k^2\}$
$[k^2]$ 
with $t = \cO(k^{\cO(1)} \cdot \log n)$. 
Such a family of hash functions can be constructed in time $\cO(k^{\cO(1)}n \log n)$. 
\end{proposition}

\begin{lemma}\label{lem:twinreduceUniverse} If there is a construction of generalized $n$-$p$-$q$-separating collections $(\hat{\cal F},\hat{\chi},\hat{\chi}')$ with initialization time $\tau_I(n,p,q)$, size $\zeta(n,p,q)$, 
$(\hat{\chi},p')$-query time ${Q_{(\hat{\chi},p')}}(n,p,q)$, $(\hat{\chi}',q')$-query time ${Q_{(\hat{\chi}',q')}}(n,p,q)$, 
%  producing a $n$-$p$-$q$-separating collection with 
$(\hat{\chi},p')$-degree $\Delta_{(\hat{\chi},p')}(n,p,q)$, and $(\hat{\chi}',q')$-degree $\Delta_{(\hat{\chi}',q')}(n,p,q)$  
% $p_1$-degree $\Delta_{p_1}(n,p,q)$ and $p_2$-degree $\Delta_{p_2}(n,p,q)$, 
then there is a construction of generalized $n$-$p$-$q$-separating collections 
% using
with following parameters.
\begin{itemize} %\setlength\itemsep{-.7mm}
\item 
% size 
$\zeta'(n,p,q) \leq \zeta\left((p+q)^2,p,q\right) \cdot  (p+q)^{\cO(1)} \cdot \log n$,
\item 
% initialization time 
$\tau_I'(n,p,q) = \cO\left(\tau_I\left((p+q)^2,p,q\right) + \zeta\left((p+q)^2,p,q\right) \cdot (p+q)^{\cO(1)} \cdot n \log n\right)$,
\item 
% $(\chi,p')$-degree 
$\Delta'_{(\chi,p')}(n,p,q) \leq \Delta_{(\hat{\chi},p')}\left((p+q)^2,p,q\right) \cdot  (p+q)^{\cO(1)} \cdot \log n$, 
\item 
% $(\chi,p')$-query time 
${Q'_{(\chi,p')}}(n,p,q) = \cO\left(\left({Q_{(\hat{\chi},p')}}\left((p+q)^2,p,q\right) + \Delta_{(\hat{\chi},p')}\left((p+q)^2,p,q\right) \right) \cdot (p+q)^{\cO(1)} \cdot \log n\right)$,
\item 
% $(\chi',q')$-degree 
$\Delta'_{(\chi',q')}(n,p,q)\leq \Delta_{(\hat{\chi}',q')}\left((p+q)^2,p,q\right) \cdot  (p+q)^{\cO(1)} \cdot \log n$,
\item 
% $(\chi',q')$-query time \\
${Q'_{(\chi',q')}}(n,p,q)=\cO\left(\left({Q_{(\hat{\chi}',q')}}\left((p+q)^2,p,q\right) + \Delta_{(\hat{\chi}',q')}\left((p+q)^2,p,q\right) \right) \cdot (p+q)^{\cO(1)} \cdot \log n\right)$
\end{itemize}
\end{lemma}
\begin{proof}
We give a construction of generalized $n$-$p$-$q$-separating collections with initialization time, query time, size and degree $\tau_I'$, ${Q}'$, $\zeta'$ and $\Delta'$ respectively using the 
construction with initialization time, query time, size and degree $\tau_I$, ${Q}$, $\zeta$ and $\Delta$ as a black box. 
 
We first describe the initialization of the data structure. Given $n$, $p$, and $q$, we construct using Proposition~\ref{prop:hashFun} a $(p+q)$-perfect family $f_1, \ldots f_t$ 
of hash functions from the universe $U$ to $[(p+q)^2]$. The construction takes time $\cO((p+q)^{\cO(1)}n \log n)$ and $t \leq  (p+q)^{\cO(1)} \cdot \log n$. 
We will store these hash functions in memory. We use the following notations.
\begin{itemize}
 \item For a set $S \subseteq U$ and $T\subseteq [(p+q)^2]$,  \\ $f_i(S)=\{f_i(s) ~:~ s \in S\}$ and $f_i^{-1}(T)=\{s \in U ~:~ f(s) \in T\}$. 
 \item For a family ${\cal Z}$ of sets over $U$ and family ${\cal W}$ of sets over $[(p+q)^2]$,\\ $f_i({\cal Z})=\{f_i(S) ~:~ S \in {\cal Z}\}$ and $f_i^{-1}({\cal W})=\{f_i^{-1}(T) ~:~ T \in {\cal W}\}$.
\end{itemize}
% For a set $S \subseteq U$, by $f_i(S)$ we will mean $\{f_i(s) ~:~ s \in S\}$. Similarly for every $S \subseteq \{1, \ldots, (p+q)^2\}$, by $f_i^{-1}(S)$ we will mean $\{s \in U ~:~ f(s) \in S\}$. For a family ${\cal Z}$ of sets over $U$, by $f_i({\cal Z})$ we will mean $\{f_i(S) ~:~ S \in {\cal Z}\}$. Finally, for a family ${\cal Z}$ of sets over $\{1,\ldots,(p+q)^2\}$, by $f_i^{-1}({\cal Z})$ we will mean  $\{f_i^{-1}(S) ~:~ S \in {\cal Z}\}$.
 
 We first use the given black box construction for $(p+q)^2$-$p$-$q$-separating collections $(\hat{\cal F}, \hat{\chi},\hat{\chi}')$ over the universe $[(p+q)^2]$. 
% This construction computes the separating collection $(\hat{\cal F}, \hat{\chi})$. 
We run the initialization algorithm of this construction and store the family $\hat{\cal F}$ in memory. We then set
\begin{align*} {\cal F} = \bigcup_{i\leq t} f_i^{-1}(\hat{\cal F}). \end{align*}

 We spent $\cO((p+q)^{\cO(1)}n \log n)$ time to construct a $(p+q)$-perfect family of hash functions,  $\cO(\tau_I((p+q)^2,p,q))$ to construct $\hat{\cal F}$ of size $\zeta((p+q)^2,p,q)$, 
and $\cO(\zeta((p+q)^2,p,q) \cdot (p+q)^{\cO(1)} \cdot n \log n)$ time to construct ${\cal F}$ from $\hat{\cal F}$ and the family of perfect hash functions. 
Thus the upper bound on $\tau_I'(n,p,q)$ follows. Furthermore,  $|{\cal F}| \leq |\hat{\cal F}| \cdot  (p+q)^{\cO(1)} \cdot \log n$, yielding the claimed bound for $\zeta'$.
 
 We now define $\chi(A)$ for every $A \in \bigcup_{p'\leq p}{U\choose p'}$ and describe the query algorithm. For every $A \in \bigcup_{p'\leq p}{U\choose p'}$ we let
\begin{align*} \chi(A) = \bigcup_{\substack{i\leq t \\ |f_i(A)|=|A|}} f_i^{-1}(\hat{\chi}(f_i(A))). \end{align*}
Since $\forall\;\hat{F} \in \hat{\chi}(f_i(A))$, $f_i(A) \subseteq \hat{F}$, it follows that $A \subseteq F$ for every $F \in \chi(A)$. Furthermore we can bound $|\chi(A)|$ for any $A\in \bigcup_{p'\leq p}{U\choose p'}$, 
as follows 
\begin{align*} |\chi(A)| \leq \sum_{\substack{i\leq t \\ |f_i(A)|=|A|}} |\hat{\chi}(f_i(A))| \leq \Delta_{(\hat{\chi},p')}((p+q)^2,p,q) \cdot  (p+q)^{\cO(1)} \cdot \log n.\end{align*}
Thus the claimed bound for $\Delta'_{(\chi,p')}$ follows. 
Similar way we define $\chi'(B)$ for every $B\in\bigcup_{q'\leq q}{U\choose q'}$ as 
\begin{align*} \chi'(B) = \bigcup_{\substack{i\leq t \\ |f_i(A)|=|A|}} f_i^{-1}(\hat{\chi}'(f_i(A))). \end{align*}
\begin{align*} |\chi'(B)| \leq \sum_{\substack{i\leq t \\ |f_i(A)|=|A|}} |\hat{\chi}'(f_i(A))| \leq \Delta_{(\hat{\chi}',q')}((p+q)^2,p,q) \cdot  (p+q)^{\cO(1)} \cdot \log n.\end{align*}
To compute $\chi(A)$ for any $A\in\bigcup_{p'\leq p} {U \choose p'}$, we go over every $i \leq t$ and check whether $f_i$ is injective on $A$. This takes time $\cO((p+q)^{\cO(1)} \cdot \log n)$. 
For each $i$ such that $f_i$ is injective on $A$, we compute $f_i(A)$ and then $\hat{\chi}(f_i(A))$ in time $\cO({Q_{(\chi,p')}}((p+q)^2,p,q))$. Then we compute $f_i^{-1}(\hat{\chi}(f_i(A)))$  in time 
$\cO(|\hat{\chi}(f_i(A))|\cdot (p+q)^{\cO(1)}) =\cO(\Delta_{(\chi,p')}((p+q)^2,p,q)\cdot (p+q)^{\cO(1)})$ and add this set to $\chi(A)$. As we need to do this $\cO((p+q)^{\cO(1)} \cdot \log n)$ times, the total time 
to compute $\chi(A)$ is upper bounded by $\cO(({Q_{(\chi,p')}}((p+q)^2,p,q) + \Delta_{(\chi,p')}((p+q)^2,p,q)) \cdot (p+q)^{\cO(1)} \cdot \log n)$, 
yielding the claimed upper bound on ${Q'_{(\chi,p')}}$. Similar way we can bound ${Q'_{(\chi',q')}}$.

It remains to argue that $({\cal F},\chi,\chi')$ is in fact a generalized $n$-$p$-$q$-separating collection. For any $r$, consider pairwise disjoint sets $A_1 \in {U \choose p_1},\ldots,A_r \in {U \choose p_r}$, and 
$B \in {U \choose q}$ such that $p_1+\ldots+p_r=p$. We need to show that $\exists F\in \chi(A_1)\cap\cdots\cap\chi(A_r)\cap\chi'(B)$. 
Since $f_1, \ldots, f_t$ is a $(p+q)$-perfect family of hash functions, there is an $i$ such that $f_i$ is injective on $A_1\cup\cdots \cup A_r \cup B$. 
Since $(\hat{\cal F}, \hat{\chi},\hat{\chi}')$ is a $(p+q)^2$-$p$-$q$-separating collection,  
$\exists \hat{F}\in \hat{\chi}(f_i(A_1))\cap\cdots\hat{\chi}(f_i(A_r))\cap \hat{\chi}'(f_i(B))$. Since $f_i$ is injective on $A_1,\ldots,A_r$ and $B$, $f_i^{-1}(\hat{F})\in \chi(A_1)\cap\cdots \chi(A_r) \cap\chi'(B)$. 
This concludes the proof.
\end{proof}
We now give a {\em splitting lemma}, which allows us to reduce the problem of finding generalized $n$-$p$-$q$-separating collections to the same problem, but with much smaller values for $p$ and $q$. To that end we need some definitions.
% 
% A {\em partition} of $U$ is a family ${\cal U}_P = \{U_1, U_2, \ldots U_t\}$ of sets over $U$ such that $U_i \cap U_j = \emptyset$ for every $i \neq j$ and $U = \bigcup_{i \leq t} U_i$. Each of the sets $U_i$ are called the {\em parts} of the partition. A {\em consecutive partition} of $\{1,\ldots,n\}$ is a partition ${\cal U}_P = \{U_1, U_2, \ldots U_t\}$ of $\{1,\ldots,n\}$ such that for every integer $i \leq t$ and integers $1 \leq x \leq y \leq z$, if $x \in U_i$ and $z \in U_i$ then $y \in U_i$ as well. In other words, in a consecutive partition each part is a consecutive interval of integers. For every integer $t$, let $\mathscr{P}_t$ denote the collection of all consecutive partitions of $\{1,\ldots,n\}$ with exaclty $t$ parts. We do not demand that all of the parts in a partition in $\mathscr{P}_t$ are non-empty. Simple counting arguments show that for every $t$, $|\mathscr{P}_t| = {n+t-1 \choose t-1}$.
% 
% We will denote by ${\cal Z}_{s,t}^p$ the set of all $t$-tuples $(p_1,p_2, \ldots, p_t)$ of integers such that $\sum_{i \leq t} p_i = p$ and $0 \leq p_i \leq s$ for all $i$. Clearly  $|{\cal Z}_{s,t}^p| \leq {p+t-1 \choose t-1}$, since this counts all the ways of writing $p$ as a sum of $t$ non-negative integers, without considering the upper bound on each one. 

\begin{definition}
A {\em partition} of $U$ is a family ${\cal U}_P = \{U_1, U_2, \ldots U_t\}$ of sets over $U$ such that $\forall i\neq j,\;U_i \cap U_j = \emptyset$ and $U = \bigcup_{i \leq t} U_i$. 
Each of the sets $U_i$ are called the {\em parts} of the partition. A {\em consecutive partition} of $\{1,\ldots,n\}$ is a partition ${\cal U}_P = \{U_1, U_2, \ldots U_t\}$ of $\{1,\ldots,n\}$ 
such that for every integer $i \leq t$ and integers $1 \leq x \leq y \leq z$, if $x \in U_i$ and $z \in U_i$ then $y \in U_i$ as well. 
% In other words, in a consecutive partition each part is a consecutive interval of integers. For every integer $t$, 
\end{definition}
\begin{proposition}
Let $\mathscr{P}_t^{n}$ denote the collection of all consecutive partitions of $\{1,\ldots,n\}$ with exactly $t$ parts. Let 
${\cal Z}_{s,t}^p$ be the set of all $t$-tuples $(p_1,p_2, \ldots, p_t)$ of integers such that $\sum_{i \leq t} p_i = p$ and $0 \leq p_i \leq s$ for all $i$. 
% We do not demand that all of the parts in a partition in $\mathscr{P}_t$ are non-empty. Simple counting arguments show that 
Then for every $t$, $|\mathscr{P}_t^{n}| = {n+t-1 \choose t-1}$ and $|{\cal Z}_{s,t}^p| \leq {p+t-1 \choose t-1}$.
\end{proposition}
% --------------------------------------------------------------------------------------------------
% We will denote by ${\cal Z}_{s,t}^p$ the set of all $t$-tuples $(p_1,p_2, \ldots, p_t)$ of integers such that $\sum_{i \leq t} p_i = p$ and $0 \leq p_i \leq s$ for all $i$. Clearly  $|{\cal Z}_{s,t}^p| \leq {p+t-1 \choose t-1}$, since this counts all the ways of writing $p$ as a sum of $t$ non-negative integers, without considering the upper bound on each one. 
% -------------------------------------------------------------------------------------------------------------------------
\begin{lemma}\label{lem:splitSolution} For any $p$, $q$ let $s = \lfloor (\log (p+q))^2 \rfloor$ and 
$t = \lceil \frac{p+q}{s} \rceil$. If there is a construction of generalized $n$-$p$-$q$-separating collections  
with initialization time $\tau_I(n,p,q)$, query times ${Q_{({\chi},p')}}(n,p,q)$ and ${Q_{({\chi}',q')}}(n,p,q)$, 
producing a generalized $n$-$p$-$q$-separating collection with size $\zeta(n,p,q)$, $({\chi},p')$-degree 
$\Delta_{({\chi},p')}(n,p,q)$ and $({\chi}',q')$-degree $\Delta_{({\chi}',q')}(n,p,q)$ then there is a construction 
of generalized $n$-$p$-$q$-separating collection with following parameters 
\begin{itemize}%\setlength\itemsep{-.7mm}
\item 
% size 
$\zeta'(n,p,q) \leq |\mathscr{P}_t^n| \cdot 
\sum_{\substack{(p_1,\ldots,p_t) \in {\cal Z}_{s,t}^{p}}} \prod_{i \leq t} \zeta(n,p_i,s-p_i)$,
\item 
% initialization time 
$\tau_I'(n,p,q) = \cO\Big(\big(\sum_{\substack{\hat{p} \leq s}} \tau_I(n,\hat{p},s-\hat{p})\big) + \zeta'(n,p,q) \cdot n^{\cO(1)}\Big)$,
\item 
% $(\chi,p')$-degree 
$\Delta_{(\chi,p')}'(n,p,q) \leq |\mathscr{P}_t^n|\cdot |{\cal Z}_{s,t}^p|\cdot
\max_{\substack{(p_1,\ldots,p_t) \in {\cal Z}_{s,t}^{p} \\ p_1'\leq p_1,\ldots, p_t'\leq p_t \\ p_1'+\ldots+p_t'=p' }} \prod_{i \leq t}  \Delta_{({\chi},p_i')}(n,p_i,s-p_i)$,
\item 
% $(\chi',q')$-degree 
$\Delta'_{(\chi',q')}(n,p,q) \leq |\mathscr{P}_t^n|\cdot |{\cal Z}_{s,t}^p|\cdot
\max_{\substack{(p_1,\ldots,p_t) \in {\cal Z}_{s,t}^{p} \\ q_1'\leq s-p_1,\ldots, q_t'\leq s-p_t \\ q_1'+\ldots+q_t'=q'}} \prod_{i \leq t}  \Delta_{({\chi}',q_i')}(n,p_i,s-p_i)$,
% \item $p_1$-query time $\tau_{Q_{p_1}}'(n,p,q) \leq O\big(\Delta_{p_1}'(n,p,q) \cdot n^{\cO(1)} + |\mathscr{P}_{t}^{n}|\cdot t s^t\cdot \sum_{\substack{\hat{p} \leq s \\ \hat{p}'\leq p_1,\hat{p}}} \tau_{Q_{\hat{p}'}}(n,\hat{p},s-\hat{p}') \big)$.
\item 
% $(\chi,p')$-query time \\
${Q'_{(\chi,p')}}(n,p,q) =\cO \Big(\Delta'_{(\chi,p')}(n,p,q) \cdot n^{\cO(1)} + |\mathscr{P}_{t}^n|\cdot |{\cal Z}_{s,t}^p| \cdot t\cdot \big( \sum_{\substack{ \hat{p}'\leq \hat{p}\leq s \\ \hat{p}-\hat{p}'\leq p-p' \\ s-\hat{p}\leq q }} Q_{(\chi_{\hat{p}},\hat{p}')}(n,\hat{p},s-\hat{p})\big)  \Big)$
\item 
% $(\chi',q')$-query time \\
${Q'_{(\chi',q')}}(n,p,q) =\cO \Big(\Delta'_{(\chi,p')}(n,p,q) \cdot n^{\cO(1)} + |\mathscr{P}_{t}^n|\cdot |{\cal Z}_{s,t}^p| \cdot t\cdot \big( \sum_{\substack{ \hat{q}'\leq \hat{q}\leq s \\ \hat{q}-\hat{q}'\leq q-q' \\ s-\hat{q} \leq p }} Q_{(\chi_{\hat{q}},\hat{q}')}(n,s-\hat{q},\hat{q})\big)  \Big)$
\end{itemize}
\end{lemma}
\begin{proof}
Set $s = \lfloor (\log (p+q))^2 \rfloor$, $t = \lceil \frac{p+q}{s} \rceil$ and $\tilde{q} = st - p$. We will give 
a construction of generalized $n$-$p$-$\tilde{q}$-separating collections with initialization time, query time, size and degree 
within the claimed bounds above. In the construction we will be using the construction with initialization time 
$\tau_I$, query times ${Q_{({\chi},p')}}$ and ${Q_{({\chi}',q')}}$, size $\zeta$, and degrees $\Delta_{({\chi},p')}$ 
and $\Delta_{({\chi}',q')}$ as a black box. Since  $\tilde{q} \geq q$, a $n$-$p$-$\tilde{q}$-separating collection 
is also a $n$-$p$-$q$-separating collection. We may assume without loss of generality that $U = \{1, \ldots, n\}$.
 
Our algorithm runs  for every $0\leq\hat{p}\leq s$, the initialization of  the given construction of  generalized 
$n$-$\hat{p}$-$(s-\hat{p})$-separating collections. We will refer by 
$({\cal F}_{\hat{p}}, {\chi}_{\hat{p}}, {\chi}'_{\hat{p}})$ to the generalized separating collection constructed for 
$\hat{p}$. For each $\hat{p}$ the initialization of the construction outputs the family ${\cal F}_{\hat{p}}$. 
 
% We need to define a few operations on families of sets. For a family ${\cal A}$ over $U$ and subset $U' \subseteq U$ we define 
% $${\cal A} \sqcap U' = \{A \cap U'~:~A \in {\cal A}\}.$$ 
% For two families ${\cal A}$ and  ${\cal B}$ over (subsets of) $U$, we define  
% $${\cal A} \circ {\cal B} = \{A \cup B ~:~A \in {\cal A} \wedge B \in {\cal B}\}.$$ 
We need to define a few operations on families of sets. For families of sets  ${\cal A}$, ${\cal B}$ over $U$ 
and subset $U' \subseteq U$ we define
\begin{eqnarray*}
{\cal A} \sqcap U' &=& \{A \cap U'~:~A \in {\cal A}\} \\ 
{\cal A} \circ {\cal B} &=& \{A \cup B ~:~A \in {\cal A} \wedge B \in {\cal B}\} 
\end{eqnarray*}
We now define ${\cal F}$ as follows.
\begin{align}\label{eqn:defineFSplit} 
{\cal F} = \bigcup_{\substack{\{U_1,\ldots,U_{t}\} \in \mathscr{P}_{t}^n \\ (p_1,\ldots,p_t) \in {\cal Z}_{s,t}^{p}}} 
(\hat{\cal F}_{p_1} \sqcap U_1) \circ (\hat{\cal F}_{p_2} \sqcap U_2) \circ \ldots \circ (\hat{\cal F}_{p_t} \sqcap U_{t}) 
\end{align}
It follows directly from the definition of  ${\cal F}$ that $|{\cal F}|$ is within the claimed bound for 
$\zeta'(n,p,q)$. For the initialization time, the algorithm spends 
$\cO\left(\sum_{\substack{\hat{p} \leq s}} \tau_I(n,\hat{p},s-\hat{p})\right)$ time to initialize the constructions 
of the generalized $n$-$\hat{p}$-$(s-\hat{p})$-separating collections for all $\hat{p}\leq s$ together. Now the 
algorithm can output the entries of ${\cal F}$ one set at a time by using~\eqref{eqn:defineFSplit}, spending 
$n^{\cO(1)}$ time per output set. Hence the time bound for $\tau'_I(n,p,q)$ follows. 

For every set $A \in \bigcup_{p'\leq p}{U \choose p'}$ we define $\chi(A)$ as follows.
\begin{align}\label{eqn:defineChiSplit} 
\chi(A) = \bigcup_{\substack{\{U_1,\ldots,U_{t}\} \in \mathscr{P}_{t}^n\\(p_1,\ldots,p_t) \in {\cal Z}_{s,t}^{p}~\mbox{\scriptsize such that}\\ \forall U_i~:~|U_i \cap A| \leq p_i}} 
\Big[({\chi}_{p_1}(A \cap U_1) \sqcap U_1) \circ ({\chi}_{p_2}(A \cap U_2) \sqcap U_2) \circ \ldots \\
\nonumber ... \circ ({\chi}_{p_t}(A \cap U_t) \sqcap U_t)\Big] 
\end{align}
Now we show that $\chi(A)\subseteq {\cal F}$. From the definition of generalized $n$-${p_i}$-$(s-{p_i})$-separating 
collections $(\hat{\cal F}_{{p_i}},\chi_{{p_i}},\chi'_{{p_i}})$, each family $\chi_{p_i}(A\cap U_i)$ in~
\eqref{eqn:defineChiSplit} is a subset of $\hat{\cal F}_{p_i}$. This implies that 
$\chi_{p_i}(A\cap U_i)\sqcap U_i\subseteq \hat{\cal F}_{p_i} \sqcap U_i$. Hence $\chi(A)\subseteq {\cal F}$.   
Similar way we can define $\chi'(B)$ for any $B\in\bigcup_{q'\leq q}{U\choose q'}$ as
\begin{align}\label{eqn:defineChiprimeSplit} 
\chi'(B) = \bigcup_{\substack{\{U_1,\ldots,U_{t}\} \in \mathscr{P}_{t}^n\\(p_1,\ldots,p_t) \in {\cal Z}_{s,t}^{p}~\mbox{\scriptsize such that} \\ \forall U_i~:~|U_i \cap B| \leq s-p_i }} 
\Big[({\chi}'_{p_1}(B \cap U_1) \sqcap U_1) \circ ({\chi}'_{p_2}(B \cap U_2) \sqcap U_2) \circ \ldots \\
\nonumber ... \circ ({\chi}'_{p_t}(B \cap U_t) \sqcap U_t)\Big] 
\end{align}
Similar to the proof of $\chi(A)\subseteq {\cal F}$, we can show that $\chi'(B)\subseteq {\cal F}$. 
It follows directly from the definition of  $\chi(A)$ and $\chi'(B)$ that $|\chi(A)|$ and $|\chi'(B)|$ is within the 
claimed bound for $\Delta_{(\chi,p')}'(n,p,q)$ and $\Delta_{(\chi',q')}'(n,p,q)$ respectively. We now describe how 
queries $\chi(A)$ can be answered, and analyze how much time it takes. Given $A$ we will compute $\chi(A)$ using  ~\eqref{eqn:defineChiSplit}. Let $|A|=p'$. 
For each $\{U_1,\ldots,U_{t}\} \in \mathscr{P}_{t}^n$ and $(p_1,\ldots,p_t) \in {\cal Z}_{s,t}^{p}$ such that $p_i' = |U_i \cap A| \leq p_i$ for all $i \leq t$, 
we proceed as follows. First we compute ${\chi}_{p_i}(A \cap U_i)$ for each $i \leq t$, spending in total 
$\cO(\sum_{i \leq t} Q_{(\chi_{p_i},p_i')}(n,p_i,s-p_i))$ time. Now we add each set in 
$({\chi}_{p_1}(A \cap U_1) \sqcap U_1) \circ ({\chi}_{p_2}(A \cap U_2) \sqcap U_2) \circ \ldots \circ ({\chi}_{p_t}(A \cap U_t) \sqcap U_t)$ 
to $\chi(A)$, spending $n^{\cO(1)}$ time per set that is added to $\chi(A)$, yielding the bound below, 
% for $\tau_{Q_{p_j}}'(n,p,q)$ for $j\in\{1,2\}$. Similarly we can bound $\tau'_{Q_q}$
\begin{align*}
Q'_{(\chi,p')}(n,p,q) \leq \cO\Big(\Delta'_{(\chi,p')}(n,p,q) \cdot n^{\cO(1)} + \sum_{\substack{\{U_1,\ldots,U_{t}\} \in \mathscr{P}_{t}\\(p_1,\ldots,p_t) \in {\cal Z}_{s,t}^{p}~\mbox{\scriptsize such that}\\ \forall U_i~:~p_i' = |U_i \cap A| \leq p_i}} \big[\sum_{i \leq t} Q_{\chi_{p_i},p_i'}(n,p_i,s-p_i)\big]\Big) \\
\leq \cO\Big(\Delta'_{(\chi,p')}(n,p,q) \cdot n^{\cO(1)} + |\mathscr{P}_{t}^n|\cdot |{\cal Z}_{s,t}^p| \cdot 
\max_{\substack{(p_1,\ldots,p_t) \in {\cal Z}_{s,t}^p\\ p_1'\leq p_1,\cdots,p_t'\leq p_t~\mbox{\scriptsize such that}\\p_1'+\cdots +p_t'=p'}} \big( \sum_{i \leq t} Q_{(\chi_{p_i},p_i')}(n,p_i,s-p_i)\big)  \Big) \\
\leq \cO\Big(\Delta'_{(\chi,p')}(n,p,q) \cdot n^{\cO(1)} + |\mathscr{P}_{t}^n|\cdot |{\cal Z}_{s,t}^p| \cdot t\cdot 
\big( \sum_{\substack{ \hat{p}'\leq \hat{p}\leq s \\ \hat{p}-\hat{p}'\leq p-p' \\ s-\hat{p}\leq q }} Q_{(\chi_{\hat{p}},\hat{p}')}(n,\hat{p},s-\hat{p})\big)  \Big) \\
% \leq \cO\Big(\Delta'(n,p,q) \cdot n^{\cO(1)} + |\mathscr{P}_{t}|\cdot t \cdot \sum_{\hat{p} \leq s} \tau_Q(n,\hat{p},s-\hat{p}) \Big)
\end{align*}
By doing similar analysis, we get required bound for $Q'_{(\chi',q')}$.  
We now need to argue that $({\cal F}, \chi,\chi')$ is in fact a generalized  $n$-$p$-$\tilde{q}$-separating collection. For any 
$r$, consider pairwise disjoint sets $A_1\in {U \choose p_1},\ldots,A_r\in {U \choose p_r}$ and 
$B \in {U \choose \tilde{q}}$ such that $p_1+\cdots+p_r=p$. Let $A=A_1\cup\cdots\cup A_r$. There exists a 
consecutive partition $\{U_1, \ldots, U_t\} \in \mathscr{P}_t^n$ of $U$ such that for every $i \leq t$ we have that 
$|(A \cup B) \cap U_i| = \frac{p+\tilde{q}}{t} = s$. For each $i \leq t$ set $p_i =  |A \cap U_i|$ and 
$q_i = |B \cap U_i| = s - p_i$. For every $i \leq t$ the tuple $({\cal F}_{p_i}, {\chi}_{p_i},\chi'_{p_i})$ form 
a $n$-$p_i$-$q_i$-separating collection. Hence 
$\exists F_i \in \chi_{p_i}(A_1\cap U_i)\cap\ldots\cap\chi_{p_i}(A_r\cap U_i) \cap \chi'_{p_i}(B\cap U_i)$ 
because $|A_1\cap U_i|+\ldots+|A_r\cap U_i|=p_i$, $|B\cap U_i|=q_i$ and 
$({\cal F}_{p_i}, {\chi}_{p_i},\chi'_{p_i})$ is a $n$-$p_i$-$q_i$-separating collection. That is 
$F_i \in \chi_{p_i}(A_j\cap U_i)$ for all $j\leq r$ and $F_i \in \chi'_{p_i}(B\cap U_i)$. Let 
$F=\bigcup_{i\leq t}F_i\cap U_i$. By construction of $\chi$ and $\chi'$, $F \in \chi(A_j)$ for all $j\leq r$  and 
$F \in \chi'(B)$. Hence $F \in \chi(A_1)\cap\ldots\cap \chi(A_r)\cap\chi'(B)$. This completes the proof
\end{proof}

Now we are ready to prove the Lemma~\ref{lem:twin_sep_coll_construction}. We restate the lemma for easiness of presentation. 

\medskip 

%\begin{lemma}
%\label{lem:twin_sep_coll_construction}
%Given a constant $x$ such that $0<x<1$, there is a construction of $n$-$p$-$q$- separating collection with the following parameters
%\begin{itemize} %\setlength\itemsep{-.7mm}
%\item size $\zeta(n,p,q) \leq 2^{O(\frac{p+q}{\log\log\log(p+q)})}\cdot \frac{1}{x^p(1-x)^q}\cdot (p+q)^{O(1)} \cdot \log n$
%\item initialization time $\tau_I(n,p,q) \leq  2^{O(\frac{p+q}{\log\log\log(p+q)})}\cdot \frac{1}{x^p(1-x)^q}\cdot (p+q)^{O(1)} \cdot n\log n$
%\item $(\chi,p')$-degree $\Delta_{(\chi,p')}(n,p,q) \leq  2^{O(\frac{p+q}{\log\log\log(p+q)})}\cdot \frac{1}{x^{p-p'}(1-x)^q}\cdot (p+q)^{O(1)} \cdot \log n$
%\item $(\chi,p')$-query time $Q_{(\chi,p')}(n,p,q) \leq  2^{O(\frac{p+q}{\log\log\log(p+q)})}\cdot \frac{1}{x^{p-p'}(1-x)^q}\cdot (p+q)^{O(1)} \cdot \log n$
%\item $(\chi',q')$-degree $\Delta_{(\chi',q')}(n,p,q) \leq  2^{O(\frac{p+q}{\log\log\log(p+q)})}\cdot \frac{1}{x^{p}(1-x)^{q-q'}}\cdot (p+q)^{O(1)} \cdot \log n$
%\item $(\chi',q')$-query time $Q_{(\chi',q')}(p,q,n) \leq  2^{O(\frac{p+q}{\log\log\log(p+q)})}\cdot \frac{1}{x^{p}(1-x)^{q-q'}}\cdot (p+q)^{O(1)} \cdot \log n$
%\end{itemize}
%\end{lemma}

%\begin{lemma}
%\label{lem:twin_sep_coll_construction}
{{\bf Lemma~\ref{lem:twin_sep_coll_construction}} \em 
Given a constant $x$ such that $0<x<1$, there is a construction of generalized $n$-$p$-$q$- separating collection with the following parameters
\begin{itemize} %\setlength\itemsep{-.7mm}
\item size, $\zeta(n,p,q) \leq 2^{\cO(\frac{p+q}{\log\log\log(p+q)})}\cdot \frac{1}{x^p(1-x)^q}\cdot (p+q)^{O(1)} \cdot \log n$
\item initialization time, $\tau_I(n,p,q) \leq  2^{O(\frac{p+q}{\log\log\log(p+q)})}\cdot \frac{1}{x^p(1-x)^q}\cdot (p+q)^{O(1)} \cdot n\log n$
\item $(\chi,p')$-degree, $\Delta_{(\chi,p')}(n,p,q) \leq  2^{O(\frac{p+q}{\log\log\log(p+q)})}\cdot \frac{1}{x^{p-p'}(1-x)^q}\cdot (p+q)^{O(1)} \cdot \log n$
\item $(\chi,p')$-query time, $Q_{(\chi,p')}(n,p,q) \leq  2^{O(\frac{p+q}{\log\log\log(p+q)})}\cdot \frac{1}{x^{p-p'}(1-x)^q}\cdot (p+q)^{O(1)} \cdot \log n$
\item $(\chi',q')$-degree, $\Delta_{(\chi',q')}(n,p,q) \leq  2^{O(\frac{p+q}{\log\log\log(p+q)})}\cdot \frac{1}{x^{p}(1-x)^{q-q'}}\cdot (p+q)^{O(1)} \cdot \log n$
\item $(\chi',q')$-query time, $Q_{(\chi',q')}(n,p,q) \leq  2^{O(\frac{p+q}{\log\log\log(p+q)})}\cdot \frac{1}{x^{p}(1-x)^{q-q'}}\cdot (p+q)^{O(1)} \cdot \log n$
\end{itemize}}

%\end{lemma}

\begin{proof}
%  The structure of the proof is simple; first we create a construction using Lemma~\ref{lem:slowBalancedUniversal}. Applying Lemma~\ref{lem:reduceUniverse} gives a new, second, construction. 
% We then make more constructions by repeatedly applying Lemmata~\ref{lem:splitSolution} and~\ref{lem:reduceUniverse}. Specifically the third construction is obtained by obtaining Lemma~\ref{lem:splitSolution} to the second, the fourth by applying Lemma~\ref{lem:reduceUniverse} to the third, the fifth by applying Lemma~\ref{lem:splitSolution} to the fourth and the sixth and final construction is obtained by applying Lemma~\ref{lem:reduceUniverse} to the fifth. The bulk of the work is to verify that the respective constructions indeed have the claimed parameters. We now proceed with the formal proof.
The structure of the proof is as follows. We first create a collection using Lemma~\ref{lem:twin_sep_coll_construction}. 
Then we apply Lemma~\ref{lem:twinreduceUniverse} and obtain another construction. From here onwards we keep applying Lemma~\ref{lem:splitSolution} and Lemma~\ref{lem:twinreduceUniverse} in phases until we achieve the 
required bounds on size, degree, query and intializitaion time. 

%\paragraph{\em Brute Force Construction and Universe Reduction.}
We first apply Lemma~\ref{lem:twin_sep_coll_brute_force} and get a construction of $n$-$p$-$q$-twin separating collections with the following parameters.
\begin{itemize}\setlength\itemsep{-.7mm}
\item size, $\zeta^1(n,p,q) =\cO \left(\frac{1}{x^{p}(1-x)^q} \cdot (p^2+q^2+1)\log n \right)$, 
% \cO({p+q \choose q} \cdot (p+q)^{\cO(1)} \cdot \log n)$, 
\item initialization time, $\tau_I^1(n,p,q) = \cO({2^n \choose \zeta(n,p,q)} \cdot \frac{1}{x^p(1-x)^q} \cdot n^{\cO(p+q)})$,
\item $(\chi,p')$-degree for $p'\leq p$, $\Delta^1_{(\chi,p')}(n,p,q) = \cO\left(\frac{1}{x^{p-p'}}\cdot\frac{(p^2+q^2+1)}{(1-x)^q} \cdot \log n\right)$
% \item $p_2$-degree $\Delta_{p_2}(n,p,q) = \cO\left(\frac{1}{x^{p_2}}\cdot\frac{(p+q+1)}{(1-x)^q} \cdot \log n\right)$, and 
\item $(\chi,p')$-query time ${Q^1_{(\chi,p')}}(n,p,q) = \cO(\frac{1}{x^{p}(1-x)^q} \cdot n^{\cO(1)})=\cO(2^n n^{\cO(1)})$
\item $(\chi',q')$-degree for $q'\leq q$, $\Delta^1_{(\chi',q')}(n,p,q)=\cO\left(\frac{1}{x^{p}(1-x)^{q-q'}}\cdot(p^2+q^2+1) \cdot \log n\right)$
\item $(\chi',q')$-query time, ${Q^1_{(\chi',q')}}(n,p,q)=\cO(\frac{1}{x^{p}(1-x)^q} \cdot n^{\cO(1)})=\cO(2^n n^{\cO(1)})$
\end{itemize}
We apply Lemma~\ref{lem:twinreduceUniverse} to this construction to get a new construction with the following parameter.
\begin{itemize} %\setlength\itemsep{-.7mm}
\item size, $\zeta^2(n,p,q)=\cO\left(\frac{1}{x^{p}(1-x)^q} \cdot  (p+q)^{\cO(1)} \cdot \log n\right)$ 
\item initialization time, 
\begin{eqnarray*}
\tau_I^2(n,p,q) &=& \cO\left(\tau_I^1\left((p+q)^2,p,q\right) + \zeta^1\left((p+q)^2,p,q\right) \cdot (p+q)^{\cO(1)} \cdot n \log n\right)\\
&=& \cO\left(
% {2^{(p+q)^2} \choose \zeta^1((p+q)^2,p,q)} \cdot 
\frac{2^{2^{(p+q)^2}}}{x^p(1-x)^q} \cdot (p+q)^{\cO(p+q)} + \left(\frac{1}{x^{p}(1-x)^q} \cdot (p+q)^{\cO(1)} \cdot n \log n \right)\right)\\
&=&\cO\left( \frac{(p+q)^{\cO(p+q)}}{x^{p}(1-x)^q}\left({2^{2^{(p+q)^2}}}+n \log n\right)\right)
\end{eqnarray*}
% \item $p_1$-degree $\Delta'_{p_1}(n,p,q) \leq \Delta_{p_1}\left((p+q)^2,p,q\right) \cdot  (p+q)^{\cO(1)} \cdot \log n$,
\item $(\chi,p')$-degree, $\Delta_{(\chi,p')}^2(n,p,q) =\cO\left( \frac{1}{x^{p-p'}{(1-x)^q}} \cdot  (p+q)^{\cO(1)} \cdot \log n\right)$
\item $(\chi,p')$-query time, $Q_{(\chi,p')}^2(n,p,q)=\cO\left(\left(2^{(p+q)^2} + \frac{1}{x^{p-p'}{(1-x)^q}}\right)   (p+q)^{\cO(1)} \cdot \log n\right)$
\item $(\chi',q')$-degree, $\Delta_{(\chi',q')}^2(n,p,q) =\cO\left( \frac{1}{x^{p}(1-x)^{q-q'}} \cdot  (p+q)^{\cO(1)} \cdot \log n\right)$
\item $(\chi,q')$-query time, $Q_{(\chi',q')}^2(n,p,q)=\cO\left(\left(2^{(p+q)^2} + \frac{1}{x^{p}(1-x)^{q-q'}}\right)   (p+q)^{\cO(1)} \cdot \log n\right)$
% \begin{eqnarray*}
% \tau_{Q_{p_j}}^2(n,p,q) &\leq& \cO\left(\left(\tau_{Q_{p_j}}^1\left((p+q)^2,p,q\right) + \Delta_{p_j}^1\left((p+q)^2,p,q\right) \right) \cdot (p+q)^{\cO(1)} \cdot \log n\right)\\
% &\leq& \cO^*\left(\tau_{Q_{p_j}}^1\left((p+q)^2,p,q\right) + \Delta_{p_j}^1\left((p+q)^2,p,q\right) \right)
% \end{eqnarray*}
\end{itemize}
We apply Lemma~\ref{lem:splitSolution} to this construction. Recall that in Lemma~\ref{lem:splitSolution} 
we set $s=\lfloor(\log (p+q))^2 \rfloor$ and $t = \lceil \frac{p+q}{s} \rceil$.
\begin{eqnarray*}
 \zeta^3(n,p,q) &\leq& |\mathscr{P}_t^{n}| \cdot 
\sum_{(p_1,\ldots,p_t) \in {\cal Z}_{s,t}^{p}} \prod_{i \leq t} \zeta^2(n,p_i,s-p_i)\\
&\leq& n^{\cO(t)}\cdot |{\cal Z}_{s,t}^{p}| \cdot \max_{(p_1,\ldots,p_t) \in {\cal Z}_{s,t}^{p}} \prod_{i \leq t} \zeta^2(n,p_i,s-p_i)\\
&\leq& n^{\cO(t)}\cdot (p+q)^{\cO(t)}\cdot \frac{1}{x^{p}(1-x)^{q+s}}\cdot s^{\cO(t)}\cdot (\log n)^{\cO(t)}\\
&\leq& n^{\cO(\frac{p+q}{\log^2(p+q)})}\cdot \frac{1}{x^{p}(1-x)^{q}} \qquad\qquad\quad \left(\mbox{Because} \left(\frac{1}{1-x}\right)^s\in n^{\cO(t)}.\right)\\
\end{eqnarray*}
\begin{eqnarray*}
%%%%%%%%%% 
 \tau_I^3(n,p,q) &=&\cO\left(\left(\sum_{\hat{p} \leq s} \tau_I^2(n,\hat{p},s-\hat{p})\right) + \zeta^3(n,p,q) \cdot n^{\cO(1)}\right)\\
&=&\cO\left(\left(\sum_{\hat{p} \leq s} \frac{s^{\cO(s)}}{x^{\hat{p}}(1-x)^{s-\hat{p}}}\left({2^{2^{s^2}}}+n \log n\right)\right) + \zeta^3(n,p,q) \cdot n^{\cO(1)}\right)\\
&=& \cO\left(\frac{(\log(p+q))^{\cO(\log^2(p+q))}}{x^{p}(1-x)^{q}}\left({2^{2^{\log^4(p+q)}}}+n \log n\right) + n^{\cO(\frac{p+q}{\log^2(p+q)})}\cdot \frac{1}{x^{p}(1-x)^{q}}\right)\\
% &=& O\left(\frac{2^{\cO(\frac{p+q}{\log(p+q)})}}{x^{p}(1-x)^{q}}\left({2^{2^{(p+q)^2}}}+n \log n + n^{O(\frac{p+q}{\log^2(p+q)})}\right)\right)\\
% &&(\mbox{ Because } s^{O(s)}\leq 2^{\cO(\frac{p+q}{\log(p+q)})})\\\\
\end{eqnarray*}
\begin{eqnarray*}
\Delta_{(\chi,p')}^3(n,p,q) &\leq& |\mathscr{P}_t^n|\cdot |{\cal Z}_{s,t}^{p}| \cdot 
\max_{\substack{(p_1,\ldots,p_t) \in {\cal Z}_{s,t}^{p} \\ p_1'\leq p_1,\ldots, p_t'\leq p_t \\ p_1'+\ldots+p_t'=p'}} \prod_{i \leq t}  \Delta_{(\chi,p')}^2(n,p_i,s-p_i)\\
&\leq& n^{\cO(t)}\cdot (p+q)^{\cO(t)} \cdot\frac{1}{x^{p-p'}(1-x)^{q+s}}\cdot s^{\cO(t)}\cdot (\log n)^{\cO(t)}\\
&\leq& n^{\cO(\frac{p+q}{\log^2(p+q)})}\cdot \frac{1}{x^{p-p'}(1-x)^{q}}\\
\end{eqnarray*}
\begin{eqnarray*}
\Delta_{(\chi',q')}^3(n,p,q) &\leq& |\mathscr{P}_t^n|\cdot |{\cal Z}_{s,t}^{p}| \cdot 
\max_{\substack{(p_1,\ldots,p_t) \in {\cal Z}_{s,t}^{p} \\ q_1'\leq s-p_1,\ldots, q_t'\leq s-q_t \\ q_1'+\ldots+q_t'=q'}} \prod_{i \leq t}  \Delta_{(\chi',q_i')}^2(n,p_i,s-p_i)\\
&\leq& n^{\cO(t)}\cdot (p+q)^{\cO(t)}\cdot \frac{1}{x^{p}(1-x)^{q+s-q'}}\cdot s^{\cO(t)}\cdot (\log n)^{\cO(t)}\\
&\leq& n^{\cO(\frac{p+q}{\log^2(p+q)})}\cdot \frac{1}{x^{p}(1-x)^{q-q'}}\\
% \cdot (\log(p+q))^{O(1)}\cdot \log n\\
\end{eqnarray*}
\begin{eqnarray*}
Q_{(\chi,p')}^3(n,p,q) &\leq& \cO\left(\Delta_{(\chi,p')}^3(n,p,q) \cdot n^{\cO(1)} + 
|\mathscr{P}_{t}^{n}|\cdot |{\cal Z}_{s,t}^{p}| \cdot t \cdot 
\sum_{\substack{ \hat{p}'\leq \hat{p}\leq s \\ \hat{p}-\hat{p}'\leq p-p' \\ s-\hat{p}\leq q }} Q_{(\chi,\hat{p}')}^2(n,\hat{p},s-\hat{p}) \right)\\
 &\leq& \cO\left(\Delta_{(\chi,p')}^3(n,p,q) \cdot n^{\cO(1)} + n^{\cO(t)} \cdot 
\sum_{\substack{ \hat{p}'\leq \hat{p}\leq s \\ \hat{p}-\hat{p}'\leq p-p' \\ s-\hat{p}\leq q }} \left(2^{s^2}+\frac{1}{x^{\hat{p}-\hat{p}'}(1-x)^{s-\hat{p}}}\right)s^{\cO(1)} \log n \right)\\
 &\leq& \cO\left(\frac{n^{\cO(\frac{p+q}{\log^2(p+q)})}}{x^{p-p'}(1-x)^{q}} + n^{\cO(t)} \cdot s^{\cO(1)}\cdot \log n
\left(2^{s^2}+\frac{1}{x^{p-p'}(1-x)^{q}}\right)  \right)\\
 &\leq& \cO\left(\frac{n^{\cO(\frac{p+q}{\log^2(p+q)})}}{x^{p-p'}(1-x)^{q}}  \right)\\
% Q_{(\chi',q')}^3(n,p,q) &\leq& \cO\left(\frac{n^{\cO(\frac{p+q}{\log^2(p+q)})}}{x^{p}(1-x)^{q-q'}}  \right)
\end{eqnarray*}
Similar way we can bound $Q_{(\chi',q')}^3$ as,
\begin{eqnarray*}
 Q_{(\chi',q')}^3(n,p,q) &\leq& \cO\left(\frac{n^{\cO(\frac{p+q}{\log^2(p+q)})}}{x^{p}(1-x)^{q-q'}}  \right)
\end{eqnarray*}

%
% ------------------------------------------------------------------------------\\
%
We apply Lemma~\ref{lem:twinreduceUniverse} to this construction to get a new construction with the following parameters.
\begin{itemize} %\setlength\itemsep{-.7mm}
\item  size, 
$\zeta^4(n,p,q) \leq 2^{\cO(\frac{p+q}{\log(p+q)})}\cdot \frac{1}{x^{p}(1-x)^{q}} \cdot  (p+q)^{\cO(1)} \cdot \log n$,
 \item  initialization time, 
\begin{eqnarray*}
\tau_I^4(n,p,q) &\leq& \cO\left(\tau_I^3\left((p+q)^2,p,q\right) + \zeta^3\left((p+q)^2,p,q\right) \cdot (p+q)^{\cO(1)} \cdot n \log n\right)\\
&\leq& 2^{2^{\log^4(p+q)}}\cdot\frac{(\log(p+q))^{\cO(\log^2(p+q))}}{x^p(1-x)^q} + \frac{2^{\cO(\frac{p+q}{\log (p+q)})}}{x^p(1-x)^q} \cdot (p+q)^{\cO(1)}n\log n\\
% &\leq& \frac{2^{O(\frac{p+q}{\log (p+q)})}}{x^p(1-x)^q} \left( 2^{2^{(p+q)^2}} +  (p+q)^{O(1)}n\log n\right)  
\end{eqnarray*}
\item $(\chi,p')$-degree, 
\begin{eqnarray*}
\Delta_{(\chi,p')}^4(n,p,q) &\leq& \Delta_{(\chi,p')}^3\left((p+q)^2,p,q\right) \cdot  (p+q)^{\cO(1)} \cdot \log n \\
&\leq& \frac{2^{\cO(\frac{p+q}{\log (p+q)})}}{x^{p-p'}(1-x)^q} \cdot  (p+q)^{\cO(1)} \cdot \log n
\end{eqnarray*}
\item $(\chi',q')$-degree, 
\begin{eqnarray*}
\Delta_{(\chi',q')}^4(n,p,q) &\leq& \Delta_{(\chi',q')}^3\left((p+q)^2,p,q\right) \cdot  (p+q)^{\cO(1)} \cdot \log n \\
&\leq& \frac{2^{\cO(\frac{p+q}{\log (p+q)})}}{x^{p}(1-x)^{q-q'}} \cdot  (p+q)^{\cO(1)} \cdot \log n
\end{eqnarray*}
\item $(\chi,p')$-query time, 
\begin{eqnarray*}
Q_{(\chi,p')}^4(n,p,q) &\leq& \cO\left(\left( Q_{(\chi,p')}^3\left((p+q)^2,p,q\right) + \Delta_{(\chi,p')}^3\left((p+q)^2,p,q\right) \right) \cdot (p+q)^{\cO(1)} \cdot \log n\right)\\
&\leq& \frac{2^{\cO(\frac{p+q}{\log (p+q)})}}{x^{p-p'}(1-x)^q}\cdot (p+q)^{\cO(1)}\log n
\end{eqnarray*}
\item $(\chi',q')$-query time, 
\begin{eqnarray*}
Q_{(\chi',q')}^4(n,p,q) 
% &\leq& \cO\left(\left(\tau_{Q_{p_j}}^3\left((p+q)^2,p,q\right) + \Delta_{p_j}^3\left((p+q)^2,p,q\right) \right) \cdot (p+q)^{\cO(1)} \cdot \log n\right)\\
&\leq& \frac{2^{\cO(\frac{p+q}{\log (p+q)})}}{x^{p}(1-x)^{q-q'}}\cdot (p+q)^{\cO(1)}\log n 
\end{eqnarray*}
\end{itemize}
We apply Lemma~\ref{lem:splitSolution} to this construction by setting $s=\lfloor(\log (p+q))^2 \rfloor$ and 
$t = \lceil \frac{p+q}{s} \rceil$.
\begin{itemize}%\setlength\itemsep{-.7mm}
\item size,
\begin{eqnarray*}
\zeta^5(n,p,q) &\leq& |\mathscr{P}_t^n| \cdot 
\sum_{(p_1,\ldots,p_t) \in {\cal Z}_{s,t}^{p} } \prod_{i \leq t} \zeta^4(n,p_i,s-p_i)\\
&\leq& n^{\cO(t)}\cdot(p+q)^{\cO(t)}\cdot s^{\cO(t)} \cdot 2^{\cO(\frac{st}{\log s})}\cdot (\log n)^{\cO(t)}\cdot \frac{1}{x^p(1-x)^{q+s}}\\ 
&\leq& n^{\cO(\frac{p+q}{\log^2(p+q)})} \cdot 2^{\cO(\frac{p+q}{\log\log(p+q)})} \frac{1}{x^p(1-x)^q}
\end{eqnarray*}
\item initialization time,
\begin{eqnarray*}
\tau_I^5(n,p,q) &\leq& \cO\left(\left(\sum_{\hat{p} \leq s} \tau_I^4(n,\hat{p},s-\hat{p})\right) + \zeta^5(n,p,q) \cdot n^{\cO(1)}\right)\\
&\leq&\cO\left(s\frac{2^{2^{\log^4s}}\cdot (\log s)^{\cO(\log^2s)}}{x^p(1-x)^q}+ \frac{2^{\cO(\frac{s}{\log s})}}{x^p(1-x)^q}\cdot n\log n + n^{\cO(\frac{p+q}{\log^2(p+q)})} \cdot  \frac{2^{\cO(\frac{p+q}{\log\log(p+q)})}}{x^p(1-x)^q}
\right)\\
&\leq&\cO\left(s\frac{2^{2^{\log^4s}}\cdot (\log s)^{\cO(\log^2s)}}{x^p(1-x)^q} + n^{\cO(\frac{p+q}{\log^2(p+q)})} \cdot  \frac{2^{\cO(\frac{p+q}{\log\log(p+q)})}}{x^p(1-x)^q}
\right)\\
&\leq&\cO\left(\frac{2^{2^{\log^4s}}\cdot (s)^{\cO(s)}}{x^p(1-x)^q} + n^{\cO(\frac{p+q}{\log^2(p+q)})} \cdot  \frac{2^{\cO(\frac{p+q}{\log\log(p+q)})}}{x^p(1-x)^q}
\right)\\
&\leq&\cO\left(\frac{2^{2^{(2\log\log(p+q))^4}}\cdot (\log(p+q))^{\cO((\log(p+q))^2)}} {x^p(1-x)^q} 
+ n^{\cO(\frac{p+q}{\log^2(p+q)})} \cdot  \frac{2^{\cO(\frac{p+q}{\log\log(p+q)})}}{x^p(1-x)^q}
\right)
% &\leq& \cO\left( \frac{2^{2^{(\log^4(p+q))}}}{x^p(1-x)^q} + n^{\cO(\frac{p+q}{\log^2(p+q)})} \cdot 2^{O(\frac{p+q}{\log\log(p+q)})} \frac{1}{x^p(1-x)^q}
% \right)
\end{eqnarray*}
 \item $(\chi,p')$-degree, 
\begin{eqnarray*}
\Delta_{(\chi,p')}^5(n,p,q) &\leq& |\mathscr{P}_t^n|\cdot |{\cal Z}_{s,t}^{p}| \cdot  
\max_{\substack{(p_1,\ldots,p_t) \in {\cal Z}_{s,t}^{p} \\ p_1'\leq p_1,\ldots, p_t'\leq p_t \\ p_1'+\ldots+p_t'=p' }} \prod_{i \leq t}  \Delta_{(\chi,p_i')}^4(n,p_i,s-p_i)\\
&\leq& n^{\cO(t)}\cdot (p+q)^{\cO(t)} \cdot \frac{2^{\cO(\frac{st}{\log s})}}{x^{p-p'}(1-x)^{q+s}}\cdot s^{\cO(t)}\cdot (\log n)^{\cO(t)}\\
&\leq& n^{\cO(\frac{p+q}{\log^2(p+q)})} \cdot 2^{\cO(\frac{p+q}{\log\log(p+q)})} \cdot \frac{1}{x^{p-p'}(1-x)^q}
\end{eqnarray*}
\item $(\chi',q')$-degree, 
\begin{eqnarray*}
\Delta_{(\chi',q')}^5(n,p,q) 
% &\leq& |\mathscr{P}_t^n|\cdot \max_{\substack{(r_1,\ldots,r_t) \in {\cal Z}_{s,t}^{p} \\ r_1''\leq r_1,\ldots, r_t''\leq r_t \\ r_1''+\ldots+r_t''=p_2  }} \prod_{i \leq t}  \Delta_{s-r_i}^4(n,r_i,s-r_i)\\
&\leq& n^{\cO(\frac{p+q}{\log^2(p+q)})} \cdot 2^{\cO(\frac{p+q}{\log\log(p+q)})} \cdot \frac{1}{x^{p}(1-x)^{q-q'}}
\end{eqnarray*}
% \item $p_1$-query time $\tau_{Q_{p_1}}'(n,p,q) \leq O\big(\Delta_{p_1}'(n,p,q) \cdot n^{\cO(1)} + |\mathscr{P}_{t}^{n}|\cdot t s^t\cdot \sum_{\substack{\hat{p} \leq s \\ \hat{p}'\leq p_1,\hat{p}}} \tau_{Q_{\hat{p}'}}(n,\hat{p},s-\hat{p}') \big)$.
\item $(\chi,p')$-query time, 
\begin{eqnarray*}
Q_{(\chi,p')}^5(n,p,q) &\leq& \cO\left(\Delta_{(\chi,p')}^5(n,p,q) \cdot n^{\cO(1)} + |\mathscr{P}_{t}^{n}|\cdot |{\cal Z}_{s,t}^{p}| \cdot   
\max_{\substack{\hat{p}'\leq \hat{p} \leq s}} Q_{(\chi,\hat{p}')}^4(n,\hat{p},s-\hat{p}) \right)\\
&\leq& n^{\cO(\frac{p+q}{\log^2(p+q)})} \cdot 2^{\cO(\frac{p+q}{\log\log(p+q)})} \cdot \frac{1}{x^{p-p'}(1-x)^q}
\end{eqnarray*}
\item $(\chi',q')$-query time,
\begin{eqnarray*}
 Q_{(\chi',q')}^5(n,p,q) 
% &\leq& \cO\left(\Delta_{(\chi',q')}^5(n,p,q) \cdot n^{\cO(1)} + |\mathscr{P}_{t}^{n}|\cdot t s^t\cdot \max_{\substack{\hat{p} \leq s}} \tau_{Q_{s-\hat{p}}}^4(n,\hat{p},s-\hat{p}) \big)\\
&\leq& n^{\cO(\frac{p+q}{\log^2(p+q)})} \cdot 2^{\cO(\frac{p+q}{\log\log(p+q)})} \cdot \frac{1}{x^{p}(1-x)^{q-q'}}
\end{eqnarray*}
\end{itemize}
We apply Lemma~\ref{lem:twinreduceUniverse} to this construction to get a new construction with the following parameters.
\begin{itemize} %\setlength\itemsep{-.7mm}
\item size, 
\begin{eqnarray*}
\zeta^6(n,p,q) &\leq& \zeta^5\left((p+q)^2,p,q\right) \cdot  (p+q)^{\cO(1)} \cdot \log n\\
&\leq & 2^{\cO(\frac{p+q}{\log\log(p+q)})}\cdot \frac{(p+q)^{\cO(1)}}{x^p(1-x)^q}\cdot \log n
\end{eqnarray*}
\item initialization time, 
\begin{eqnarray*}
\tau_I^6(n,p,q) &\leq& \cO\left(\tau_I^5\left((p+q)^2,p,q\right) + \zeta^5\left((p+q)^2,p,q\right) \cdot (p+q)^{\cO(1)} \cdot n \log n\right)\\
% &\leq& \cO\left(\frac{1}{x^p(1-x)^q} \cdot \left(2^{2^{\log^4(p+q)}}+2^{\cO(\frac{p+q}{\log\log(p+q)})}\right) \cdot (p+q)^{\cO(1)} \cdot n \log n\right)\\
% &\leq& \cO\left(\frac{1}{x^p(1-x)^q} \cdot 2^{2^{\log^4(p+q)}} \cdot (p+q)^{\cO(1)} \cdot n \log n\right)
&=&\cO\left(\frac{2^{2^{(2\log\log(p+q))^4}}\cdot (\log(p+q))^{\cO((\log(p+q))^2)}} {x^p(1-x)^q} + 2^{\cO(\frac{p+q}{\log\log(p+q)})}\cdot \frac{(p+q)^{\cO(1)}}{x^p(1-x)^q}\cdot n\log n \right)
\end{eqnarray*}
\item $(\chi,p')$-degree, 
\begin{eqnarray*}
\Delta_{(\chi,p')}^6(n,p,q) &\leq& \Delta_{(\chi,p')}^5\left((p+q)^2,p,q\right) \cdot  (p+q)^{\cO(1)} \cdot \log n \\
&\leq& \cO\left( 2^{\cO(\frac{p+q}{\log\log(p+q)})} \cdot \frac{1}{x^{p-p'}(1-x)^q} \cdot (p+q)^{\cO(1)} \cdot \log n\right)
\end{eqnarray*}
\item $(\chi,p')$-query time,
\begin{eqnarray*}
Q_{(\chi,p')}^6(n,p,q) &\leq& \cO\left(\left(Q_{(\chi,p')}^5\left((p+q)^2,p,q\right) + \Delta_{(\chi,p')}^5\left((p+q)^2,p,q\right) \right) \cdot (p+q)^{\cO(1)} \cdot \log n\right)\\
&\leq& \cO\left( 2^{\cO(\frac{p+q}{\log\log(p+q)})} \cdot \frac{1}{x^{p-p'}(1-x)^q} \cdot (p+q)^{\cO(1)} \cdot \log n\right)
\end{eqnarray*}
\item $(\chi',q')$-degree, 
\begin{eqnarray*}
\Delta_{(\chi',q')}^6(n,p,q)&=&\Delta_{(\chi',q')}^5\left((p+q)^2,p,q\right) \cdot  (p+q)^{\cO(1)} \cdot \log n\\
&\leq& \cO\left( 2^{\cO(\frac{p+q}{\log\log(p+q)})} \cdot \frac{1}{x^{p}(1-x)^{q-q'}} \cdot (p+q)^{\cO(1)} \cdot \log n\right)
\end{eqnarray*}
\item $(\chi',q')$-query time, 
\begin{eqnarray*}
Q_{(\chi',q')}^6(n,p,q)&=&\cO\left(\left(Q_{(\chi',q')}^5\left((p+q)^2,p,q\right) + \Delta_{(\chi',q')}^5\left((p+q)^2,p,q\right) \right) \cdot (p+q)^{\cO(1)} \cdot \log n\right)\\
&\leq& \cO\left( 2^{\cO(\frac{p+q}{\log\log(p+q)})} \cdot \frac{1}{x^{p}(1-x)^{q-q'}} \cdot (p+q)^{\cO(1)} \cdot \log n\right)
\end{eqnarray*}
\end{itemize}
We apply Lemma~\ref{lem:splitSolution} to this construction by setting $s=\lfloor(\log (p+q))^2 \rfloor$ and 
$t = \lceil \frac{p+q}{s} \rceil$.
\begin{itemize}%\setlength\itemsep{-.7mm}
\item size,
\begin{eqnarray*}
\zeta^7(n,p,q) &\leq& |\mathscr{P}_t^n| \cdot 
\sum_{(p_1,\ldots,p_t) \in {\cal Z}_{s,t}^{p} } \prod_{i \leq t} \zeta^6(n,p_i,s-p_i)\\
&\leq& n^{O(t)}\cdot(p+q)^{\cO(t)}\cdot s^{\cO(t)} \cdot 2^{\cO(\frac{st}{\log\log s})}\cdot (\log n)^{\cO(t)}\cdot \frac{1}{x^p(1-x)^{q+s}}\\ 
&\leq& n^{\cO(\frac{p+q}{\log^2(p+q)})} \cdot 2^{\cO(\frac{p+q}{\log\log\log(p+q)})} \frac{1}{x^p(1-x)^q}
\end{eqnarray*}
\item initialization time,
\begin{eqnarray*}
 \tau_I^7(n,p,q) &\leq& \cO\left(\left(\sum_{\hat{p} \leq s} \tau_I^6(n,\hat{p},s-\hat{p})\right) + \zeta^7(n,p,q) \cdot n^{\cO(1)}\right)\\
&\leq& 2^{2^{(2\log\log (s))^4}}\cdot \frac{(\log s )^{\cO(\log^2(s))}}{x^p(1-x)^q} + n^{\cO(\frac{p+q}{\log^2(p+q)})} \cdot 2^{\cO(\frac{p+q}{\log\log\log(p+q)})} \frac{1}{x^p(1-x)^q}\\
&\leq& n^{\cO(\frac{p+q}{\log^2(p+q)})} \cdot 2^{\cO(\frac{p+q}{\log\log\log(p+q)})} \frac{1}{x^p(1-x)^q}\\
\end{eqnarray*}
% \textcolor{red}{
($\because 2^{2^{(2\log \log s)^4}} , (\log s )^{\cO(\log^2(s))} \leq 2^{\cO(\frac{p+q}{\log \log\log(p+q)})}$.
This inequality holds because $\log\log 2^{2^{(2\log \log s)^4}}$ is upper bounded by a polynomial function in 
$\log\log\log(p+q)$ where as $\log\log 2^{\cO(\frac{p+q}{\log \log\log(p+q)})}$ is lower bounded by a polynomial 
function in $\log(p+q)$. Similarly $\log (\log s )^{\cO(\log^2(s))}$ is upper bounded by a polynomial function in 
$\log(p+q)$ where as $\log 2^{\cO(\frac{p+q}{\log \log\log(p+q)})}$ is lower bounded by a polynomial function in $(p+q)$)
% }
% &\leq&\cO\left(\frac{2^{\cO(\frac{s}{\log s})}\cdot s^{\cO(1)}}{x^s(1-x)^s}\left(2^{2^{s^2}}+\log n\right) + n^{\cO(\frac{p+q}{\log^2(p+q)})} \cdot 2^{O(\frac{p+q}{\log(p+q)})} \frac{1}{x^p(1-x)^q}
% \right)\\
% &\leq& \cO\left( \frac{2^{2^{(\log^4(p+q))}}}{x^p(1-x)^q} + n^{O(\frac{p+q}{\log^2(p+q)})} \cdot 2^{O(\frac{p+q}{\log(p+q)})} \frac{1}{x^p(1-x)^q}
% \right)
% \end{eqnarray*}
 \item $(\chi,p')$-degree, 
\begin{eqnarray*}
\Delta_{(\chi,p')}^7(n,p,q) &\leq& |\mathscr{P}_t^n|\cdot |{\cal Z}_{s,t}^{p}|\cdot 
\max_{\substack{(p_1,\ldots,p_t) \in {\cal Z}_{s,t}^{p} \\ p_1'\leq p_1,\ldots, p_t'\leq p_t \\ p_1'+\ldots+p_t'=p'  }} \prod_{i \leq t}  \Delta_{(\chi,p_i')}^6(n,p_i,s-p_i)\\
&\leq& n^{\cO(t)}\cdot(p+q)^{\cO(t)}\cdot s^{\cO(t)} \cdot 2^{\cO(\frac{st}{\log\log s})}\cdot (\log n)^{\cO(t)}\cdot \frac{1}{x^{p-p'}(1-x)^{q+s}}\\
&\leq& n^{\cO(\frac{p+q}{\log^2(p+q)})} \cdot 2^{\cO(\frac{p+q}{\log\log\log(p+q)})} \cdot \frac{1}{x^{p-p'}(1-x)^{q}}
\end{eqnarray*}
\item $(\chi',q')$-degree, 
\begin{eqnarray*}
\Delta_{(\chi',q')}^7(n,p,q) 
% &\leq& |\mathscr{P}_t^n|\cdot \max_{\substack{(r_1,\ldots,r_t) \in {\cal Z}_{s,t}^{p} \\ r_1''\leq r_1,\ldots, r_t''\leq r_t \\ r_1''+\ldots+r_t''=p_2  }} \prod_{i \leq t}  \Delta_{s-r_i}^6(n,r_i,s-r_i)\\
&\leq& n^{\cO(\frac{p+q}{\log^2(p+q)})} \cdot 2^{\cO(\frac{p+q}{\log\log\log(p+q)})} \cdot \frac{1}{x^{p}(1-x)^{q-q'}}
\end{eqnarray*}
\item $(\chi,p')$-query time, 
\begin{eqnarray*}
Q_{(\chi,p')}^7(n,p,q) &\leq& \cO\left(\Delta_{(\chi,p')}^7(n,p,q) \cdot n^{\cO(1)} + |\mathscr{P}_{t}^{n}|\cdot |{\cal Z}_{s,t}^{p}|\cdot t \cdot 
\max_{\hat{p}'\leq \hat{p} \leq s } Q_{(\chi,\hat{p}')}^6(n,\hat{p},s-\hat{p}) \right)\\
&\leq& n^{\cO(\frac{p+q}{\log^2(p+q)})}\cdot  2^{\cO(\frac{p+q}{\log\log\log(p+q)})} \cdot \frac{1}{x^{p-p'}(1-x)^q}\log n 
\end{eqnarray*}
\item $(\chi',q')$-query time,
\begin{eqnarray*}
 Q_{(\chi',q')}^7(n,p,q) 
% &\leq& \cO\big(\Delta_{q}^7(n,p,q) \cdot n^{\cO(1)} + |\mathscr{P}_{t}^{n}|\cdot t s^t\cdot 
% \max_{\substack{\hat{p} \leq s}} \tau_{Q_{s-\hat{p}}}^6(n,\hat{p},s-\hat{p}) \big)\\
% &\leq& n^{\cO(\frac{p+q}{\log^2(p+q)})}\cdot  2^{O(\frac{p+q}{\log\log(p+q)})} \cdot \frac{1}{x^{p}}
&\leq& n^{\cO(\frac{p+q}{\log^2(p+q)})}\cdot  2^{\cO(\frac{p+q}{\log\log\log(p+q)})} \cdot \frac{1}{x^{p}(1-x)^{q-q'}}\log n 
% &\leq& n^{O(\frac{p+q}{\log^2(p+q)})} \cdot \frac{1}{x^{p}}
\end{eqnarray*}
\end{itemize}
We apply Lemma~\ref{lem:twinreduceUniverse} to this construction to get a new construction with the following parameters.
\begin{itemize} %\setlength\itemsep{-.7mm}
\item size, 
\begin{eqnarray*}
\zeta^8(n,p,q) &\leq& \zeta^7\left((p+q)^2,p,q\right) \cdot  (p+q)^{\cO(1)} \cdot \log n\\
&\leq & 2^{\cO(\frac{p+q}{\log\log\log(p+q)})}\cdot \frac{1}{x^p(1-x)^q}\cdot (p+q)^{\cO(1)} \cdot \log n
\end{eqnarray*}
\item initialization time, 
\begin{eqnarray*}
\tau_I^8(n,p,q) &\leq& \cO\left(\tau_I^7\left((p+q)^2,p,q\right) + \zeta^7\left((p+q)^2,p,q\right) \cdot (p+q)^{\cO(1)} \cdot n \log n\right)\\
&\leq & 2^{\cO(\frac{p+q}{\log\log\log(p+q)})}\cdot \frac{1}{x^p(1-x)^q}\cdot (p+q)^{\cO(1)} \cdot n\log n
% &\leq& \cO\left(\frac{1}{x^p(1-x)^q} \cdot \left(2^{2^{\log^4(p+q)}}+2^{O(\frac{p+q}{\log\log(p+q)})}\right) \cdot (p+q)^{\cO(1)} \cdot n \log n\right)\\
% &\leq& \cO\left(\frac{1}{x^p(1-x)^q} \cdot 2^{2^{\log^4(p+q)}} \cdot (p+q)^{\cO(1)} \cdot n \log n\right)
\end{eqnarray*}
\item $(\chi,p')$-degree, 
\begin{eqnarray*}
\Delta_{(\chi,p')}^8(n,p,q) &\leq& \Delta_{(\chi,p')}^7\left((p+q)^2,p,q\right) \cdot  (p+q)^{\cO(1)} \cdot \log n \\
&\leq & 2^{\cO(\frac{p+q}{\log\log\log(p+q)})}\cdot \frac{1}{x^{p-p'}(1-x)^q}\cdot (p+q)^{\cO(1)} \cdot \log n
% &\leq& \cO\left( 2^{O(\frac{p+q}{\log\log(p+q)})} \cdot \frac{1}{x^{p-p_j}(1-x)^q} \cdot (p+q)^{\cO(1)} \cdot \log n\right)
\end{eqnarray*}
\item $(\chi,p')$-query time,
\begin{eqnarray*}
Q_{(\chi,p')}^8(n,p,q) &\leq& \cO\left(\left(Q_{(\chi,p')}^7\left((p+q)^2,p,q\right) + \Delta_{(\chi,p')}^7\left((p+q)^2,p,q\right) \right) \cdot (p+q)^{\cO(1)} \cdot \log n\right)\\
&\leq & 2^{\cO(\frac{p+q}{\log\log\log(p+q)})}\cdot \frac{1}{x^{p-p'}(1-x)^q}\cdot (p+q)^{\cO(1)} \cdot \log n
% &\leq& \cO\left( 2^{O(\frac{p+q}{\log\log(p+q)})} \cdot \frac{1}{x^{p-p_j}(1-x)^q} \cdot (p+q)^{\cO(1)} \cdot \log n\right)
\end{eqnarray*}
\item $(\chi',q')$-degree, 
\begin{eqnarray*}
\Delta_{(\chi',q')}^8(n,p,q)&=&\Delta_{(\chi',q')}^7\left((p+q)^2,p,q\right) \cdot  (p+q)^{\cO(1)} \cdot \log n\\
&\leq & 2^{\cO(\frac{p+q}{\log\log\log(p+q)})}\cdot \frac{1}{x^{p}(1-x)^{q-q'}}\cdot (p+q)^{\cO(1)} \cdot \log n
% &\leq& \cO\left( 2^{O(\frac{p+q}{\log\log(p+q)})} \cdot \frac{1}{x^{p}} \cdot (p+q)^{\cO(1)} \cdot \log n\right)
\end{eqnarray*}
\item $(\chi',q')$-query time, 
\begin{eqnarray*}
Q_{(\chi',q')}^8(n,p,q)&=&\cO\left(\left(Q_{(\chi',q')}^7\left((p+q)^2,p,q\right) + \Delta_{(\chi',q')}^7\left((p+q)^2,p,q\right) \right) \cdot (p+q)^{\cO(1)} \cdot \log n\right)\\
&\leq & 2^{\cO(\frac{p+q}{\log\log\log(p+q)})}\cdot \frac{1}{x^{p}(1-x)^{q-q'}}\cdot (p+q)^{\cO(1)} \cdot \log n
% &\leq& \cO\left( 2^{O(\frac{p+q}{\log\log(p+q)})} \cdot \frac{1}{x^{p}} \cdot (p+q)^{\cO(1)} \cdot \log n\right)
\end{eqnarray*}
\end{itemize}
The final construction satisfies all the claimed bounds. This concludes the proof. 
\end{proof}

    %!TEX root = main.tex
\subsection{Representative Sets for Product Families}
%Now 
We are ready to give  the main theorem about product families using the constructions of generalized 
$n$-$p$-$q$-separating collections.  
%In this subsection we give our main the

\begin{theorem}
\label{thm:product_uniform}
Let ${\cal L}_1$ be a $p_1$-family of sets  and ${\cal L}_2$ be a $p_2$-family of sets over a universe $U$ of 
size $n$. 
%Let  $w_j~:~{\cal L}_j\rightarrow \mathbb{N}$, $j\in \{1,2\}$, be  non-negative weight functions. Furthermore let  $w~:~{\cal L}_1\bullet{\cal L}_2  \rightarrow \mathbb{N}$ be  a non-negative weight function such that $w(X\cup Y)=w_1(X)+w_2(Y)$.  
Let $w~:~2^{U}\rightarrow \mathbb{N}$ be an additive weight function. 
Let ${\cal L}={\cal L}_1\bullet{\cal L}_2$ and $p=p_1+p_2$. For any $0<x_1,x_2<1$, there exist 
\lminrep{\cal L}{k-p_1-p_2} of size $2^{\cO(\frac{k}{\log\log\log(k)})}\cdot\frac{1}{x_1^p(1-x_1)^{k-p}}\cdot k^{\cO(1)}\log n$ 
and it can be computed in time 
$\cO\left(z(n,k,W)\cdot \left(\frac{1}{x_1^p(1-x_1)^q}+\frac{1}{x_2^{p_1}(1-x_2)^{p_2}}+
\frac{|{\cal L}_1|}{x_1^{p_2}(1-x_1)^{q}(1-x_2)^{p_2}} + \frac{|{\cal L}_2|}{x_1^{p_1}(1-x_1)^{q}x_2^{p_1}}\right)\right)$, 
where $z(n,k,W)=2^{\cO(\frac{k}{\log\log\log(k)})}k^{\cO(1)}n\log n\log W$ and $W$ is the maximum weight defined by $w$.
\end{theorem}
\begin{proof}
 We set $p=p_1+p_2$ and $q=k-p$. To obtain the desired construction we first 
 define an auxiliary graph and then use it to obtain the $q$-representative for the product family $\cal L$. 
 We first obtain two families of separating collections.
 \begin{itemize}
 \item Apply Lemma~\ref{lem:twin_sep_coll_construction} for $0<x_1<1$ and construct a $n$-$p$-$q$-separating collection $({\cal F,\chi_{\cal F}, \chi_{\cal F}'})$ of size 
$2^{\cO(\frac{p+q}{\log\log\log(p+q)})}\cdot\frac{1}{x_1^p(1-x_1)^q}\cdot(p+q)^{\cO(1)}\log n$ in time linear in the size of ${\cal F}$. 
 \item Apply Lemma~\ref{lem:twin_sep_coll_construction} for $0<x_2<1$ and construct a $n$-$p_1$-$p_2$-separating collection $({\cal H,{\chi}_{\cal H}, \chi_{\cal H}'})$ of size 
$2^{\cO(\frac{p_1+p_2}{\log\log\log(p_1+p_2)})}\cdot\frac{1}{x_2^{p_1}(1-x_2)^{p_2}}\cdot(p_1+p_2)^{\cO(1)}\log n$ in time linear in the size of ${\cal H}$. 
\end{itemize}
Now we construct a graph $G=(V,E)$ where the vertex set $V$ contains a vertex each for sets in ${\cal F}\uplus{\cal H}\uplus {\cal L}_1\uplus{\cal L}_2$.  For clarity of presentation we name the vertices by the corresponding set. Thus, the vertex set  $V={\cal F}\uplus{\cal H}\uplus {\cal L}_1\uplus{\cal L}_2$.  The edge set  $E=E_1\uplus E_2\uplus E_3\uplus E_4$, 
where each $E_i$ for $i\in\{1,2,3,4\}$ is defined as follows (see Figure~\ref{fig_uniform_repset}).
\begin{eqnarray*}
 E_1&=&\Big\{(A,F)~\Big|~A\in {\cal L}_1,~F\in\chi_{\cal F}(A) \Big\}\\
 E_2&=&\Big\{(B,F)~\Big|~B\in {\cal L}_2,~F\in\chi_{\cal F}(B) \Big\}\\
 E_3&=&\Big\{(A,H)~\Big|~A\in {\cal L}_1,~H\in\chi_{\cal H}(A) \Big\}\\
E_4&=& \Big\{(B,F)~\Big|~B\in {\cal L}_2,~F\in\chi_{\cal H}'(B)\Big\}\\
\end{eqnarray*}
Thus $G$ is essentially a $4$-partite graph. 

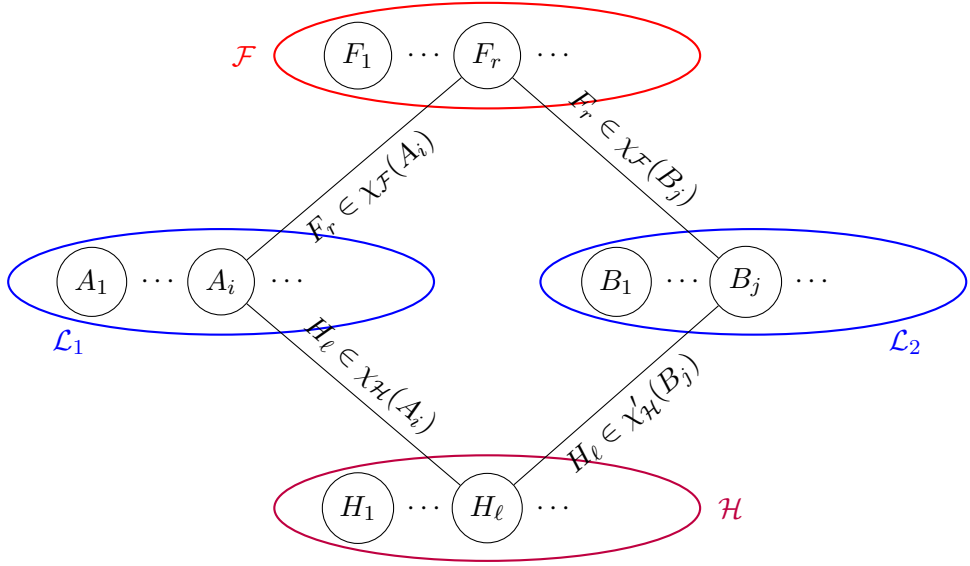
\begin{figure}[!]
   \centering
\begin{tikzpicture}
 \node [red] (F) at (-3.2,0) {${\cal F}$}; 
 \draw [red, thick](0,0) ellipse (2.8cm and 0.7cm);
 \node[draw,circle] (F1) at (-1.7,0) {$F_1$};
 \node (d1) at ( -0.8,0) {$\cdots$}; 
\node[draw,circle] (Fr) at (0,0) {$F_{r}$};
\node (d1) at ( 0.9,0) {$\cdots$};

\node [blue](L_1) at (-5.5,-3.8) {${\cal L}_1$};
\draw [blue,thick](-3.5,-3) ellipse (2.8cm and 0.7cm);
 \node[draw,circle] (A1) at (-5.2,-3) {$A_1$};
 \node (d1) at ( -4.3,-3) {$\cdots$}; 
\node[draw,circle] (Ai) at (-3.5,-3) {$A_{i}$};
\node (d1) at ( -2.6,-3) {$\cdots$};

\node [blue](L_2) at (5.5,-3.8) {${\cal L}_2$};
 \draw [blue,thick](3.5,-3) ellipse (2.8cm and 0.7cm);
 \node[draw,circle] (B1) at (1.7,-3) {$B_1$};
 \node (d1) at ( 2.6,-3) {$\cdots$}; 
\node[draw,circle] (Bj) at (3.4,-3) {$B_{j}$};
\node (d1) at ( 4.3,-3) {$\cdots$};

\node [purple](H) at (3.2,-6) {${\cal H}$};
\draw [purple,thick](0,-6) ellipse (2.8cm and 0.7cm);
\node[draw,circle] (H1) at (-1.7,-6) {$H_1$};
 \node (d1) at ( -0.8,-6) {$\cdots$}; 
\node[draw,circle] (Hl) at (0,-6) {$H_{\ell}$};
\node (d1) at ( 0.9,-6) {$\cdots$};

\draw (Ai) edge node[below,sloped] {$F_r\in {\chi}_{{\cal F}}(A_i)$} (Fr);
\draw (Fr) edge node[above,sloped] {$F_r\in {\chi}_{{\cal F}}(B_j)$} (Bj);
\draw (Ai) edge node[above,sloped] {$H_\ell\in {\chi}_{{\cal H}}(A_i)$} (Hl);
\draw (Hl) edge node[below,sloped] {$H_\ell\in {\chi}'_{{\cal H}}(B_j)$} (Bj);
\end{tikzpicture}
\caption{Graph constructed from ${\cal L}_1, {\cal L}_2, {\cal F}$ and ${\cal H}$}
\label{fig_uniform_repset}
\end{figure}

\paragraph{Algorithm.} The construction of $\widehat{\cal L}$ is as follows. For a set $F\in{\cal F}$, we call a pair 
% $A\cup B $ to $\widehat{\cal L}$, where 
of sets $(A,B)$ {\em cyclic}, if $A\in{\cal L}_1,~B\in{\cal L}_2$ and there exists $H\in{\cal H}$ such that $FAHB$ forms a cycle of length four in $G$. Let ${\cal J}(F)$ denote the family of cyclic pairs for a set $F\in {\cal F}$ and 
\[w_F=\min_{(A,B)\in {\cal J}(F)} w(A)+w(B).\] 
We obtain the family $\widehat{\cal L}$ by adding  $A\cup B$ for every set $F\in {\cal F}$ such that $(A,B)\in {\cal J}(F)$ and $w(A)+w(B)=w_F$. Indeed, if the family ${\cal J}(F)$ is empty then we do not add any set to $\widehat{\cal L}$ corresponding to $F$. 
%with the property that $w(A\cup B)$ is as small as possible (if there 
%exists one such $A\cup B$). 
%
%\[{\cal J}(F)=\{(A,B)~|~A\in{\cal L}_1 \wedge B\in{\cal L}_2 \wedge(\mbox{$\exists~H\in{\cal H}$ such that $FABH$ forms a cycles of length $4$} \} \]
The procedure to find the smallest weight $A\cup B$ for any $F$ is as follows. We first mark 
the vertices of $N_G(F)$ (the neighbors of $F$). Now we mark the neighbors of  
${\cal P}=(N_G(F)\cap {\cal L}_1)$ in ${\cal H}$. For every marked vertex $H\in {\cal H}$, we associate a set $A$  of minimum weight such that $A\in ({\cal P}\cap N_G(H))$.
%and associate 
%the smallest weight $A$ in $N_G(F)\cap {\cal L}_1$ with each marked vertex $H$ in ${\cal H}$, such that 
%$A\in N_G(H)$. 
This can be done sequentially  as follows. Let ${\cal P}=\{S_1,\ldots,S_\ell\}$. Now iteratively visit the neighbors of 
$S_i$ in ${\cal H}$, $i\in [\ell]$, and for each vertex of $\cal H$ store the smallest weight vertex $S \in {\cal P}$ it has seen so far. 
%by visiting the neighbors of $S$ in ${\cal H}$ in a  and for each vertex of 
%$\cal H$  storing the smallest weighted marked vertex in $S$ seen so far.  
After this we have a marked set of vertices 
in ${\cal H}$ such that with each marked vertex $H$ in ${\cal H}$ we stored a smallest weight marked vertex in 
${\cal L}_1$ which is a neighbor of $H$.  Now for each marked vertex $B$ in ${\cal L}_2$, we go through the neighbors of  $B$ in the marked set of vertices in ${\cal H}$ and associate (if possible) a second vertex (which is a minimum weighted marked neighbor from ${\cal L}_2$) with each marked vertex in ${\cal H}$. We obtain a pair of sets $(A,B)\in{\cal J}(F)$ such that $w(A)+w(B)=w_F$. 
%Now we want to add minimum weighted 
%$A\cup B$, such that there exist a marked vertex $H\in{\cal H}$ such that its two associated vertices are $A$ and 
%$B$. 
This can be easily done by keeping a variable that stores a minimum weighted $A\cup B$ seen after every step of 
marking procedure. 
% second phase
%This can be done by remembering the minimum weighted $A\cup B$ which is associated with any marked vertex 
%in ${\cal H}$ so far during the marking process.     
% then add $A\cup B$ to $\widehat{\cal L}$. 
Since for each $F\in{\cal F}$ we add at most one set to $\widehat{\cal L}$, the size of $\widehat{\cal L}$ follows.

\paragraph{Correctness.} We first show that $\widehat{\cal L}\subseteq \cal L$. Towards this we only need to show 
that for every  $A\cup B\in\widehat{\cal L}$ we have that $A\cap B=\emptyset$. Observe that if $A\cup B\in \widehat{\cal L}$  
then there exists a $F\in{\cal F},~H\in {\cal H}$ such that $FAHB$ forms a cycle of length four in the graph 
$G$. So $H\in\chi_{\cal H}(A)$ and $H\in \chi_{\cal H}'(B)$. This means $A\subseteq H$ and $B\cap H=\emptyset$. So we conclude $A$ and $B$ are disjoint and hence $\widehat{\cal L}\subseteq \cal L$. 
We also need to show that if there exist pairwise disjoint sets $A\in{\cal L}_1,B\in{\cal L}_2,C\in{U \choose q}$, 
then there exist $\hat{A}\in{\cal L}_1,\hat{B}\in{\cal L}_2$ such that $\hat{A}\cup\hat{B}\in\widehat{\cal L}$,
$\hat{A},\hat{B},C$ are pairwise disjoint and $w(\hat{A})+w(\hat{B})\leq w(A)+w(B)$. By the property of separating 
collections $({\cal F},\chi_{\cal F},\chi_{\cal F}')$ and $({\cal H},\chi_{\cal H},\chi_{\cal H}')$, we know that 
there exists $F\in\chi_{\cal F}(A)\cap\chi_{\cal F}(B)\cap\chi_{\cal F}'(C),~H\in\chi_{\cal H}(A)\cap\chi_{H}'(B)$. This 
implies that $FAHB$ forms a cycle of length four in the graph $G$. Hence in the construction of 
$\widehat{\cal L}$, we should have chosen  $\hat{A}\in{\cal L}_1$ and $\hat{B}\in{\cal L}_2$ corresponding to $F$ 
such that $w(\hat{A})+w(\hat{B})\leq w(A)+w(B)$ and added to $\widehat{\cal L}$. 
So we know that $F\in\chi_{\cal F}(\hat{A})\cap \chi_{\cal F}(\hat{B})$. Now we claim that $\hat{A},\hat{B}$ 
and $C$ are pairwise disjoint. Since $\hat{A}\cup\hat{B}\in \widehat{\cal L}$, $\hat{A}\cap \hat{B}=\emptyset$. Finally,   since $F\in\chi_{\cal F}(\hat{A})\cap \chi_{\cal F}(\hat{B})$ and $F\in\chi_{\cal F}'(C)$, we get 
$\hat{A},\hat{B}\subseteq F$ and $F\cap C=\emptyset$ which implies $C$ is disjoint from $\hat{A}$ and $\hat{B}$. 
This completes the correctness proof.

\paragraph{Running Time Analysis.}
We first consider the time $T_G$ to construct the graph $G$.
We can construct $\cal F$ in time $2^{\cO(\frac{p+q}{\log\log\log(p+q)})}\cdot\frac{1}{x_1^p(1-x_1)^q}\cdot(p+q)^{\cO(1)}\cdot n\log n$.
We can construct $\cal H$ in time $2^{\cO(\frac{p_1+q}{\log\log\log(p_1+p_2)})}\cdot\frac{1}{x_2^{p_1}(1-x_2)^{p_2}}\cdot(p_1+p_2)^{\cO(1)}\cdot n\log n$. 
Now to add edges in the graph we do as follows. For each vertex in ${\cal L}_1\cup {\cal L}_2$, we query the data structure created, spending the query time mentioned in Lemma~\ref{lem:twin_sep_coll_construction}, and add edges to the  vertices in ${\cal F}\cup {\cal H}$ from it. So the running time to construct $G$ is,
\begin{eqnarray*}
 T_G &\leq& 2^{\cO(\frac{k}{\log\log\log(k)})}k^{\cO(1)}n\log n\Big(
\frac{1}{x_1^p(1-x_1)^q}+\frac{1}{x_2^{p_1}(1-x_2)^{p_2}}+\frac{|{\cal L}_1|}{x_1^{p_2}(1-x_1)^{q}} \\
&&\;\;\;\;\;\;\;\;\;\;\;\;\;\;\;\;\;\;\;\;\;\;\;\;\;\;\;\;\;\;\;\; +\frac{|{\cal L}_2|}{x_1^{p_1}(1-x_1)^{q}} +\frac{|{\cal L}_1|}{(1-x_2)^{p_2}}+\frac{|{\cal L}_2|}{x_2^{p_1}}
\Big). 
\end{eqnarray*}
Now  we bound the time $T_C$ taken to construct $\widehat{\cal L}$ from $G$. To do the analysis we see how may times a 
vertex $A$ in ${\cal L}_1\cup{\cal L}_2$ is visited. It is exactly equal to the product of the degree of $A$ to 
${\cal F}$  (denoted by ${\rm degree}_{\cal F}(A)$) and the degree of $A$ to ${\cal H}$ 
(denoted by ${\rm degree}_{\cal H}(A)$). Also note that two weights can be compared in $\cO(\log W)$ time. Then 
\begin{eqnarray*}
 T_C& \leq&\log W\left( \sum_{A\in{\cal L}_1}{\rm degree}_{\cal F}(A)\cdot {\rm degree}_{\cal H}(A) + 
\sum_{A\in{\cal L}_2} {\rm degree}_{\cal F}(A)\cdot {\rm degree}_{\cal H}(A) \right) \\
&\leq&\log W\left( \sum_{A\in{\cal L}_1} \Delta_{(\chi_{\cal F},p_1)}(n,p,q)\cdot \Delta_{(\chi_{\cal H},p_1)}(n,p_1,p_2) +
\sum_{A\in{\cal L}_2} \Delta_{({\chi}_{\cal F},p_2)}(n,p,q)\cdot \Delta_{(\chi'_{\cal H},p_2)}(n,p_1,p_2)\right) \\
&\leq&2^{\cO(\frac{k}{\log\log\log(k)})}k^{\cO(1)}\log^2 n \log W\left(\frac{|{\cal L}_1|}{x_1^{p_2}(1-x_1)^{q}(1-x_2)^{p_2}} + \frac{|{\cal L}_2|}{x_1^{p_1}(1-x_1)^{q}x_2^{p_1}} \right).
\end{eqnarray*}
So the total running time $T$ is,
\begin{eqnarray*}
T&=&T_G+T_C\\
&\leq&2^{\cO(\frac{k}{\log\log\log(k)})}k^{\cO(1)}n\log n\cdot \log W\Big(\frac{1}{x_1^p(1-x_1)^q}+\frac{1}{x_2^{p_1}(1-x_2)^{p_2}}\\
&& \;\;\;\;\;\;\;\;\;\;\;\;\;\;\;\;\;\;\;\;\;\;\;\;\;\;\;\;\;\;\;\;+\frac{|{\cal L}_1|}{x_1^{p_2}(1-x_1)^{q}(1-x_2)^{p_2}} + \frac{|{\cal L}_2|}{x_1^{p_1}(1-x_1)^{q}x_2^{p_1}} \Big).
\end{eqnarray*}
This completes the proof of the theorem. 
\end{proof}

Now we give a ready to use corollary for Theorem~\ref{thm:product_uniform}. 
\begin{corollary}\label{cor:product_uniform}
Let ${\cal L}_1$ be a $p_1$-family of sets  and ${\cal L}_2$ be a $p_2$-family of sets over a universe 
$U$ of size $n$. Furthermore, let $w~:~2^{U}\rightarrow \mathbb{N}$ be an additive weight function, 
 $|{\cal L}_1|={k\choose p_1}$, $|{\cal L}_2|={k \choose p_2}$, ${\cal L}={\cal L}_1\bullet{\cal L}_2$, $p=p_1+p_2$ and $q=k-p$. 
% For any $0<x_1<1,~0<x_2<1$, 
There exists \lminrep{\cal L}{k-p_1-p_2} of size ${k\choose p}\cdot 2^{o(k)}\cdot \log n$
% $2^{O(\frac{k}{\log\log(k)})}\cdot\frac{1}{x_1^p(1-x_1)^{k-p}}\cdot k^{O(1)}\log n$ 
and 
it can be computed in time 
%\\ 
\[\min_{0<x_1,x_2<1}\cO\left(
% \frac{1}{x_1^p(1-x_1)^q}+
\frac{z(n,k,W)}{x_2^{p_1}(1-x_2)^{p_2}}+\frac{{k \choose p_1}\cdot z(n,k,W)}{x_1^{p_2}(1-x_1)^{q}(1-x_2)^{p_2}} + 
%\\ 
\frac{{k\choose p_2}\cdot z(n,k,W)}{x_1^{p_1}(1-x_1)^{q}x_2^{p_1}}+ \frac{(\frac{k}{q})^{q}\cdot z(n,k,W)}{x_1^p(1-x_1)^{q}}\right).\]
%\\
Here $z(n,k,W)=2^{\cO(\frac{k}{\log\log\log(k)})}k^{\cO(1)}n\log n\cdot \log W$ and $W$ is the maximum weight defined by $w$. 
\end{corollary}
\begin{proof}
% Let $q=k-p$. 
We apply Theorem~\ref{thm:product_uniform} for $0<x_1,x_2<1$ 
% (we fix it later) 
and find ${\cal L}'\subseteq_{minrep}^{k-p_1-p_2}{\cal L}$ of size 
$2^{O(\frac{k}{\log\log\log(k)})}\cdot\frac{1}{x_1^p(1-x_1)^{k-p}}\cdot k^{O(1)}\log n$ 
in time $T_1=\cO(\frac{z(n,k,W)}{x_1^p(1-x_1)^q}+\frac{z(n,k,W)}{x_2^{p_1}(1-x_2)^{p_2}}+\frac{z(n,k,W)\cdot|{\cal L}_1|}{x_1^{p_2}(1-x_1)^{q}(1-x_2)^{p_2}} + \frac{z(n,k,W)\cdot |{\cal L}_2|}{x_1^{p_1}(1-x_1)^{q}x_2^{p_1}})$. 
Now we apply Theorem~\ref{thm:repsetweighted} and get ${\widehat{\cal L}}\subseteq_{minrep}^{k-p_1-p_2}{\cal L}'$ of 
size ${k\choose p}2^{o(k)}\log n$ in time $T_2=\\\cO\left(\left(\frac{k}{q}\right)^{q} 2^{o(k)}\cdot\frac{1}{x_1^p(1-x_1)^{k-p}}\cdot k^{\cO(1)}\log^2 n\cdot\log W \right)$. Due to Lemma~\ref{lem:reptransitive}, \lminrep{{\cal L}}{k-p_1-p_2}. 
Now we choose $x_1,x_2$ such that $T_1+T_2$ is minimized. So the total running 
time $T$ to construct ${\widehat{\cal L}}$ is,
\begin{eqnarray*}
 T&=&\min_{x_1,x_2}\left(T_1+T_2\right)\\
&=&\min_{x_1,x_2}\cO\left(
\frac{z(n,k,W)}{x_2^{p_1}(1-x_2)^{p_2}}+\frac{z(n,k,W)\cdot|{k \choose p_1}|}{x_1^{p_2}(1-x_1)^{q}(1-x_2)^{p_2}} + 
\frac{z(n,k,W)\cdot |{k\choose p_2}|}{x_1^{p_1}(1-x_1)^{q}x_2^{p_1}}+ \frac{z(n,k,W)\cdot(\frac{k}{q})^{q}}{x_1^p(1-x_1)^{q}}\right).
\end{eqnarray*}
 This completes the proof.
\end{proof}
%%%%%%%%%%%%%%%%%%%%%%%%%%%%%%%%%%%%%%%%%%%

  \section{Representative set computation for product families of a  linear matroid}
   %!TEX root = main.tex
%  \begin{definition}
% Given two families of independent sets ${\cal L}_1$ and ${\cal L}_2$ of a matroid \mat{}, we define 
% \begin{eqnarray*}
% {\cal L}_1 \bullet {\cal L}_2&=&\{X \cup Y~|~X \in {\cal L}_1 \wedge Y \in {\cal L}_2 \wedge X \cap Y = \emptyset \wedge X\cup Y\in \I\}\\
% \end{eqnarray*}
% \end{definition}
In this section we give an algorithm to compute $q$-representative for product families of a linear matroid. That is, given a matroid \mat{}, families of independent sets ${\cal A}$ and ${\cal B}$ of sets of sizes $p_1$ and $p_2$ respectively, and a positive integer $q$, we compute \rep{{\cal F}}{q}, where, ${\cal F}={\cal A} \bullet {\cal B}$, of size ${p_1+p_2+q \choose p_1+p_2}$ efficiently.   We compute $q$-representative for $\cal F$ in two steps. In the first step we compute an intermediate family of $q$-representative and then apply Theorem~\ref{thm:repsetlovaszweighted} to compute 
$q$-representative of the desired size. The intermediate family of $q$-representative is obtained by computing  
$q$-representative of {\em slices}, ${\cal A} \bullet \{B\}$ for all $B\in {\cal B}$, and then take its union. 
%and then comp
We start with the following lemma that will be central to our faster algorithm for computing the desired $q$-representative for product families of a linear matroid.

\begin{lemma}[Slice Computation Lemma]
\label{lemma:repset_product_1}
Let \mat{}   be a linear matroid of rank $k$,  ${\cal L}$ be a $p_1$-family of independent sets of $M$ 
and $S\in\I$ of size $p_2$. Furthermore, let 
%Let 
$w~:~{\cal L}\bullet \{S\}~\rightarrow~{\mathbb N}$ be a non-negative weight function.   
%Then there exists \lminrep{{\cal L} \bullet {\{S\}}}{k-p_1-p_2} 
% (\lmaxrep{{\cal L} \bullet {\{S\}}}{k-i-j}) 
%of size ${k-p_2 \choose p_1}$. 
% \bnoml{p+q}{p}.  
%Moreover, 
Then given a representation \repmat{M}  of $M$ over a field $ \mathbb{F}$, we can find  
\lminrep{{\cal L} \bullet \{S\}}{k-p_1-p_2} 
% (\lmaxrep{{\cal L} \bullet {\{S\}}} {k-i-j}) 
of size at most ${k-p_2 \choose p_1}$ 
in $\cO\left({k-p_2 \choose p_1} |{\cal L}| p_1^\omega + |{\cal L}| {k-p_2 \choose p_1} ^{\omega-1} \right)$ operations over ${\mathbb F}$.  
\end{lemma}
\begin{proof}
%Consider the set ${\cal L} \bullet \{S\}$. By definition,

Observe that  ${{\cal L}\bullet \{S\}}$ is a $p_1+p_2$-family of independent sets of $M$ and all sets in 
${{\cal L}\bullet \{S\}}$ contain $S$ as a subset.  Let \repmat{M}  the matrix representing the matroid $M$ over a field 
${\mathbb F}$.
%Consider the matrix representation \repmat{M} of $M$ over a field ${\mathbb F}$.  
Without loss of generality we can assume that the first $p_2$ columns of \repmat{M} correspond to the elements in $S$. 
Furthermore, we can also 
%Without loss of generality we can 
assume that the first $p_2$ columns and $p_2$ rows form an identity matrix $I_{p_2\times p_2}$. That is, if $S$ denotes the first $p_2$ columns and $Z$ denotes the first $p_2$ rows then the submatrix $A_M[Z,S]$ is $I_{p_2\times p_2}$. The reason for the last assertion is that if the matrix is not in the required form then we can apply elementary row operations and obtain the matrix in the desired form. This also allows us to assume that the number of rows in \repmat{M} is $k$. 
% with $S$ being the first $p_2$ columns 
%and 
%Otherwise we can do row reductions to make \repmat{M} the above mentioned form. 
So \repmat{M} have the following form. 
\[ \left( \begin{array}{c|c} I_{p_2\times p_2} &  \;\;A\;\; \\ \hline
0 & \;\;B\;\; \end{array} \right) \]

\medskip 

Let \repmat{M/S} be the matrix obtained after deleting first $p_2$ rows and first $p_2$ columns from \repmat{M}. 
That is, \repmat{M/S}$=B$. 
Let $M/S=(E_s,{\cal I}_s)$ be the matriod represented by the \repmat{M/S} on the underlying ground set 
$E_s=E\setminus S$.  Observe that the \rank{M/S}=\rank{B}$=k-p_2$, else the \rank{A_M} would become strictly smaller than $k$.  Let $e_1,e_2,\ldots, e_{p_2}$ be the first ${p_2}$ column vectors of \repmat{M}, i.e., they are columns corresponding to the elements of $S$. 
For a column vector $v$ in \repmat{M}, $\bar{v}$ is used to denote the column vector restricted to the matrix \repmat{M/S} (i.e., $\bar{v}$ contains the last $k-{p_2}$ entries of $v$).

Now consider the set ${\cal L}(S)=\{X\;|\;X\cup S\in {\cal L}\bullet\{S\}\}$. We also define a new weight function 
 $w'~:~{\cal L}(S)\rightarrow{\mathbb N}$  as follows: $w'(X)=w(X\cup S)$.  We would like to compute $k-p_2$ representative for ${\cal L}(S)$. Towards that goal we first show that  ${\cal L}(S)$ is a ${p_1}$-family of independent sets of $M/S$. Let $X\in {\cal L}(S)$. We know that $X\cup S\in {\cal I}$. 
% Let $e_1,e_2,\dots, e_j$ be the first $j$ column vectors in \repmat{M} (they are the elements in $S$). 
Let $v_1,v_2,\ldots,v_{p_1}$ be the column vectors in \repmat{M} corresponding to the elements in $X$. 
% Let $\bar{v}_i$ denote the column vector $v_i$ restricted to the matrix \repmat{M}$/S$ (i.e, $\bar{v}_i$ contains last $k-j$ entries of $v_i$). 
Suppose $X\notin {\cal I}_s$. Then there exist coefficients $\lambda_1,\dots,\lambda_{p_1}$ such that $\lambda_1 \bar{v}_1+\lambda_2\bar{v}_2+\dots +\lambda_{p_1}\bar{v}_{p_1}=\vec{0}$ and at least one of them is non-zero. Then 
 \begin{displaymath}
     \lambda_1 v_1+\lambda_2 v_2+\dots +\lambda_{p_1} v_{p_1}=\left(
      \begin{array}{c} % brackets may be (...), [...], \{...\}, or left out
         a_1\\
         \vdots\\
         a_{p_2}\\
         0\\
         \vdots\\ 
         0 
      \end{array} \right)
      \end{displaymath}
This implies that $-a_1 e_1- a_2 e_2-\cdots-a_{p_2} e_{p_2}+\lambda_1 v_1+\lambda_2 v_2+\dots +\lambda_{p_1} v_{p_1}=\vec{0}$, which contradicts the fact that $S\cup X\in {\cal I}$.
Hence $X\in {\cal I}_s$ and ${\cal L}(S)$ is a $p_1$-family of independent sets of $M/S$.

Now we apply Theorem~\ref{thm:repsetlovaszweighted} and find \lminrep{{\cal L}(S)}{k-p_1-p_2} 
% (\lmaxrep{{\cal L}(S)}{k-i-j}) 
of size ${k-p_2 \choose p_1}$, by considering ${\cal L}(S)$ as a  
$p_1$-family of independent sets of the matroid $M/S$. We claim that $\whnd{{\cal L}(S)}\bullet {\{S\}}\subseteq_{minrep}^{k-p_1-p_2}{\cal L}\bullet {\{S\}}$. 
% ($\whnd{{\cal L}(S)}\bullet {\{S\}}\subseteq_{maxrep}^{k-i-i}{\cal L}\bullet {\{S\}}$). 
Let $X\cup S\in {\cal L}\bullet {\{S\}}$ and $Y\subseteq E\setminus (X\cup S)$ such that $|Y|=k-p_1-p_2$ and 
$X\cup S\cup Y\in{\cal I}$. We need to show that there exists a $ \whnd{X}\in \whnd{{\cal L}(S)}$ such that 
$\whnd{X}\cup S\cup Y\in {\cal I}$ and $w(\whnd{X}\cup S)\leq w(X\cup S)$. 
%We know that $X\in {\cal L}(S)$. 
We start by showing that that $X\cup Y\in {\cal I}_s$. Let $v_1,v_2,\ldots, v_{k-p_2}$ be the column vectors in \repmat{M} corresponding to the elements of  $X\cup Y$. Suppose $X\cup Y\notin {\cal I}_s$. 
 Then there exist coefficients $\lambda_1,\dots,\lambda_{k-p_2}$ such that $\lambda_1 \bar{v}_1+\lambda_2\bar{v}_2+\dots +\lambda_{k-p_2}\bar{v}_{k-p_2}=\vec{0}$ and at least one of them is non-zero. Then we have the following. 
 \begin{displaymath}
     \lambda_1 v_1+\lambda_2 v_2+\dots +\lambda_{k-p_2} v_{k-p_2}=\left(
      \begin{array}{c} % brackets may be (...), [...], \{...\}, or left out
         b_1\\
         \vdots\\
         b_{p_2}\\
         0\\
         \vdots\\ 
         0 
      \end{array} \right)
      \end{displaymath}
However this implies  that $-b_1 e_1- b_2 e_2-\cdots-b_{p_2} e_{p_2}+\lambda_1 v_1+\lambda_2 v_2+\dots +\lambda_{k-{p_2}} v_{k-{p_2}}=\vec{0}$, which contradicts the fact that $S\cup X\cup Y \in {\cal I}$.
Hence $X\cup Y\in {\cal I}_s$. Since \lminrep{{\cal L}(S)}{k-p_1-p_2}, there exists a set 
% (\lmaxrep{{\cal L}(S)}{k-i-j}), 
$ \whnd{X}\in {\cal L}(S)$, with $w'(\whnd{X})\leq w'(X)$ (i.e $w(\whnd{X}\cup S)\leq w(X\cup S)$) 
% ($w(\whnd{X})\geq w(X)$) 
and 
$\whnd{X}\cup Y\in {\cal I}_s$. We claim that $\whnd{X}\cup S\cup Y \in {\cal I}$. Let $u_1,u_2,\ldots , u_{k-{p_2}}$ be the column vectors in \repmat{M} corresponding to the elements of $\whnd{X}\cup Y$. 
Suppose $\whnd{X}\cup S \cup Y \notin {\cal I}$. Then there exist coefficients $\alpha_1,\dots, \alpha_{k}$ such that 
$\alpha_1 e_1+\alpha_2 e_2 +\cdots+\alpha_{p_2} e_{p_2}+\alpha_{{p_2}+1} u_1+\cdots + \alpha_{k} u_{k-{p_2}}=\vec{0}$ and at least one of the coefficients is non-zero. This implies that $\alpha_{{p_2}+1} \bar{u}_1+\cdots + \alpha_{k} \bar{u}_{k-{p_2}}=\vec{0}$,  where $\bar{u_j}$ are restrictions of $u_j$ to the last $k-p_2$ entries.
% \repmat{M/S} 
This contradicts our assumption that $\whnd{X}\cup Y \in{\cal I}_s$. 
Thus we have shown that  $\whnd{X}\cup Y\cup S\in {\cal I}$. 
The size of $\whnd{{\cal L}(S)}\bullet {\{S\}}$ is ${k-{p_2} \choose p_1}$ and 
it can be found in $\cO\left({k-{p_2} \choose p_1} |{\cal L}| p_1^\omega + |{\cal L}| {k-{p_2} \choose p_1} ^{\omega-1} \right)$ operations over ${\mathbb F}$.  
\end{proof}

Now we are ready to prove the main theorem of this section by using Lemma~\ref{lemma:repset_product_1}.

\begin{theorem}
\label{thm:repset_product}
Let \mat{}   be a linear matroid of rank $k$,  ${\cal L}_1$ be a $p_1$-family of independent sets of $M$ and ${\cal L}_2$ be a $p_2$-family of independent sets of $M$. 
%Then there exists \lminrep{{\cal L}_1 \bullet {\cal L}_2}{k-p_1-p_2} 
% (\lmaxrep{{\cal L}_1 \bullet {\cal L}_2}{k-p_1-p_2}) 
%of size ${k \choose p_1+p_2}$. 
% $k3^k$.  
%Moreover, 
Given a representation \repmat{M}  of $M$ over a field $ \mathbb{F}$, we can find  \lminrep{{\cal L}_1 \bullet {\cal L}_2}{k-p_1-p_2} 
% (\lmaxrep{{\cal L}_1 \bullet {\cal L}_2}{k-p_1-p_2}) 
of size at most 
${k \choose p_1+p_2}$ in 
% $\cO\left(\left(|{\cal L}_2| {k-p_2 \choose p_1} (p_1+p_2)^\omega \right)\left(|{\cal L}_1|+ {k \choose p_1+p_1}\right)+ |{\cal L}_1||{\cal L}_2| {k-p_2 \choose p_1} ^{\omega-1}+|{\cal L}_2|{k-p_2 \choose p_1} {k \choose p_1+p_2} ^{\omega-1} \right)$  
$\cO\left(|{\cal L}_2||{\cal L}_1| {k-p_2 \choose p_1} ^{\omega-1} p_1^{\omega} + |{\cal L}_2|{k-p_2 \choose p_1} {k \choose p_1+p_2} ^{\omega-1} (p_1+p_2)^\omega \right)$
operations over ${\mathbb F}$.
\end{theorem}
\begin{proof}
%   Partition ${\cal L}_1,{\cal L}_2$ according to the size of independent sets.
Let ${\cal L}_2=\{S_1,S_2,\dots,S_{\ell}\}$. Then we have
$${\cal L}_1\bullet {\cal L}_2= \bigcup_{i=1}^\ell {\cal L}_1\bullet \{S_i\}.$$
%\cup {\cal L}_1\bullet \{S_2\} \cup \dots {\cal L}_1\cup \{S_l\}$$
By Lemma~\ref{lem:repunion}, $${\cal L}=\bigcup_{i=1}^{\ell}\whnd{{\cal L}_1\bullet \{S_i\}}\subseteq_{minrep}^{k-p_1-p_2} {\cal L}_1\bullet {\cal L}_2.$$ 
% (${\cal L}=\bigcup_{r=1}^{l}\whnd{{\cal L}_1\bullet \{S_r\}}\subseteq_{maxrep}^{k-p_1-p_2} {\cal L}_1\bullet {\cal L}_2$). 
Using Lemma~\ref{lemma:repset_product_1}, for all $1\leq i\leq \ell$, we find \lminrep{{\cal L}_1 \bullet {\{S_i\}}}{k-p_1-p_2} 
% (\lmaxrep{{\cal L}_1 \bullet {\{S_r\}}}{k-p_1-p_2}) 
of size ${k-p_2 \choose p_1}$ in 
$\cO\left({k-p_2 \choose p_1} |{\cal L}_1| p_1^\omega + |{\cal L}_1| {k-p_2 \choose p_1} ^{\omega-1} \right)=\cO\left(|{\cal L}_1| {k-p_2 \choose p_1} ^{\omega-1} p_1^{\omega} \right)$ operations over ${\mathbb F}$. 
Now $|{\cal L}|=|\bigcup_{i=1}^{\ell}\whnd{{\cal L}_1\bullet \{S_i\}}|\leq |{\cal L}_2|{k-p_2 \choose p_1}$. Now we apply Theorem~\ref{thm:repsetlovaszweighted} and find 
% \lminrep{\bigcup_{r=1}^{l}\whnd{{\cal L}_1\bullet \{S_r\}}}{k-p_1-p_2} (\lmaxrep{\bigcup_{r=1}^{l}\whnd{{\cal L}_1\bullet \{S_r\}}}{k-p_1-p_2}) 
\lminrep{\cal L}{k-p_1-p_2} 
% (\lmaxrep{\cal L}{k-p_1-p_2}) 
of size ${k \choose p_1+p_2}$. The number of operations, denoted by  $T_1$, over ${\mathbb F}$ to find $\whnd{{\cal L}}$ from ${\cal L}$ is
\begin{eqnarray*}
T_1&=&\cO\left({k \choose p_1+p_1} |{\cal L}_2|{k-p_2 \choose p_1} (p_1+p_2)^\omega + |{\cal L}_2|{k-p_2 \choose p_1} {k \choose p_1+p_2} ^{\omega-1} \right)\\
  &=&\cO\left(|{\cal L}_2|{k-p_2 \choose p_1} {k \choose p_1+p_2} ^{\omega-1} (p_1+p_2)^\omega \right).
\end{eqnarray*}
By Lemma~\ref{lem:reptransitive}, $\whnd{{\cal L}}\subseteq_{minrep}^{k-p_1-p_2} {\cal L}_1\bullet {\cal L}_2$. 
% ($\whnd{{\cal L}}\subseteq_{maxrep}^{k-p_1-p_2} {\cal L}_1\bullet {\cal L}_2$). 
% We can find $\whnd{{\cal L}}$ in \\
%  $\cO\left(\left(|{\cal L}_2| {k-p_2 \choose p_1} (p_1+p_2)^\omega \right)\left(|{\cal L}_1|+ {k \choose p_1+p_1}\right)+ |{\cal L}_1||{\cal L}_2| {k-p_2 \choose p_1} ^{\omega-1}+
%       |{\cal L}_2|{k-p_2 \choose p_1} {k \choose p_1+p_2} ^{\omega-1} \right)$  operations over ${\mathbb F}$.
The number of operations, denoted by $T$, over ${\mathbb F}$ to find $\whnd{{\cal L}}$ from ${\cal L}_1$ and ${\cal L}_2$ is
\begin{eqnarray*}
T&=&|{\cal L}_2|\cdot \cO\left(|{\cal L}_1| {k-p_2 \choose p_1} ^{\omega-1} p_1^{\omega} \right) + T_1 \\
    &=&\cO\left(|{\cal L}_2||{\cal L}_1| {k-p_2 \choose p_1} ^{\omega-1} p_1^{\omega} + |{\cal L}_2|{k-p_2 \choose p_1} {k \choose p_1+p_2} ^{\omega-1} (p_1+p_2)^\omega \right).
%      \cO\left({k \choose p_1+p_1} |{\cal L}_2|{k-p_2 \choose p_1} (p_1+p_2)^\omega + |{\cal L}_2|{k-p_2 \choose p_1} {k \choose p_1+p_2} ^{\omega-1} \right)\\
%   &=& \cO\left(\left(|{\cal L}_2| {k-p_2 \choose p_1} (p_1+p_2)^\omega \right)\left(|{\cal L}_1|+ {k \choose p_1+p_1}\right)+ |{\cal L}_1||{\cal L}_2| {k-p_2 \choose p_1} ^{\omega-1}+
%       |{\cal L}_2|{k-p_2 \choose p_1} {k \choose p_1+p_2} ^{\omega-1} \right)
 \end{eqnarray*}
 This completes the proof of the theorem.
\end{proof}
The following form of Theorem~\ref{thm:repset_product} will be directly useful in some applications.
 \begin{corollary}\label{cor:product_general_matroid}
  Let \mat{}   be a linear matroid of rank $k$, ${\cal L}_1$ and ${\cal L}_2$ be two families of independent sets 
of $M$ and the number of sets of size $p$ in ${\cal L}_1$ and ${\cal L}_2$ be at most ${k+c \choose p}$. Here, $c$ is 
a fixed constant. 
% Let ${\cal L}_1\star {\cal L}_2=\bigcup_{i,j}{\cal L}_{1,i}\bullet {\cal L}_{2,j}$ where 
Let ${\cal L}_{r,i}$ be the set of independent sets of size exactly $i$ in ${\cal L}_r$ for $r\in\{1,2\}$. Then 
for all the pairs $i,j\in[k]$, we can find  \lminrep{{\cal L}_{1,i}\bullet {\cal L}_{2,j}}{k-i-j} 
% (\lmaxrep{{\cal L}_{1,i}\bullet {\cal L}_{2,j}}{k-i-j}) 
of size ${k \choose i+j}$, 
 in  total of 
$\cO\left( k^{\omega}\left(2^{\omega}+2\right)^k + k^\omega 2^{k(\omega-1)} 3^k \right)$ operations over 
${\mathbb F}$.
 \end{corollary}
\begin{proof}
 By using Theorem~\ref{thm:repset_product} we can find \lminrep{{\cal L}_{1,i}\bullet {\cal L}_{2,j}}{k-i-j} 
% (\lmaxrep{{\cal L}_{1,i}\bullet {\cal L}_{2,j}}{k-i-j}) 
of size ${k \choose i+j}$ for any 
$i,j\in[k]$ in $\cO\left({k+c \choose j} {k+c \choose i} {k-j \choose i} ^{\omega-1} i^{\omega} + {k+c \choose j}{k-j \choose i} {k \choose i+j} ^{\omega-1} (i+j)^\omega \right)$
operations over ${\mathbb F}$. Let $k'=k+c$. So the total number of operations, denoted by $T$, over ${\mathbb F}$ to find 
$\whnd{{\cal L}_{1,i}\bullet {\cal L}_{2,j}}$ for all $i,j\in[k]$ is,
\begin{eqnarray*}
 T&=&\cO\left(\left(\sum_{i=0}^k\sum_{j=0}^k {k' \choose j} {k' \choose i} {k-j \choose i} ^{\omega-1} i^{\omega}\right) + 
     \left(\sum_{i=0}^k\sum_{j=0}^k {k' \choose j}{k-j \choose i} {k \choose i+j} ^{\omega-1} (i+j)^\omega \right)\right)\\
&=&\cO\left(\left( k^{\omega} \sum_{i=0}^k {k' \choose i} \sum_{j=0}^k {k' \choose j} 2^{(k-j)(w-1)}\right)+ 
     \left(k^\omega \sum_{j=0}^k {k' \choose j} \sum_{i=0}^{k-j} {k-j \choose i} {k \choose i+j} ^{\omega-1} \right) \right)\\
&=&\cO\left(\left( k^{\omega}2^{k(\omega-1)} \sum_{i=0}^k {k' \choose i} \left(1+\frac{1}{2^{(\omega-1)}}\right)^{k'} \right)+ 
\left(k^\omega 2^{k(w-1)}\sum_{j=0}^k {k' \choose j} \sum_{i=0}^{k-j} {k-j \choose i} \right) \right)\\
&=&\cO\left(\left( k^{\omega}2^{k'}\left(2^{(\omega-1)}+1\right)^k \right)+ 
\left(k^\omega 2^{k(w-1)} \sum_{j=0}^k {k' \choose j} 2^{k-j} \right) \right)\\
&=&\cO\left( k^{\omega}2^k\left(2^{(\omega-1)}+1\right)^k + k^\omega 2^{k(\omega-1)} 3^k \right)\\
&=&\cO\left( k^{\omega}\left(2^{\omega}+2\right)^k + k^\omega 2^{k(\omega-1)} 3^k \right).
\end{eqnarray*}
The above simplification completes the proof. 
\end{proof}

\section{Application I: Multilinear Monomial Testing}
 %!TEX root = main.tex
 In this section we first design a faster algorithm for a weighted version of {\sc $k$-MlD} and then give an 
 algorithm for an extension of this to a matroidal version. In the weighted version of {\sc $k$-MlD} in addition to 
 an arithmetic circuit $C$ over variables $X = \{x_1,x_2,\ldots,x_n\}$ representing a polynomial $P(X)$ over 
 $\Bbb{Z}^+$, we are also given an additive weight function  $w~:~2^{X}\rightarrow \mathbb{N}$. The task is that 
if there exists a $k$-multilinear term  then find one with minimum weight. We call the weighted variant by 
{\sc $k$-wMlD}. We start with the definition of an arithmetic circuit.  

% \todo[inline]{define $k$-multilinear term and let us use it uniformly}
\begin{definition}
 An {\em arithmetic circuit} $C$ over a commutative ring $R$ is a simple labelled directed acyclic graph with 
its  internal nodes are labeled by $+$ or $\times$ and leaves (in-degree zero nodes) are labeled from 
$X\cup R$, where $X = \{x_1,x_2,\ldots,x_n\}$, a set of variables. There is a node of out-degree zero, 
called the root node or the output gate. The size of $C$, $s(C)$ is the number of vertices in the graph.
\end{definition}

It is well known that we can replace any arithmetic circuit $C$ with an equivalent circuit with fan-in two 
for all the internal nodes with quadratic blow up in the size. For an example, by replacing each node of in-degree 
greater than $2$, with at most $s(C)$ many nodes of the same label and in-degree $2$, we can convert a circuit $C$ 
to a circuit $C'$ of size $s(C')=s(C)^2$. So from now onwards we always assume that we are given a circuit of this 
form. We assume $W$ be the maximum weight defined by $w$.  

\begin{theorem}
\label{thm:montest}
{\sc $k$-wMlD} can be solved in time $\cO(3.8408^k2^{o(k)}s(C) n \log^2 n\cdot \log W)$. 
\end{theorem}

\begin{proof}
An arithmetic circuit $C$ over ${\mathbb Z}^+$ with all leaves labelled from $X\cup {\mathbb Z}^+$ will represent 
sum of monomials with positive integer coefficients. With each multilinear term $\Pi_{j=1}^\ell x_{i_j}$ 
we associate a set $\{x_{i_1},\ldots, x_{i_l}\}\subseteq X$. With any polynomial we can associate a family of 
subsets of $X$ which corresponds to the set of multilinear terms in it.  Since $C$ is a directed acyclic graph, there exists a topological ordering $\pi=v_1,\ldots,v_n$, such that all the nodes corresponding to variables appear before any other gate and for every directed arc $uv$  we have that $u<_{\pi} v$.   For a node $v_i$ of the circuit let $P_i(X)$ be the multivariate polynomial represented by the subcircuit containing all the nodes $w$ such that 
$w\leq_{\pi} v_i$.
%% and let $d_i\leq k$ denotes its degree. 
At every node we keep a family ${\cal F}_{v_i}^j$ of {\em $j$-multilinear term}, where 
$j\in\{1,\ldots,k\}$.  Let ${\cal F}_{v_i}=\cup_{x=1}^k {\cal F}_{v_i}^x$. 
%Our algorithm goes from left to right and computes ${\cal F}_{v_i}$ from the families previously computed. 
Given a circuit $C$, if we compute associated family of subsets of $X$ for each node  
%from left to right in the topological ordering of the nodes 
%of the circuit, 
we can answer the question of having a $k$-multilinear term of minimum weight in the polynomial computed by 
$C$.  But the size of the family of subsets could be exponential in $n$, the number of variables. That is, the size of 
 ${\cal F}_{v_i}^j$ could be ${n \choose j}$.  So instead of storing all subsets, we store a representative family for 
the associated family of subsets of each node. That is, we store 
$\widehat{{\cal F}_{v_i}^j} \subseteq^{k-j}_{minrep}{\cal F}_{v_i}^j$.  The correctness of this step follows from the definition of $k-j$-representative family. 

We make a dynamic programming algorithm to detect a multilinear monomial of order $k$ as follows. Our algorithm goes from left to right following the ordering given by $\pi$ and computes ${\cal F}_{v_i}$ from the families previously computed. The algorithm computes an appropriate representative family corresponding to each node of $C$. 
%from 
%left to right in the topological order of the nodes of $C$. 
We show that we can compute a representative family 
${\cal F}_v$ associated with any node $v$, where the number of subsets with $p$ elements in ${\cal F}_v$ is at most 
${k \choose p}2^{o(k)}\log n$. When $v$ is an input node then the associated family contains only one set. That is,  if $v$ is labelled with $x_i$  then  ${\cal F}_v=\{\{x_i\}\}$ and if  $v$ is 
labelled from ${\mathbb Z}^+$ then ${\cal F}_v=\{\emptyset\}$. When $v$ is 
not an input node, then we have two cases.

 %We make a dynamic programming algorithm to detect a multilinear monomial of order $k$ as follows. 
\begin{description}
 \item[Addition Gate.] $v=v_1+v_2$ \\
Due to the left to right computation in the topological order, we have a representative families ${\cal F}_{v_1}$ 
and ${\cal F}_{v_2}$ for $v_1$ and $v_2$ respectively, where the number of subsets with $p$ elements in 
${\cal F}_{v_1}$ as well as in ${\cal F}_{v_2}$ will be at most ${k \choose p}2^{o(k)}\log n$. So the representative 
family corresponding to $v$ will be the representative family of ${\cal F}_{v_1}\cup {\cal F}_{v_2}$. We 
partition ${\cal F}_{v_1}\cup {\cal F}_{v_2}$ based on the size of subsets in it.  Let 
${\cal F}_{v_1}\cup {\cal F}_{v_2}=\biguplus_{p\leq k} {\cal H}_p$, where ${\cal H}_p$ contains all subsets of size $p$ in 
${\cal F}_{v_1}\cup {\cal F}_{v_2}$. Note that $|{\cal H}_p|\leq 2 {k \choose p}2^{o(k)}\log n$. Now using 
Theorem~\ref{thm:repsetweighted}, we can compute all $\widehat{{\cal H}_p}\subseteq_{minrep}^{k-p} {\cal H}_p$ 
in time 
% $$\cO\left( 2^{o(k)}\log n \cdot \max_{p<k} \left\{ 2 {k \choose p} \cdot \left(\frac{k}{k-p}\right)^{k-p} \right\} \right)$$
% \textcolor{red}{
$$\cO\left( 2^{o(k)}\log^2 n \cdot \log W \cdot \sum_{p<k} \left\{ 2 {k \choose p} \cdot \left(\frac{k}{k-p}\right)^{k-p} \right\} \right)$$
where $W$ is the maximum weight defined by weight function $w$. 
% }
The above running time is upper bounded by 
% $\cO(2.851^k 2^{o(k)}\log n)$
% \textcolor{red}{
$\cO(2.851^k 2^{o(k)}\log^2 n\log W)$,  
% }
 by the similar analysis done for the 
{\sc $k$-Path} problem in~\cite{FominLS13}. We output $\bigcup_{p\leq k} \widehat{{\cal H}_p}$ as the representative family corresponding to the node $v$.

\item[Multiplication Gate.] $v=v_1\times v_2$ \\
Similar to the previous case we have a representative families ${\cal F}_{v_1}$ and ${\cal F}_{v_2}$ for $v_1$ and 
$v_2$ respectively, where the number of subsets with $p$ elements in ${\cal F}_{v_1}$ as well as in ${\cal F}_{v_2}$, is at most ${k \choose p}2^{o(k)}\log n$. Here, the representative family corresponding to $v$ will be the 
representative family of ${\cal F}_{v_1}\bullet {\cal F}_{v_2}$.  The idea is that we first get an intermediate  representative family using Corollary~\ref{cor:product_uniform} and then find its representative of this using 
Theorem~\ref{thm:repsetweighted} to get our final family.  We have that 
\[{\cal F}_{v_1}\bullet {\cal F}_{v_2}= \bigcup_{p_1,p_2} {\cal F}_{v_1}^{p_1}\bullet {\cal F}_{v_2}^{p_2},\] 
where ${\cal F}_{v_i}^{p_i}$ contains all the subsets of size $p_i$ in ${\cal F}_{v_i}$. 
We know that $|{\cal F}_{v_i}^{p_i}|\leq {k \choose p_i}2^{o(k)}\log n$. Now by using 
% \textcolor{red}{
a variant of 
% }
Corollary~\ref{cor:product_uniform}, we compute 
\lminrep{{\cal F}_{v_1}^{p_1}\bullet {\cal F}_{v_2}^{p_2} }{k-p_1-p_2} of size 
${k\choose p_1+p_2}\cdot 2^{o(k)}\cdot \log n$ for all $p_1,p_2$ such that $p_1+p_2\leq k$. 
Let $q=k-p_1-p_2$, then all these computation  can be done in time 
% $$\max_{p_1,p_2}\min_{x_1,x_2}\cO\Big(
% \frac{z(n,k)}{x_2^{p_1}(1-x_2)^{p_2}}+\frac{z(n,k) \cdot |{k \choose p_1}|}{x_1^{p_2}(1-x_1)^{q}(1-x_2)^{p_2}} +\\ 
% \frac{z(n,k) \cdot |{k\choose p_2}|}{x_1^{p_1}(1-x_1)^{q}x_2^{p_1}}+ \frac{z(n,k) \cdot (\frac{k}{q})^{q}}{x_1^p(1-x_1)^{q}}\Big).$$ 
% \textcolor{red}{
$$\sum_{p_1,p_2}\min_{x_1,x_2}\cO\Big(
\frac{z'(n,k,W)}{x_2^{p_1}(1-x_2)^{p_2}}+\frac{z'(n,k,W) \cdot |{k \choose p_1}|}{x_1^{p_2}(1-x_1)^{q}(1-x_2)^{p_2}} +\\ 
\frac{z'(n,k,W) \cdot |{k\choose p_2}|}{x_1^{p_1}(1-x_1)^{q}x_2^{p_1}}+ \frac{z'(n,k,W) \cdot (\frac{k}{q})^{q}}{x_1^p(1-x_1)^{q}}\Big).$$ 
% There will be an additional $\log n$ because of size of $F_{v_i}^{p_1}$
% }

Here, $z'(n,k,W)=2^{\cO(\frac{k}{\log\log\log(k)})}k^{\cO(1)}n\log^2 n\cdot \log W$. The above running time is upper bounded by 
% $\cO(3.841^k2^{o(k)}n \log n)$ 
% \textcolor{red}
% {
$\cO(3.8408^k2^{o(k)}k^{\cO(1)}n \log^2 n\cdot \log W)$
% } 

%$z(n,k)=2^{\cO(\frac{k}{\log\log\log(k)})}k^{\cO(1)}n\log n$.

Now let 
${\cal F}= \bigcup_{p_1,p_2} \widehat{{\cal F}_{v_1}^{p_1}\bullet {\cal F}_{v_2}^{p_2}}= \uplus_p {\cal H}_p$, 
where $\uplus_p {\cal H}_p$ is the partition of ${\cal F}$ based on size of subsets. It is easy to see that 
$|{\cal H}_p|\leq k {k \choose p}2^{o(k)}\log n$.   
Now using Theorem~\ref{thm:repsetweighted} we can compute $\widehat{{\cal H}_p}\subseteq_{minrep}^{k-p} {\cal H}_p$ 
for all $p\leq k$ together in time 
% $$\cO\left( k^2\cdot 2^{o(k)}\log n \cdot \max_{p\leq k} \left\{ {k \choose p} \cdot \left(\frac{k}{k-p}\right)^{k-p} \right\} \right).$$
% \textcolor{red}
% {
$$\cO\left( k^2\cdot 2^{o(k)}\log^2 n\cdot \log W \cdot \sum_{p\leq k} \left\{ {k \choose p} \cdot \left(\frac{k}{k-p}\right)^{k-p} \right\} \right).$$
% }
The above running time is upper bounded by 
% $\cO(2.851^k 2^{o(k)}\log n)$
% \textcolor{red}
% {
$\cO(2.851^k 2^{o(k)}k^2\log^2 n\cdot \log W)$.
% } 
We output $\bigcup_{p\leq k} \widehat{{\cal H}_p}$ 
as the representative family corresponding to the node $v$. 
\end{description} 
Now we output {\em Yes} and a minimum weight set of size $k$ (if exists) among the representative family corresponding to the root node. 
% contains a subset of size $k$. 
Since there are $s(C)$ nodes in $C$, the total running time is bounded by 
% $\cO(3.841^k2^{o(k)}s(C)n \log n)$
% \textcolor{red}
% {
$\cO(3.8408^k2^{o(k)}s(C)n \log^2 n\cdot \log W)$.
% } 
This completes the proof.    
\end{proof}

\subsection{Matroidal Multilinear Monomial Detection}
In this section we extend the {\sc $k$-wMlD} problem to a matroidal version and design an algorithm for this. 
The problem is defined as follows. 

\defparproblem{{\sc Matroidal Multilinear Monomial Detection ($k$-wMMlD)}}{An arithmetic circuit $C$ over variables $X = \{x_1,x_2,\ldots,x_n\}$ representing a polynomial $P(X)$ over $\Bbb{Z}$, a linear matroid  
\mat{} where  the ground set $E=X$ with its representation matrix $A_M$ and an additive weight function $w~:~2^{X}\rightarrow \mathbb{N}$.}
{$k$}
{Does  $P (X)$ construed as a sum of monomials contains a multilinear monomial $Z$ of degree $k$ such that $Z\in {\cal I}$? If yes find a minimum weighted such $Z$.}

% \todo[inline]{Define $2^X$ in prelim.}
Our main theorem of this section is as follows.  The proof of this theorem is along the lines of 
Theorem~\ref{thm:montest}. The only difference is that we compute representative with respect to the given matroid.   

%\begin{theorem}
% Let $P[x_1,x_2,\ldots,x_n]$ be a polynomial over ${\mathbb Z}$ 
%% ****${\mathbb Q}$ (field of rational numbers)**** 
%represented by an arithmetic circuit $C$ of size $s(C)$ with $+$ gates and $\times$ gates of fan-in two
%and all the leaves are labelled from $X=\{x_1,x_2,\ldots,x_n\}$ and ${\mathbb Z}^+$. Let \mat{} be a 
%linear matroid of rank $k$ and  $f~:~X\rightarrow E$, be a one to one function. Then in time 
%$\cO\left( k^{\omega}2^k\left(2^{(\omega-1)}+1\right)^k + k^\omega 2^{k(\omega-1)} 3^k \right)$ we can 
%test whether there exist a multilinear monomial $\Pi_{j=1}^k x_{i_j}$ of degree $k$ in $P[x_1,x_2,\ldots,x_n]$ 
%such that $\{f(x_{i_1}),\ldots,f(x_{i_k})\}\in\I$. 
%\end{theorem}

\begin{theorem}
{\sc $k$-wMMlD} can be solved in time $\cO(7.7703^k k^\omega s(C))$. 
%$\cO\left( k^{\omega}2^k\left(2^{(\omega-1)}+1\right)^k + k^\omega 2^{k(\omega-1)} 3^k ?????\right)$. 
\end{theorem}
\begin{proof} We outline a proof here.  Let  $\pi=v_1,\ldots,v_n$ be a topological ordering of $C$ such that all the nodes corresponding to variables appear before any other gate and for every directed arc $uv$  we have that $u<_{\pi} v$. As in Theorem~\ref{thm:montest}, at every node we keep a family 
${\cal F}_{v_i}^j$ of {\em $j$-multilinear term} that are also members of $\cal I$, where $j\in\{1,\ldots,k\}$.  
Let ${\cal F}_{v_i}=\cup_{x=1}^k {\cal F}_{v_i}^x$.  So ${\cal F}_v \subseteq {\cal I}$. 
%With each multilinear monomial $\Pi_{j=1}^lx_{i_j}$ we associate a subset $\{f(x_{i_1}),\ldots, f(x_{i_l})\}$ of $E$. 
%With any polynomial we can associate a family ${\cal F}$ of subsets of $E$ such that $S\in{\cal F}$ iff 
%$f^{-1}(S)$ is a multilinear monomial in the polynomial. 
% Given a circuit $C$, if we compute 
% associated family of subsets of $E$ for each node from left to right in the topological order of the nodes of $C$, 
% we can answer the question of having a multilinear monomial $M_1$ such that $f(M_1)\in\I$ in the polynomial 
% rpresented by $C$. But the size of the family of subsets could be exponential in $n$, the number of variables. 
% So 
We process the nodes from left to right and keep $ \widehat{{\cal F}_{v_i}^j}\subseteq_{minrep}^{k-j} {\cal F}_{v_i}^j$ of size ${k \choose p}$. 
%As like in the Theorem~\ref{thm:montest} we compute and store a representative family for 
%the associated family of subsets corresponding to the polynomial computed by each node 
%from left to right in the topological order of the nodes of $C$. 
% We compute the representative family corresponding to each node left to right. 
%We show that we can compute a representative family ${\cal F}_v$ associated with any node $v$, 
%where the number of subsets with $p$ elements in ${\cal F}_v$ is at most 
%${k \choose p}$. 

When $v$ is an input node then the associated family contains only one set. That is,  if $v$ is labelled with $x_i$  
and $\{x_i\}\in {\cal I}$ then  ${\cal F}_v=\{\{x_i\}\}$ and if  $v$ is 
labelled from ${\mathbb Z}^+$ then ${\cal F}_v=\{\emptyset\}$. When $v$ is 
not an input node, then we have two cases. 
%When $v$ is an input node the associated family contains only one set 
%(i.e, ${\cal F}_v=\{\{f(x_i)\}\}$ if $v$ is labelled with $x_i$ and ${\cal F}_v=\{\emptyset\}$ if $v$ is 
%labelled from ${\mathbb Z}^+$). When $v$ is not an input node, then we have two cases. 

\begin{description}
 \item[Addition Gate.] $v=v_1+v_2$ \\
%Due to the left right computation in the topological order, we have a representative families ${\cal F}_{v_1}$ 
%and ${\cal F}_{v_2}$ for $v_1$ and $v_2$ respectively, where the number of subsets with $p$ elements in 
%${\cal F}_{v_1}$ as well as in ${\cal F}_{v_2}$ will be  at most ${k \choose p}$. So the representative family 
%corresponding to $v$ will be the representative family of ${\cal F}_{v_1}\cup {\cal F}_{v_2}$. Now we partition 
%${\cal F}_{v_1}\cup {\cal F}_{v_2}$ based on the 
%size of subsets. Let ${\cal F}_{v_1}\cup {\cal F}_{v_2}=\uplus_{p<k} {\cal H}_p$, where ${\cal H}_p$ contains all 
%subsets of size $p$ in ${\cal F}_{v_1}\cup {\cal F}_{v_2}$. Note that $|{\cal H}_p|\leq 2 {k \choose p}$. 
%Now using Theorem~\ref{thm:repsetlovaszweighted} we can compute all $\widehat{{\cal H}_p}\subseteq_{rep}^{k-p} {\cal H}_p$ 
%in time 
%$$\cO\left( 2\sum_{p<k} \left\{ {k\choose p}{k \choose p}p^{\omega} + {k\choose p} {k\choose p}^{w-1} \right\} \right)$$
%The above running time is upper bounded by $\cO(4^kp^{w}+2^{wk})$. We output $\bigcup \widehat{{\cal H}_p}$ 
%as the representative family corresponding to the node $v$.
Due to the left to right computation in the topological order, we have a representative families ${\cal F}_{v_1}$ 
and ${\cal F}_{v_2}$ for $v_1$ and $v_2$ respectively, where the number of subsets with $p$ elements in 
${\cal F}_{v_1}$ as well as in ${\cal F}_{v_2}$ will be at most ${k \choose p}$. So the representative 
family corresponding to $v$ will be the representative family of ${\cal F}_{v_1}\cup {\cal F}_{v_2}$. We 
partition ${\cal F}_{v_1}\cup {\cal F}_{v_2}$ based on the size of subsets in it.  Let 
${\cal F}_{v_1}\cup {\cal F}_{v_2}=\biguplus_{p\leq k} {\cal H}_p$, where ${\cal H}_p$ contains all subsets of size $p$ in ${\cal F}_{v_1}\cup {\cal F}_{v_2}$. Note that $|{\cal H}_p|\leq 2 {k \choose p}$. Now using Theorem~\ref{thm:repsetlovaszweighted} we can compute all $\widehat{{\cal H}_p}\subseteq_{minrep}^{k-p} {\cal H}_p$ 
in time 
$$\cO\left( 2\sum_{p\leq k} \left\{ {k\choose p}{k \choose p}p^{\omega} + {k\choose p} {k\choose p}^{\omega-1} \right\} \right).$$ 
The above running time is upper bounded by $\cO(4^kp^{\omega}k+2^{\omega k}k)$. We output $\bigcup_{p\leq k} \widehat{{\cal H}_p}$ 
as the representative family corresponding to the node $v$.

\item[Multiplication Gate.] $v=v_1\times v_2$ \\
Similar to the previous case we have a representative families ${\cal F}_{v_1}$ and ${\cal F}_{v_2}$ for $v_1$ and 
$v_2$ respectively, where the number of subsets with $p$ elements in ${\cal F}_{v_1}$ as well as in ${\cal F}_{v_2}$, is at most ${k \choose p}$. Here, the representative family corresponding to $v$ will be the 
representative family of ${\cal F}_{v_1}\bullet {\cal F}_{v_2}$.  
%The idea is that we first get an intermediate  representative family using Corollary~\ref{cor:product_uniform} and then find its representative of this using 
%Theorem~\ref{thm:repsetweighted} to get our final family.  
We have that 
\[{\cal F}_{v_1}\bullet {\cal F}_{v_2}= \bigcup_{p_1,p_2} {\cal F}_{v_1}^{p_1}\bullet {\cal F}_{v_2}^{p_2},\] 
where ${\cal F}_{v_i}^{p_i}$ contains all the subsets of size $p_i$ in ${\cal F}_{v_i}$. 
We know that $|{\cal F}_{v_i}^{p_i}|\leq {k \choose p_i}$.  Now by using Corollary~\ref{cor:product_general_matroid}, we 
can compute \lminrep{{\cal F}_{v_1}^{p_1}\bullet {\cal F}_{v_2}^{p_2} }{k-p_1-p_2} of size 
${k\choose p_1+p_2}$ for all $p_1,p_2$ together in time 
$\cO\left( k^{\omega}2^k\left(2^{(\omega-1)}+1\right)^k + k^\omega 2^{k(\omega-1)} 3^k \right).$

 Now let 
${\cal F}= \bigcup_{p_1,p_2} \widehat{{\cal F}_{v_1}^{p_1}\bullet {\cal F}_{v_2}^{p_2}}= \uplus_p {\cal H}_p$, 
where $\uplus_p {\cal H}_p$ is the partition of ${\cal F}$ based on the size of subsets. It is easy to see that 
$|{\cal H}_p|\leq k {k \choose p}$.   
Now using Theorem~\ref{thm:repsetlovaszweighted} we can compute  
$\widehat{{\cal H}_p}\subseteq_{minrep}^{k-p} {\cal H}_p$ for all $p\leq k$ together in time 
$$\cO\left( k\sum_{p\leq k} \left\{ {k\choose p}{k \choose p}p^{\omega} + {k\choose p} {k\choose p}^{\omega -1} \right\} \right)$$
The above running time is upper bounded by $\cO(4^kk^2p^{\omega}+2^{ \omega k}k^2)$. We output $\bigcup_{p\leq k} \widehat{{\cal H}_p}$ 
as the representative family corresponding to the node $v$.

\end{description} 
Now we output {\em Yes} and a minimum weight set of size $k$ (if exists) among the representative family corresponding to the root node. 
% contains a subset of size $k$. 
Since there are $s(C)$ nodes in $C$, the total running time is bounded by  $\cO\left( k^{\omega}2^k\left(2^{(\omega-1)}+1\right)^k s(C)+ k^\omega 2^{k(\omega-1)} 3^k s(C)\right)$. This completes the proof.    
%Now we output {\em Yes}, if the representative family corresponding to root node contains a subset of size $k$. 
%Since there are $s(C)$ nodes in $C$, the total running time is bounded by 
%$\cO\left( k^{\omega}2^k\left(2^{(\omega-1)}+1\right)^k s(C)+ k^\omega 2^{k(\omega-1)} 3^k s(C)\right)$
% $O(3.841^k2^{o(k)}s(C)\log n)$   
\end{proof}

%%%%%%%%%%

\section{Application II: Dynamic Programming over graphs of bounded treewidth}
%\subsection{Dynamic Programming over graphs of bounded  }
% \input{applications.tex}
%!TEX root = main.tex

 %\defparproblem{ {\sc  Steiner Tree} }{An undirected $G$ together with a tree-decomposition of width $w$, $T\subseteq V(G)$ and positive integers $k$.}{$w$ }
 %{Is there a set $X\subseteq V(G)$ of size at most $k$ such that $T\subseteq X$ and $G[ X]$ is connected?}
 
In this section we discuss deterministic algorithms for ``connectivity problems'' such as 
% {\sc Hamiltonian Path}, 
{\sc Steiner Tree}, {\sc Feedback Vertex Set}  
parameterized by the treewidth of the input graph. The algorithms are based on Theorem~\ref{thm:repsetlovaszweighted} 
and Corollary~\ref{cor:product_general_matroid}. The idea of designing deterministic algorithms for connectivity problems parameterized by the treewidth of the input graph based on fast computation of representative families was outlined in~\cite{FominLS13}. Here, we show how we can speed the method described in~\cite{FominLS13} using the fast computation of representative families for product families coming from a graphic matroid. The method  described in this section gives the fastest known deterministic algorithms for most  the connectivity problems parameterized by the treewidth. 
%
% on graphic matroids to take care of connectivity constraints. 
%The approach is  generic and can be used whenever all the relevant information about a ``partial solution'' can be 
%encoded as an independent set of a specific linear matroid. 
We exemplify the methods on  {\sc Steiner Tree} and  {\sc Feedback Vertex Set}. 
%problem. 
%We provide a cookbook for this approach in Section~\ref{subsect:frameworktw}.

\subsection{Treewidth}
\label{subsect:twprelim}
%\noindent 
Let $G$ be a graph.  A {\em tree-decomposition} of a graph $G$ is a pair 
$(\mathbb{T},\mathcal{ X}=\{X_{t}\}_{t\in V({\mathbb T})})$ such that
\begin{itemize}
\setlength\itemsep{-1mm}
\item $\cup_{t\in V(\mathbb{T})}{X_t}=V(G)$,
\item for every edge $xy\in E(G)$ there is a $t\in V(\mathbb{T})$ such that  $\{x,y\}\subseteq X_{t}$, and 
\item for every  vertex $v\in V(G)$ the subgraph of $\mathbb{T}$ induced by the set  $\{t\mid v\in X_{t}\}$ is connected.
\end{itemize}

The {\em width} of a tree decomposition is $\max_{t\in V(\mathbb{T})} |X_t| -1$ and the {\em treewidth} of $G$ 
is the  minimum width over all tree decompositions of $G$ and is denoted by $\tw(G)$. 
%  If in the definition of {\em tree decomposition}  
%the tree $\mathbb T$ is a path, then the corresponding  tree-decomposition is the {\em path decomposition}. 
% The {\em pathwidth} of a graph of $G$  is the  minimum width over all path-decompositions of $G$ and is denoted by $\pw(G)$. 

A tree
decomposition  $(\mathbb{T},\mathcal{ X})$ is called a {\em nice tree
decomposition} if $\mathbb{T}$ is a tree rooted at some node $r$ where
$X_{r}=\emptyset$, each node of $\mathbb{T}$ has at most two children, and each
node is of one of the following kinds:
\begin{enumerate}
\item {\bf Introduce node}: a node $t$ that has only one child $t'$ where $X_{t}\supset X_{t'}$ and  $|X_{t}|=|X_{t'}|+1$.
%, and such that $t'$ is not an introduce node.
\item {\bf  Forget node}: a node $t$ that has only one child $t'$  where $X_{t}\subset X_{t'}$ and  $|X_{t}|=|X_{t'}|-1$.
%, and such that $t'$ is not a forget node.
\item {\bf Join node}:  a node  $t$ with two children $t_{1}$ and $t_{2}$ such that $X_{t}=X_{t_{1}}=X_{t_{2}}$.
\item {\bf Base node}: a node $t$ that is a leaf of $\mathbb T$, is different than the root, and $X_{t}=\emptyset$. 
\end{enumerate}
Notice that, according to the above definition, the root $r$ of $\mathbb{T}$ is
either a forget node or a join node. It is well known that any tree
decomposition of $G$ can be transformed into a nice tree decomposition
maintaining the same
width in linear time~\cite{Kloks94}. We use $G_t$ to denote the graph induced  by the
vertex set  $\cup_{t'}X_{t'}$, where $t'$ ranges over all descendants of $t$,
including $t$. By $E(X_t)$ we denote the edges present in $G[X_t]$.
We use $H_t$ to denote the graph on vertex set $V(G_t)$ and the edge set 
$E(G_t)\setminus E(X_t)$.  For clarity of presentation we use the term nodes to refer to the vertices of the tree 
$\mathbb T$.

\subsection{{\sc Steiner Tree} parameterized by treewidth}
\label{subsect:steinertreetw}
%\medskip
%Minimum Weight 
The problem we study in this section is defined below.

\medskip
\begin{center} 
\fbox{\begin{minipage}{0.96\textwidth}
\noindent{\sc Steiner Tree} \\ %\hfill {\bf Parameter:} $\tw$ \\
\noindent {\bf Input}: An undirected graph $G$ 
%together with a tree-decomposition $(\mathbb{T},\mathcal{ X})$ of width $\tw$, $T\subseteq V(G)$  \\
with a set of terminals $T\subseteq V(G)$, and  a  weight\\
\noindent{\phantom{{\em Input}:}} 
 function $w:E(G)\rightarrow \mathbb{N}$.\\
\noindent{\bf Task}: Find a subtree  in $G$ of minimum weight spanning all vertices of $T$. 
%and $G[ X]$ is \\
%\noindent{\phantom{{\em  Question}:}} connected?
\end{minipage}}
\end{center}
\medskip

%A $T\subseteq V(G)$ (defined in the problem) is called {\em terminals}.  

Let $G$ be an input  graph of the {\sc Steiner Tree} problem. Throughout this section, we say that 
$E'\subseteq E(G)$ is a {\em solution} if the subgraph induced on this edge set is connected and 
it contains all the terminal vertices. We call $E'\subseteq E(G)$ an {\em optimal solution} if $E'$ is a solution 
of the  minimum weight. % tree that contains all the terminal vertices.   
Let $\mathscr{S}$ be a family of edge subsets such that every edge subset corresponds to an optimal solution. 
That is, 
$$\mathscr{S}=\{E'\subseteq E(G)~|~E' \mbox{ is an optimal solution}\}.$$ 
Observe that any edge set in $\mathscr{S}$ induces a forest. We start with few definitions that will be useful in explaining the 
algorithm. Let $(\mathbb{T},\mathcal{ X})$  be a tree decomposition of $G$ of width $\tw$. Let $t$ be a node of 
$V(\mathbb{T})$. By $\mathcal{S}_t$ we denote the family of edge subsets  of $E(H_t)$, 
$\{E'\subseteq E(H_t)~|~G[E'] \mbox{ is a forest}\}$,  
that satisfies the following properties.  
\begin{itemize}
\item Either $E'$ is  a solution tree (that is, the subgraph induced on this edge set is connected and it contains all the 
terminal vertices); or 
\item every vertex of $(T\cap V(G_t))\setminus X_t$ is incident with  some edge from  $E'$,  and every connected 
component of the graph  induced by $E'$  contains a  vertex from $X_t$.
\end{itemize}

We call $\mathcal{S}_t$ a \emph{family of partial solutions} for $t$. We denote by  $K^t$   a complete graph on the 
vertex set $X_t$. For an edge subset $E^* \subseteq E(G)$ and   bag $X_t$ corresponding to a node $t$, we define the 
following. 
\begin{enumerate}
\item  Set  $\partial^t(E^*)= X_t\cap V(E^*)$,  the set of endpoints  of $E^*$ in $X_t$.
\item Let $G^*$ be the subgraph of $G$ on the vertex set $V(G)$ and the edge set $E^*$. Let $C_1',\ldots ,C_\ell'$ 
be the connected components of  $G^*$ such that for all $i\in [\ell]$,  $C_i'\cap X_t\neq \emptyset$. 
Let $C_i=C_i'\cap X_t$. Observe that $C_1,\ldots,C_\ell$ is a partition of $\partial^t(E^*)$. 
By $F(E^*)$ we denote a forest $\{Q_1,\ldots,Q_\ell\}$ where each $Q_i$ is an arbitrary spanning tree of $K^t[C_i]$. 
For an example, since 
$K^t[C_i]$ is a complete graph we could take $Q_i$ as a star. The purpose of  $F(E^*)$ is to keep track  
for the vertices in $C_i$ whether they were in the same connected component of $G^*$. 
\item We define $w(F(E^*))=w(E^*)$. 
\end{enumerate}

Let ${\cal A}$ and ${\cal B}$ be two family of edge subsets of $E(G)$, then we define 
$${\cal A}\diamond {\cal B}=\{E_1\cup E_2\;|\;E_1\in {\cal A}\wedge E_2\in {\cal B}\wedge E_1\cap E_2=\emptyset \wedge G[E_1\cup E_2] \mbox{ is a forest}\}.$$

%Our description of the algorithm  slightly deviates from the usual table look-up based expositions of dynamic 
%programming algorithms on graphs of bounded treewidth.
With every node $t$ of  $\mathbb T$,   we associate a subgraph of $G$. In our case it will be $H_t$. For every node 
$t$, 
%rather than keeping a table, 
we keep  a family of partial solutions for the graph $H_t$. That is, for every 
optimal solution $L\in \mathscr{S}$ and its intersection $L_t=E(H_t)\cap L$ with the graph $H_t$, we have some  
partial solution in the family that is ``as good as $L_t$''. More precisely, we have some partial solution, say 
$\hat{L}_t$ in our family such that $\hat{L}_t\cup L_R$ is also an optimum solution for the whole graph, where  
$L_R=L\setminus L_t$. As we move from one node $t$ in the decomposition tree to the next node $t'$ the graph $H_t$ 
changes to $H_{t'}$, and so does the set of partial solutions. The algorithm updates its set of partial solutions 
accordingly. Here matroids come into play: in order to bound the size of the family of partial solutions that the 
algorithm stores at each node we employ Theorem~\ref{thm:repsetlovaszweighted} and 
Corollary~\ref{cor:product_general_matroid} for graphic matroids. More details 
are given in the proof of the following theorem, which is  the main result of this section.

%{\bf Need to check if the weight enters the running time of the theorem below as well as weighted representative theorem.}

\begin{theorem}\label{thm:steinertree_DP}
Let $G$ be an $n$-vertex graph given together with its tree decomposition of with $\tw$. Then 
{\sc Steiner Tree} on $G$ can be solved in time  $\cO\left( \left(1+ 2^{\omega-1}\cdot 3\right)^{\tw} 
 {\tw}^{\cO(1)}n \right)$.
%$\cO\left(\left( \left(2^{\omega}+3\right)^{\tw}+ \left( 1+ 2^{\omega-1}\cdot 3\right)^{\tw}\right) {\tw}^{O(1)}n \right).$
% $\cO((1+2^{\omega+1})^{\tw} \tw^{\cO(1)}n)$.
\end{theorem}
\begin{proof}
% We first outline an algorithm with running time $\cO((1+2^{\omega+1})^{\tw} \tw^{\cO(1)}n^2)$ for a simple 
% exposition. Later we point out how we can remove the extra factor of $n$ at the cost of a factor polynomial in $\tw$. 

For every node $t$ of  $\mathbb T$ and subset $Z\subseteq X_t$,  we store a family of edge subsets  
$\widehat{\mathcal{S}}_t[Z]$ of $H_t$  satisfying the following correctness invariant.
\begin{quote}
%\item 
{\bf Correctness Invariant:} For every $L\in \mathscr{S}$ we have the following. 
Let $L_t=E(H_t)\cap L$, $L_R=L\setminus L_t$,  and $Z=\partial^t(L)$. Then there exists 
$\hat{L}_t\in \widehat{\mathcal{S}}_t[Z]$ such that $w(\hat{L}_t)\leq w(L_t)$, $\hat{L}=\hat{L}_t\cup L_R$ is a 
solution, and $\partial^t(\hat{L})=Z$. Observe that since $w(\hat{L}_t)\leq w(L_t)$ and $L\in \mathscr{S}$, we have 
that $\hat{L} \in \mathscr{S}$.
\end{quote}

We process the nodes of the tree $\mathbb T$  from base nodes to the root node while doing the dynamic programming. 
Throughout the process we maintain the correctness invariant, which will prove the correctness of the algorithm. 
However, our main idea is to use representative sets  to obtain   $\widehat{\mathcal{S}}_t[Z]$ of small size. 
That is, given the set $\widehat{\mathcal{S}}_t[Z]$ (as a product of two families ${\cal A}$ and ${\cal B}$, i.e
 $\widehat{\mathcal{S}}_t[Z]= {\cal A}\diamond {\cal B}$) that satisfies the correctness invariant,  
we use 
% Theorem~\ref{thm:repsetlovaszweighted} or 
Corollary~\ref{cor:product_general_matroid} to obtain a subset 
$\widehat{\mathcal{S}}_t'[Z]$ of $\widehat{\mathcal{S}}_t[Z]$ that also satisfies the  correctness invariant 
and has size upper bounded by $2^{|Z|}$ in total.  More precisely, the number of partial solutions with $i$ 
%number of 
connected components in $\widehat{\mathcal{S}}_t'[Z]$ is upper bounded by ${|Z| \choose |Z|-i}={|Z| \choose i}$. 
Thus, we maintain the following size invariant.

\begin{quote}
%\item 
{\bf Size Invariant:} After  node $t$ of $\mathbb T$ is processed by the algorithm,  for every $Z\subseteq X_t$ 
we have that 
 $|\widehat{\mathcal{S}}_t[Z,i]|\leq {|Z| \choose i}$, where $\widehat{\mathcal{S}}_t[Z,i]$ is the partial 
solutions with $i$ connected components in $\widehat{\mathcal{S}}_t[Z]$. 
\end{quote}
 
 The main ingredient of the dynamic programming algorithm for {\sc Steiner Tree} is the use of 
Theorem~\ref{thm:repsetlovaszweighted} and Corollary~\ref{cor:product_general_matroid} to compute 
$\widehat{\mathcal{S}}_t[Z]$ maintaining the size invariant. The next  lemma shows how to implement it.

\begin{lemma}[Product Shrinking Lemma]
\label{lem:sizeinvariant_2}
Let $t$ be a node of  $\mathbb T$, and let  $Z\subseteq X_t$ be a set of size $k$. 
Let ${\cal P}$ and ${\cal Q}$ be two family of edge sets of $H_t$. Furthermore, let  
$\widehat{\mathcal{S}}_t[Z]={\cal P}\diamond {\cal Q}$ be the family of edge subsets of $H_t$ 
satisfying the correctness invariant. If the number of edge sets with $i$ connected components in ${\cal P}$ 
as well as in ${\cal Q}$ is bounded by ${k+c \choose i}$ where $c$ is some fixed constant, 
% If $|\widehat{\mathcal{S}}_t[Z]|=\ell$, 
then in time $\cO\left( k^{\omega}\left(2^{\omega}+2\right)^k n + k^\omega 2^{k(\omega-1)} 3^k n \right)$
% $\cO\left(2^{k(\omega-1)} k^{\cO(1)} \ell \cdot n\right)$ 
we can compute $\widehat{\mathcal{S}}_t'[Z] \subseteq \widehat{\mathcal{S}}_t[Z]$ 
satisfying correctness and size invariants. 
\end{lemma}
 \begin{proof}
 We start by associating a matroid with  node $t$ and the set $Z\subseteq X_t$  as follows. We consider 
a graphic matroid \mat{} on $K^t[Z]$. Here, the element set $E$ of the matroid is  the edge set $E(K^t[Z])$ and 
the family of independent sets $\cal I$ consists of  forests of  $K^t[Z]$. 
 
Let ${\cal P}=\{A_1^t,\ldots,A_\ell^t\}$ and ${\cal Q}=\{B_1^t,\ldots,B_{\ell'}^t\}$. 
Let ${\cal L}_1=\{ F(A_1^t),\ldots,F(A_\ell^t)\}$ and ${\cal L}_2=\{ F(B_1^t),\ldots,F(B_{\ell'}^t)\}$ be 
the set of forests in $K^t[Z]$ corresponding to the edge subsets in ${\cal P}$ and ${\cal Q}$ respectively. For 
$i\in \{1,\ldots,k-1\}$ and $r\in \{1,2\}$, let ${\cal L}_{r,i}$ be the family of forests of ${\cal L}_r$ with $i$ edges.  
% (i.e, ${\cal N}_i$ corresponds to $\widehat{\mathcal{S}}_t[Z,k-i]$). 
% For each of family ${\cal N}_i$ we apply 
% Theorem~\ref{thm:repsetlovaszweighted} and compute its min $(k-1-i)$-representative. That is, 
% $$\widehat{{\cal N}}_i \subseteq_{minrep}^{k-1-i}{\cal N}_i.$$
Now we apply Corollary~\ref{cor:product_general_matroid} and find 
$\widehat{{\cal L}_{1,i}\bullet {\cal L}_{2,j}} \subseteq_{minrep}^{k-1-i-j} {\cal L}_{1,i}\bullet {\cal L}_{2,j}$ 
of size ${k-1 \choose i+j}$ for all $i,j\in [k]$. 
Let $\widehat{\mathcal{S}}_t'[Z,k-d] \subseteq \widehat{\mathcal{S}}_t[Z,k-d]$ be such that for every 
$E^t\in \widehat{\mathcal{S}}_t'[Z,k-d] $ we have that 
$F(E^t)\in \bigcup_{i+j=d}\widehat{{\cal L}_{1,i}\bullet {\cal L}_{2,j}}$.  
(Note that $F(E^t)$ has $d$ edges if and only if  $G[E^t]$ have $k-d$ connected components). 
Let $\widehat{\mathcal{S}}_t'[Z]= \cup_{j=1}^k \widehat{\mathcal{S}}_t'[Z,j]$. 
By Corollary~\ref{cor:product_general_matroid}, 
$|\widehat{\mathcal{S}}_t'[Z,k-d]| \leq k {k-1 \choose d}\leq {k \choose k-d}$, and hence 
$\widehat{\mathcal{S}}_t'[Z]$ maintains the size invariant.  
Now we show that the $\widehat{\mathcal{S}}_t'[Z]$ maintains the correctness invariant. 
 
Let $L\in \mathscr{S}$ and let $L_t=E(H_t)\cap L$, $L_R=L\setminus L_t$ and $Z=\partial^t(L)$. 
Then there exists $E_j^t \in \widehat{\mathcal{S}}_t[Z]$ such that $w(E_j^t)\leq w(L_t)$, 
$\hat{L}=E_j^t \cup L_R$ is an optimal solution and $\partial^t(\hat{L})=Z$. 
Since $\widehat{\mathcal{S}}_t[Z]={\cal P}\diamond {\cal Q}$, there exists $A^t_{j_1}\in {\cal P}$ and 
$B^t_{j_2}\in {\cal Q}$ such that $E^t_j=A^t_{j_1}\cup B^t_{j_2}$. 
Observe that $G[E^t_j], G[A^t_{j_1}]$ and $G[B^t_{j_2}]$ form forests.    
Consider the forests $F(A^t_{j_1})$ and $F(B^t_{j_2})$. 
Suppose $|F(A^t_{j_1})|=i_1$ and $|F(B^t_{j_2})|=i_2$, then $F(E_j^t)\in {\cal L}_{1,i_1}\bullet {\cal L}_{1,i_2}$. 
This is because, if $F(E_j^t)$ contain a cycle, then corresponding to that cycle we can get a cycle 
in $G[E^t_j]$, which is a contradiction. Now let $F(L_R)$ be the forest corresponding to $L_R$ 
with respect to the bag $X_t$. Since $\hat{L}$ is a solution, we have that $F(E_j^t)\cup F(L_R)$ is a spanning 
tree in  $K^t[Z]$. Since $\widehat{{\cal L}_{1,i_1}\bullet {\cal L}_{2,i_2} }\subseteq_{minrep}^{k-1-i_1-i_2}
{\cal L}_{1,i_1}\bullet {\cal L}_{2,i_2} $, we have that there exists 
a forest $F(E_h^t) \in \widehat{{\cal L}_{1,i_1}\bullet {\cal L}_{2,i_2} }$ 
such that $w(F(E_h^t)) \leq w(F(E_i^t)) $ and $F(E_h^t) \cup F(L_R)$ 
is a spanning tree in  $K^t[Z]$. 
 Thus, we know that $E_h^t \cup L_R$ is an optimum solution and $E_h^t \in \widehat{\mathcal{S}}_t'[Z]$. 
This proves that $\widehat{\mathcal{S}}_t'[Z]$ maintains the correctness invariant.   
 
 The running time to compute $\widehat{\mathcal{S}}'_t[Z]$ is, 
 \begin{eqnarray*}
 \cO\left( k^{\omega}\left(2^{\omega}+2\right)^k n + k^\omega 2^{k(\omega-1)} 3^k n \right). 
 \end{eqnarray*}
 For a given edge set we also need to compute the forest and that can take $\cO(n)$ time. 
 \end{proof}

% In our algorithm the size of $\widehat{\mathcal{S}}_t[Z]$ can grow larger than $2^{|Z|}$ in intermediate steps 
% but it will be at most $4^{|Z|}$ and thus we can use Shrinking Lemma (Lemma~\ref{lem:sizeinvariant}) to reduce its size efficiently. 

We now return to the dynamic programming algorithm over the tree-decomposition $(\mathbb{T},\mathcal{ X})$ of $G$ 
and prove that it maintains the correctness invariant. We assume that $(\mathbb{T},\mathcal{ X})$ is a nice 
tree-decomposition of $G$. By $\widehat{\mathcal{S}}_t$ we denote 
$\cup_{Z\subseteq X_t} \widehat{\mathcal{S}}_t[Z]$ (also called a \emph{representative family of partial solutions}). 
 We show how $\widehat{\mathcal{S}}_t$  is obtained by doing dynamic programming from base node to the root node.

\paragraph{Base node $t$.}  Here the graph $H_t$ is empty and thus we take $\widehat{\mathcal{S}}_t=\emptyset$.  

\paragraph{Introduce node $t$ with child $t'$.} Here, we know that $X_{t}\supset X_{t'}$ and  $|X_{t}|=|X_{t'}|+1$. 
Let $v$ be the vertex in $X_{t}\setminus X_{t'}$. Furthermore observe that $E(H_t)=E(H_{t'})$ and $v$ is degree zero 
vertex in $H_t$. Thus the graph $H_t$ only differs from $H_{t'}$ at a isolated vertex $v$. Since we have not added 
any edge to the new graph, the family of solutions, which contains edge-subsets, does not change. Thus, we take 
$\widehat{\mathcal{S}}_t=\widehat{\mathcal{S}}_{t'}$. Formally, we take 
$\widehat{\mathcal{S}}_t[Z]=\widehat{\mathcal{S}}_{t'}[Z\setminus \{v\}]$.   
Since, $H_t$ and $H_{t'}$ have  same set of edges the invariant is vacuously maintained. 
 
 \paragraph{Forget node $t$ with child $t'$.} Here we know $X_{t}\subset X_{t'}$ and  $|X_{t}|=|X_{t'}|-1$. 
Let $v$ be the vertex in $X_{t'}\setminus X_{t}$. Let ${\cal E}_v[Z]$ denote the set of edges between $v$ and 
the vertices in $Z\subseteq X_t$. Observe that $E(H_t)=E(H_{t'})\cup {\cal E}_v[X_t]$. Before we define  things 
formally, observe that in this step the graphs $H_t$ and $H_{t'}$ differ by at most $\tw$ edges - the edges with 
one endpoint in $v$ and the other in $X_t$. We go through every possible way an optimal solution can intersect with 
these newly added edges. 
% The idea is that for every edge subset in our family of partial solutions we make several 
% new partial solutions, one each for every  subset of newly added edges.  
Let ${\cal P}_v[Z]=\{Y\;|\; Y\subseteq {\cal E}_v[Z]\}$. Then the new set of partial 
solutions is defined as follows.  
%   $$ \widehat{\mathcal{S}}_t[Z]= \widehat{\mathcal{S}}_{t'}[Z] \bigcup_{X\subseteq {\cal E}_v[Z]} \widehat{\mathcal{S}}_{t'}[Z\cup \{v\}]\circ  X.$$
$$ \widehat{\mathcal{S}}_t[Z]= \widehat{\mathcal{S}}_{t'}[Z\cup \{v\}]\diamond  {\cal P}_v[Z].$$
%  Recall that for two families ${\cal A}$ and  ${\cal B}$, we defined  ${\cal A} \circ {\cal B} = \{A \cup B ~:~A \in {\cal A} \wedge B \in {\cal B}\}.$ 

 Now we show that $\widehat{\mathcal{S}}_t$ maintains the invariant of the algorithm. Let $L\in \mathscr{S}$.
 
  \begin{enumerate}
  \label{eqn:twdpone}
    \item  Let $L_t=E(H_t)\cap L$ and $L_R=L\setminus L_t$. Furthermore, edges of $L_t$ can be partitioned into 
    $L_{t'}=E(H_{t'})\cap L$ and $L_v=  L_t\setminus L_{t'}$.  That is, $L_t= L_{t'} \uplus L_v$.  
    \item Let $Z=\partial^t(L)$ and $Z'=\partial^{t'}(L)$. 
    \end{enumerate}
 
 By the property of $\widehat{\mathcal{S}}_{t'}$, there exists a $\hat{L}_{t'}\in \widehat{\mathcal{S}}_{t'}[Z']$ 
such that 
 \begin{eqnarray}
 \label{eqn:twdptwo}
 L\in \mathscr{S}  & \iff &  L_{t'} \uplus L_v \uplus L_R \in \mathscr{S}  \nonumber\\
                            & \iff &  \hat{L}_{t'} \uplus L_v \uplus L_R \in \mathscr{S} 
 \end{eqnarray}
 and $\partial^{t'}(L)=\partial^{t'}(\hat{L}_{t'} \uplus L_v \uplus L_R )=Z'$.  
 
 \medskip
 
 \noindent
 We put { $\hat{L}_t=\hat{L}_{t'} \cup L_v$ and $\hat{L}=\hat{L}_t \cup L_R $.}  We now show that 
$\hat{L}_t \in \widehat{\mathcal{S}}_t[Z]$. Towards this just note that since $Z'=Z$ or $Z'=Z\cup\{v\}$, 
we have that $\widehat{\mathcal{S}}_t[Z]$ contains $\widehat{\mathcal{S}}_{t'}[Z']\diamond \{L_v\}$. 
By \eqref{eqn:twdptwo},  $\hat{L} \in \mathscr{S} $. Finally, we need to show that   $\partial^{t}(\hat{L})=Z$. 
Towards this just note that $\partial^{t}(\hat{L})=Z'\setminus \{v\}=Z$.  This concludes the proof for the 
fact that $\widehat{\mathcal{S}}_t$ maintains the correctness invariant. 
 
\paragraph{Join node $t$ with two children $t_{1}$ and $t_{2}$.} Here, we know that  $X_{t}=X_{t_{1}}=X_{t_{2}}$. 
Also we know that the edges of $H_t$ is obtained by the union of edges of $H_{t_1}$ and $H_{t_2}$ which are disjoint. 
Of course they are separated by the vertices in $X_t$. A natural way to obtain a family of partial solutions for 
$H_t$ is that we take the union of edges subsets of the families stored at nodes $t_1$ and $t_2$. This is exactly 
what we do. Let 
   $$ \widehat{\mathcal{S}}_t[Z]= \widehat{\mathcal{S}}_{t_1}[Z]\diamond  \widehat{\mathcal{S}}_{t_2}[Z].$$ 
   
    Now we show that $\widehat{\mathcal{S}}_t$ maintains the invariant. Let $L\in \mathscr{S}$. 
    \begin{enumerate}
    \item Let $L_t=E(H_t)\cap L$ and $L_R=L\setminus L_t$. Furthermore edges of $L_t$ can be partitioned into those belonging to $H_{t_1}$ and those belonging to  $H_{t_2}$. Let  $L_{t_1}=E(H_{t_1})\cap L$ and $L_{t_2}=E(H_{t_2})\cap L$. Observe that since 
    $E(H_{t_1})\cap E(H_{t_2})=\emptyset$,  we have that $L_{t_1} \cap L_{t_2}=\emptyset$. 
Also observe that $L_t=L_{t_1}\uplus L_{t_2}$ and $G[L_{t_1}],G[L_{t_1}]$ form forests. 
    \item  Let $Z=\partial^t(L)$. Since $X_{t}=X_{t_{1}}=X_{t_{2}}$  this implies  that $Z=\partial^t(L)=\partial^{t_1}(L)=\partial^{t_2}(L)$. 
    \end{enumerate}
 
 Now observe that 
 \begin{eqnarray*}
 L\in \mathscr{S}  & \iff &  L_{t_1} \uplus  L_{t_2} \uplus L_R \in \mathscr{S}  \\
                            & \iff &  \hat{L}_{t_1}  \uplus L_{t_2} \uplus L_R \in \mathscr{S}~~~~\mbox{(by the property of $\widehat{\mathcal{S}}_{t_1}$ we have 
                                     that  $\hat{L}_{t_1}\in \widehat{\mathcal{S}}_{t_1}[Z]$)}\\
                            & \iff &  \hat{L}_{t_1} \uplus \hat{L}_{t_2} \uplus L_R \in \mathscr{S}~~~~\mbox{(by the property of $\widehat{\mathcal{S}}_{t_2}$ we have  that  $\hat{L}_{t_2}\in \widehat{\mathcal{S}}_{t_2}[Z]$)}
 \end{eqnarray*}
 \noindent 
{We put    $\hat{L}_t=\hat{L}_{t_1} \cup \hat{L}_{t_2}$.}  
By the definition of  $\widehat{\mathcal{S}}_t[Z]$,  we have that 
$\hat{L}_{t_1} \cup \hat{L}_{t_2}\in \widehat{\mathcal{S}}_t[Z]$. The above inequalities also show that 
$\hat{L}=\hat{L}_t\cup L_R \in \mathscr{S}$. It remains to show  that 
$\partial^{t}(\hat{L})=Z$.  
  Since $\partial^{t_1}(L)=Z$,  we have that 
 $\partial^{t_1}(\hat{L}_{t_1}  \uplus L_{t_2} \uplus L_R)=Z$. Now since $X_{t_1}=X_{t_2}$ we have that  $\partial^{t_2}(\hat{L}_{t_1}  \uplus L_{t_2} \uplus L_R)=Z$ and thus $\partial^{t_2}(\hat{L}_{t_1}  \uplus \hat{L}_{t_2} \uplus L_R)=Z$. Finally, because $X_{t_2}=X_t$, we conclude  that 
 $\partial^{t}(\hat{L}_{t_1}  \uplus \hat{L}_{t_2} \uplus L_R)=\partial^{t}(\hat{L})=Z$. This concludes the proof of correctness invariant. 
 
 \paragraph{Root node $r$.} Here, $X_{r}=\emptyset$. We go through all the solution in $\widehat{\mathcal{S}}_r[\emptyset]$ and output the one with the 
 minimum weight.   This concludes the description of the dynamic programming algorithm.  

\paragraph{Computation of $\widehat{\mathcal{S}}_t$.}
Now we show how to implement the algorithm described above in the desired running time by making use of  
% Lemma~\ref{lem:sizeinvariant_1} and 
Lemma~\ref{lem:sizeinvariant_2}. For our discussion let us fix a node $t$ and 
$Z\subseteq X_t$ of size $k$. While doing dynamic programming algorithm from the base nodes to the root node 
we always maintain the size invariant. 
% That is, $\widehat{\mathcal{S}}_t[Z]|\leq  2^{k}. $

\paragraph{Base node $t$.} Trivially, in this case we have maintained size invariant. 
% $|\widehat{\mathcal{S}}_t[Z]|\leq   2^{k}$. 
\paragraph{Introduce node $t$ with child $t'$.}
Here, we have that $\widehat{\mathcal{S}}_t[Z]=\widehat{\mathcal{S}}_{t'}[Z\setminus \{v\}]$ and thus 
the number of partial solutions with $i$ connected components in $\widehat{\mathcal{S}}_t[Z]$ is bounded 
${k \choose i}$

\paragraph{Forget node $t$ with child $t'$.} In this case, 
  \[ \widehat{\mathcal{S}}_t[Z]= \widehat{\mathcal{S}}_{t'}[Z\cup \{v\}]\diamond  {\cal P}_v[Z].\] 

It is easy to see that the number of edge subsets with $i$ connected components in 
$\widehat{\mathcal{S}}_{t'}[Z\cup \{v\}]$ and ${\cal P}_v[Z]$ is upper bounded by ${k+1 \choose i}$ 
So we apply Lemma~\ref{lem:sizeinvariant_2} and obtain $\widehat{\mathcal{S}}_t'[Z]$ that maintains the 
correctness and size invariants. % and has size at most $2^k$.  
We update  $\widehat{\mathcal{S}}_t[Z]=\widehat{\mathcal{S}}_t'[Z]$.

 The running time $T$ to compute $\widehat{\mathcal{S}}_t$ (that is, across all subsets of $X_t)$ is % dominated by:
\begin{eqnarray*}
T&=&\cO\left( \sum_{i=1}^{\tw +1} {\tw+1 \choose i}\left( i^{\omega}\left(2^{\omega}+2\right)^i n + i^\omega 2^{i(\omega-1)} 3^i n \right)\right)\\
&=&\cO\left({\tw}^{\omega}n\left(2^{\omega}+3\right)^{\tw}+{\tw}^\omega n\left( 1+ 2^{\omega-1}\cdot 3\right)^{\tw}\right)
% \cO\left(\sum_{i=1}^{\tw +1} \binom{\tw+1}{i} 2^{i(\omega-1)} 4^i  \cdot \tw^{\cO(1)} n \right)= \cO\left( (1+2^{\omega +1})^\tw   \cdot \tw^{\cO(1)} n \right). 
\end{eqnarray*}

\paragraph{Join node $t$ with two children $t_{1}$ and $t_{2}$.}   Here we defined  
   $$ \widehat{\mathcal{S}}_t[Z]= \widehat{\mathcal{S}}_{t_1}[Z]\diamond  \widehat{\mathcal{S}}_{t_2}[Z].$$ 

The number of edge subsets with $i$ connected components in 
$\widehat{\mathcal{S}}_{t_1}[Z]$ and  $\widehat{\mathcal{S}}_{t_2}[Z]$ by ${k \choose i}$. 
Now, we apply Lemma~\ref{lem:sizeinvariant_2} and obtain $\widehat{\mathcal{S}}_t'[Z]$ that maintains the 
correctness invariant and has size at most $2^k$. We {put } $\widehat{\mathcal{S}}_t[Z]=\widehat{\mathcal{S}}_t'[Z]$. 
The running time to compute $\widehat{\mathcal{S}}_t$  is 
\begin{eqnarray*}
\cO\left({\tw}^{\omega}n\left(2^{\omega}+3\right)^{\tw}+{\tw}^\omega n\left( 1+ 2^{\omega-1}\cdot 3\right)^{\tw}\right). 
\end{eqnarray*}
Thus the whole algorithm takes 
$\cO\left({\tw}^{\omega}n^2\left(2^{\omega}+3\right)^{\tw}+{\tw}^\omega n^2\left( 1+ 2^{\omega-1}\cdot 3\right)^{\tw}\right)=\cO(8.7703^{\tw}n^2)$ 
as the number of nodes in a nice tree-decomposition is upper bounded by $\cO(n)$. However, observe that we do not 
need to compute the forests and the  associated weight at every step of the algorithm. The size of the forest is 
at most $\tw+1$ and we can maintain these forests across the bags during dynamic programming in time $\tw^{\cO(1)}$. 
This will lead to an algorithm with the claimed running time. 
% The last remark we would like to make is that one can do better at {\bf forget node} by forgetting a single edge at a time. However, we did not try to optimize this, as the running time to compute the family of 
% partial solutions at  
% {\bf join node} is the most expensive operation. 
This completes the proof. 
\end{proof}

% The approach of  Theorem~\ref{thm:steinertree_DP} can be used  to obtain single-exponential algorithms parameterized 
% by the treewidth of an input graph for many other connectivity problems such as 
% \textsc{Hamiltonian Cycle}, \textsc{Feedback Vertex Set}, and  \textsc{Connected Dominated Set}. For all these problems, checking 
% whether two partial solutions can be glued together to form a global solution can be checked by testing independence in a specific graphic 
% matroid. We believe that there exist interesting problems where this check corresponds to testing independence in a different class of linear matroids. 

%!TEX root = main.tex

\subsection{{\sc Feedback Vertex Set} parameterized by treewidth}
In this section we study the  {\sc Feedback Vertex Set} problem which is defined as follows.   
\medskip
\begin{center} 
\fbox{\begin{minipage}{0.96\textwidth}
\noindent{\sc Feedback Vertex Set} \\ %\hfill {\bf Parameter:} $\tw$ \\
\noindent {\bf Input}: An undirected graph $G$ and a weight function $w~:~V(G)\rightarrow \mathbb{N}$.\\
% \noindent{\phantom{{\em Input}:}} 
%  of width $\tw$.\\
\noindent{\bf Task}: Find a minimum weight set $Y\subseteq V(G)$ 
% $(w(Y)=\sum_{v\in Y}w(v))$ 
such that $G[V(G)\setminus Y]$ is a forest.\\ 
% \noindent{\phantom{{\em Input}:}}
%and $G[ X]$ is \\
%\noindent{\phantom{{\em  Question}:}} connected?
\end{minipage}}
\end{center}
\medskip

Let $G$ be an input  graph of the {\sc Feedback Vertex Set} problem. In this section instead of saying 
feedback vertex set $Y\subseteq V(G)$ is a solution, we say that $V(G)\setminus Y$ is a solution, i.e, our 
objective is to find a maximum weight set $V'\subseteq V(G)$ such that $G[V']$ is a forest. We call 
$V'\subseteq V(G)$ is an {\em optimal solution} if $V'$ is a solution with maximum weight. Let $\mathscr{S}$ be a 
family of vertex subsets such that every vertex subset corresponds to an optimal solution. 
That is, 
$$\mathscr{S}=\{V'\subseteq V(G)~|~V' \mbox{ is an optimal solution}\}.$$  
Let $(\mathbb{T},\mathcal{ X})$  be a tree decomposition of $G$ of width $\tw$. 
For each tree node $t$ and $Z\subseteq X_t$, we define ${\cal S}_t[Z]$, 
{\em family of partial solutions} as follows.  
\begin{eqnarray*}
\mathcal{S}_t[Z]=\{U\subseteq V(G_t)~|~ U\cap X_t=Z \mbox{ and } G_t[U]\mbox{ is a forest }\} 
\end{eqnarray*}
  
% By ${\cal S}_t^e[Z]$ we denote the set $\{E(G_t[U])~|~U\in {\cal S}_t[Z]\}$ and for any set 
% $A\in {\cal S}_t^e[Z]$ we use $A^e$ to denote $E(G_t[A])$.  
We denote by $K^t$ a complete graph on the vertex set $X_t$. Let $G^*$ be subgraph of $G$. 
Let $C_1',\ldots,C_\ell'$ be the connected components of $G^*$ that have nonempty intersection with $X_t$.  
Let $C_i=C_i'\cap X_t$. By $F(G^*)$ we denote the a forest $\{Q_1,\ldots,Q_\ell\}$ where each $Q_i$ is an 
arbitrary spanning tree of $K^t[C_i]$. 
% For any graph $G$ and a family ${\cal P}\subseteq 2^{V(G)}$, we denote 
% ${\cal P}^e$ to represent the set $\{E(G[U])~|~U\in {\cal P}\}$. 

For two family of vertex subsets ${\cal P}$ and ${\cal Q}$ of a graph $G$, 
we denote 
$$ {\cal P} \otimes {\cal Q}=\{U_1\cup U_2 ~|~ U_1\in {\cal P}, U_2\in {\cal Q} \mbox{ and } 
G_t[U_1\cup U_2] \mbox{ is a forest }\}.$$

%A short description of our algorithm is as follows. For every node $t$ of ${\mathbb T}$  
%we keep a subset of family of partial solutions of the vertex set $V(G_t)$.   
%That is for every optimal solution $L\in \mathscr{S}$ with $L\cap X_t=Z$ and its intersection 
%$L_t=V(G_t)\cap L$ with the graph $G_t$, we have some partial solution $\hat{L_t}$ in our subset such that 
%$\hat{L_t}\cap X_t=Z$ and $\hat{L_t}\cup L_R$ is an optimal solution, i.e $G[\hat{L_t}\cup L_R]$ is a 
%forest, where $L_R=L\setminus L_t$ and $w(\hat{L_t}\cup L_R)\geq w(L)$. Since we are looking for a maximum 
%weight forest and any bag in the tree decomposition 
%splits the forest into two parts with at most $\tw$ many common vertices we could use the max representative set tool 
%of graphic matroid to compute a {\em representative family} for partial solutions in each bag in bottom up fashion. 
Now we are ready to state the main theorem. 

\begin{theorem}\label{thm:fvs_DP}
Let $G$ be an $n$-vertex graph given together with its tree decomposition of with $\tw$. Then 
{\sc Feedback Vertex Set} on $G$ can be solved in time $\cO\left( \left(1+ 2^{\omega-1}\cdot 3\right)^{\tw} 
 {\tw}^{\cO(1)}n \right)$.
%$$\cO\left(\left( \left(2^{\omega}+3\right)^{\tw}+ \left( 1+ 2^{\omega-1}\cdot 3\right)^{\tw}\right) {\tw}^{O(1)}n \right).$$
\end{theorem}
\begin{proof}
% We first outline an algorithm with running time $\cO((1+2^{\omega+1})^{\tw} \tw^{\cO(1)}n^2)$ for a simple 
% exposition. Later we point out how we can remove the extra factor of $n$ at the cost of a factor polynomial in $\tw$. 

For every node $t$ of  $\mathbb T$ and $Z\subseteq X_t$, we store a family of vertex subsets  
$\widehat{\mathcal{S}}_t[Z]$ of $V(G_t)$  satisfying the following correctness invariant.
\begin{quote}
%\item 
{\bf Correctness Invariant:} For every $L\in \mathscr{S}$ we have the following. 
Let $L_t=V(G_t)\cap L$, $L_R=L\setminus L_t$ and $L\cap X_t=Z$. Then there exists 
$\hat{L}_t\in \widehat{\mathcal{S}}_t[Z]$ such that 
% $w(\hat{L}_t)\geq w({L}_t)$, 
$\hat{L}=\hat{L}_t\cup L_R$ is an optimal solution, i.e $G[\hat{L}_t\cup L_R]$ is a forest with 
$w(\hat{L}_t)\geq w({L}_t)$
% at least $|L|$  vertices. 
% Observe that since $w(\hat{L}_t)\leq w(L_t)$ and $L\in \mathscr{S}$, we have 
Thus we have that $\hat{L} \in \mathscr{S}$.
\end{quote}

We process the nodes of the tree $\mathbb T$  from base nodes to the root node while doing the dynamic programming. 
Throughout the process we maintain the correctness invariant, which will prove the correctness of the algorithm. 
However, our main idea is to use representative sets  to obtain   $\widehat{\mathcal{S}}_t[Z]$ of small size. 
That is, given the set $\widehat{\mathcal{S}}_t[Z]$ 
% (as a product of two families ${\cal A}$ and ${\cal B}$, i.e
%  $\widehat{\mathcal{S}}_t[Z]= {\cal A}\bullet{\cal B}$) 
that satisfies the correctness invariant,  
we use {\em representative set} tool 
% Theorem~\ref{thm:repsetlovaszweighted} or Corollary~\ref{cor:product_general_matroid} 
to obtain a subset 
$\widehat{\mathcal{S}}_t'[Z]$ of $\widehat{\mathcal{S}}_t[Z]$ that also satisfies the  correctness invariant 
and has size upper bounded by $2^{|Z|}$ in total. More precisely, the number of partial solutions in $\widehat{\mathcal{S}}_t'[Z]$ that have $i$  connected components with nonempty intersection with $X_t$  is upper 
bounded by ${|Z| \choose i}$. Thus, we maintain the following size invariant.

\begin{quote}
%\item 
{\bf Size Invariant:} After  node $t$ of $\mathbb T$ is processed by the algorithm, we have that 
 $|\widehat{\mathcal{S}}_t[Z,i]|\leq {|Z| \choose i}$, where $\widehat{\mathcal{S}}_t[Z,i]$ is the set of partial 
solutions  that have $i$  connected components with nonempty intersection with $X_t$. 
%with number of connected components having non empty intersection with $X_t$ is exactly $i$ in 
%$\widehat{\mathcal{S}}_t[Z]$. 
\end{quote}
 
% The ingredient of the dynamic programming algorithm is the use of 
% Theorem~\ref{thm:repsetlovaszweighted} and Corollary~\ref{cor:product_general_matroid} to compute 
% $\widehat{\mathcal{S}}_t[Z]$ maintaining the size invariant.

\begin{lemma}[Product Shrinking Lemma]
\label{lem:sizeinvariant_3}
Let $t$ be a join node of  $\mathbb T$ with children $t_1$ and $t_2$. Let  $Z\subseteq X_t$ be a set of size 
$k$. Let $\widehat{\mathcal{S}}_{t_1}[Z]$ and $\widehat{\mathcal{S}}_{t_2}[Z]$ be two family of vertex subsets of 
$V(G_{t_1})$ and $V(G_{t_1})$ satisfying the size and correctness invariants. Furthermore, let  
$\widehat{\mathcal{S}}_t[Z]=\widehat{\mathcal{S}}_{t_1}[Z] \otimes \widehat{\mathcal{S}}_{t_2}[Z]$ be the family of 
vertex subsets of $V(G_t)$ satisfying the correctness invariant.  
Then in time $\cO\left( k^{\omega}\left(2^{\omega}+2\right)^k n + k^\omega 2^{k(\omega-1)} 3^k n \right)$ 
we can compute $\widehat{\mathcal{S}}_t'[Z] \subseteq \widehat{\mathcal{S}}_t[Z]$ 
satisfying correctness and size invariants. 
\end{lemma}
 \begin{proof}
 We start by associating a matroid with  node $t$ and the set $Z\subseteq X_t$  as follows. We consider 
a graphic matroid \mat{} on $K^t[Z]$. Here, the element set $E$ of the matroid is  the edge set $E(K^t[Z])$ and 
the family of independent sets $\cal I$ consists of spanning forests of  $K^t[Z]$. 
Here our objective is to find a small subfamily of 
$\widehat{\mathcal{S}}_t[Z]=\widehat{\mathcal{S}}_{t_1}[Z] \otimes \widehat{\mathcal{S}}_{t_2}[Z]$
satisfying correctness and size invariants using efficient computation of representative family in the 
graphic matroid $M$. For an independent set $U\in \widehat{\mathcal{S}}_{t_1}[Z]\cup \widehat{\mathcal{S}}_{t_2}[Z]$, for $U_1\in \widehat{\mathcal{S}}_{t_1}[Z]$ and $U_2\in \widehat{\mathcal{S}}_{t_2}[Z]$, it is natural to associate $F(G[U_1])\cup F(G[U_2])$ as the corresponding independent set in the graphic matroid. 
%
%When we associate an independent set $F(G[U])$ for a set 
%$U\in \widehat{\mathcal{S}}_{t_1}[Z]\cup \widehat{\mathcal{S}}_{t_2}[Z]$, for $U_1\in \widehat{\mathcal{S}}_{t_1}[Z]$
%and $U_2\in \widehat{\mathcal{S}}_{t_2}[Z]$,
However,  $F(G[U_1])\cup F(G[U_2])$ may not form a forest even if $G[U_1\cup U_2]$ is a forest. This 
happens precisely when there exists an edge in $Z$. To overcome this difficulty we associate 
$F(G[U]\setminus E(Z))$ with any $U\in \widehat{\mathcal{S}}_{t_2}[Z]$. We can observe that for any 
$U_1\in \widehat{\mathcal{S}}_{t_1}[Z]$ and $U_2\in \widehat{\mathcal{S}}_{t_2}[Z]$,  $G[U_1\cup U_2]$ is a forest 
if and only if $F(G[U_1])\cup F(G[U_2]\setminus E(Z))$ is a forest in $K^t[Z]$. 
 
Let $\widehat{\mathcal{S}}_{t_1}[Z]=\{A_1,\ldots,A_\ell\}$ and $\widehat{\mathcal{S}}_{t_2}[Z]=\{B_1,\ldots,B_{\ell'}\}$. 
Let ${\cal L}_1=\{ F(G[A_1]),\ldots,F(G[A_\ell])\}$ and ${\cal L}_2=\{ F(G[B_1]\setminus E(Z)),\ldots,F(G[B_{\ell'}]\setminus E(Z))\}$ 
be the set of forests in $K^t[Z]$ corresponding to the vertex subsets in $\widehat{\mathcal{S}}_{t_1}[Z]$ and 
$\widehat{\mathcal{S}}_{t_2}[Z]$ respectively. 
For each $F(G[A_i])\in {\cal L}_1$ we set $w(F(G[A_i]))=w(A_i)$, and for each 
$F(G[B_j]\setminus E(Z))$ we set $w(F(G[B_j]\setminus E(Z)))=w(B_j\setminus Z)$. 
For $i\in [k]$ and $r\in \{1,2\}$, 
let ${\cal L}_{r,i}$ be the family of forests of ${\cal L}_r$ with $i$ edges. 
Now we apply Theorem~\ref{thm:repsetlovaszweighted} and compute 
$\widehat{\cal L}_{2,j}\subseteq_{maxrep}^{k-1-j}{\cal L}_{2,j}$ for all $j$, of size ${k-1 \choose j}$ in time 
$\cO(2^k {k \choose j}^{w-1})$ (because $|{\cal L}_{2,j}|\leq 2^k$). 
Now we apply Corollary~\ref{cor:product_general_matroid} and find 
$\widehat{{\cal L}_{1,i}\bullet \widehat{{\cal L}}_{2,j}} \subseteq_{maxrep}^{k-1-i-j} {\cal L}_{1,i}\bullet {\cal L}_{2,j}$ 
of size ${k-1 \choose i+j}$ for all $i,j\in [k]$. 
Let $\widehat{\mathcal{S}}_t'[Z,k-m] \subseteq \widehat{\mathcal{S}}_t[Z,k-m]$ be such that for every 
$U_1\cup U_2\in \widehat{\mathcal{S}}_t'[Z,k-m] $ we have that 
$F(G[U_1])\cup F(G[U_2]\setminus Z) \in \cup_{i+j=m}\widehat{{\cal L}_{1,i}\bullet \widehat{{\cal L}}_{2,j}}$.  
Let $\widehat{\mathcal{S}}_t'[Z]= \cup_{j=0}^k \widehat{\mathcal{S}}_t'[Z,j]$. 
By Corollary~\ref{cor:product_general_matroid}, 
$|\widehat{\mathcal{S}}_t'[Z,k-m]| \leq k {k-1 \choose m}\leq {k \choose k-m}$, and hence 
$\widehat{\mathcal{S}}_t'[Z]$ maintains the size invariant.  

Now we show that the $\widehat{\mathcal{S}}_t'[Z]$ maintains the correctness invariant. 
Let $L\in \mathscr{S}$ and let $L_t=V(G_t)\cap L$, $L_R=L\setminus L_t$ and $Z=L\cap X_t$. 
Since $\widehat{\mathcal{S}}_t[Z]$ satisfy correctness invariant, there exists 
$\hat{L}_t \in \widehat{\mathcal{S}}_t[Z]$ such that $w(\hat{L}_t)\geq w(L_t)$,   
$\hat{L}= \hat{L}_t \cup L_R$ is an optimal solution and $\hat{L}\cap X_t=Z$. 
Since $\widehat{\mathcal{S}}_t[Z]=\widehat{\mathcal{S}}_{t_1}[Z] \otimes \widehat{\mathcal{S}}_{t_2}[Z]$, 
there exists $U_1\in \widehat{\mathcal{S}}_{t_1}[Z]$ and 
$U_2\in \widehat{\mathcal{S}}_{t_2}[Z]$ such that $\hat{L}_t= U_1\cup U_2$. 
Observe that $G[U_1\cup U_2]$ form a forest.    
Consider the forests $F(G[U_1])$ and $F(G[U_2]\setminus E(Z))$. 
Suppose $|F(G[U_1])|=i_1$ and $|F(G[U_2]\setminus E(Z))|=i_2$, 
then $F(G[U_1])\cup F(G[U_2]\setminus E(Z))\in {\cal L}_{1,i_1}\bullet {\cal L}_{1,i_2}$. 
This is because, if $F(G[U_1])\cup F(G[U_2]\setminus E(Z))$ contain a cycle, then corresponding to that cycle we 
can get a cycle in $G[U_1\cup U_2]$, which is a contradiction. 
Now let $E'=F(G[L_R\cup Z]\setminus E(Z))$ be the forest corresponding to $L_R\cup Z$ 
with respect to the bag $X_t$. Since $\hat{L}$ is a solution, we have that 
$F(G[U_1])\cup F(G[U_2]\setminus E(Z))\cup E'$ is a 
spanning tree in  $K^t[Z]$. Since $\widehat{{\cal L}_{1,i_1}\bullet \widehat{\cal L}_{2,i_2} }\subseteq_{maxrep}^{k-1-i_1-i_2}
{\cal L}_{1,i_1}\bullet {\cal L}_{2,i_2} $, we have that there exists 
a forest $F(G[U_1'])\cup F(G[U_2']\setminus E(Z)) \in \widehat{{\cal L}_{1,i_1}\bullet \widehat{\cal L}_{2,i_2} }$ 
such that $w(F(G[U_1'])\cup F(G[U_1']\setminus E(Z)) \geq w(F(G[L_t])) $ and $F(G[U_1'])\cup F(G[U_2']\setminus E(Z)) \cup E'$ 
is a spanning tree in  $K^t[Z]$. Thus, we can conclude that $U_1\cup U_2 \cup L_R$ is an optimal solution and 
$U_1\cup U_2 \in \widehat{\mathcal{S}}_t'[Z]$. 
This proves that $\widehat{\mathcal{S}}_t'[Z]$ maintains the correctness invariant.   
 
Since we are applying Corollary~\ref{cor:product_general_matroid} the running time to compute 
$\widehat{\mathcal{S}}'_t[Z]$ is upper bounded by, $\cO\left( k^{\omega}\left(2^{\omega}+2\right)^k n + k^\omega 2^{k(\omega-1)} 3^k n \right).$
%\begin{eqnarray*}
%\cO\left( k^{\omega}\left(2^{\omega}+2\right)^k n + k^\omega 2^{k(\omega-1)} 3^k n \right)
%\end{eqnarray*}
%  For a given edge set we also need to compute the forest and that can take $\cO(n)$ time. 
 \end{proof}

We now explain the dynamic programming algorithm over the tree-decomposition $(\mathbb{T},\mathcal{ X})$ of $G$ 
and prove that it maintains the correctness invariant. We assume that $(\mathbb{T},\mathcal{ X})$ is a nice 
tree-decomposition of $G$. By $\widehat{\mathcal{S}}_t$ we denote 
$\cup_{Z\subseteq X_t} \widehat{\mathcal{S}}_t[Z]$ (also called a \emph{representative family of partial solutions}). 
 We show how $\widehat{\mathcal{S}}_t$  is obtained by doing dynamic programming from base node to the root node.

\paragraph{Base node $t$.}  Here the graph $G_t$ is empty and thus we take $\widehat{\mathcal{S}}_t=\emptyset$.  

\paragraph{Introduce node $t$ with child $t'$.} Here, we know that $X_{t}\supset X_{t'}$ and  $|X_{t}|=|X_{t'}|+1$. 
Let $v$ be the vertex in $X_{t}\setminus X_{t'}$. The graph $G_t= G_{t'}\setminus\{v\}$. So each partial solution 
in $G_{t'}$ is a partial solution in $G_t$ or it differs at vertex $v$ from a partial solution in $G_t$, i.e,  

$$\widehat{\mathcal{S}}_{t}[Z]=
     \left\{ \begin{array}{lcl}
\widehat{\mathcal{S}}_{t'}[Z] & \mbox{if} & v\notin Z \\ 
\left\{ U\cup\{v\}~|~ U \in \widehat{\mathcal{S}}_{t'}[Z\setminus \{v\}]\mbox{ and } G[U\cup\{v\}] \mbox{ is a forest }\right\} 
& \mbox{if} & v\in Z 
\end{array}\right.$$
% In other words, 
% $$\widehat{\mathcal{S}}_{t}[Z]=
%      \left\{ \begin{array}{lcl}
% \widehat{\mathcal{S}}_{t'}[Z] & \mbox{if} & v\notin Z \\ 
% \widehat{\mathcal{S}}_{t'}[Z]\bullet \{Z\} 
% & \mbox{if} & v\in Z 
% \end{array}\right.$$
When $v\notin Z$, $\widehat{\mathcal{S}}_{t}[Z]$ satisfies correctness and size invariant. 
When $v\in Z$, $|\widehat{\mathcal{S}}_{t}[Z,i]|\leq 2^k$ and we can apply Theorem~\ref{thm:repsetlovaszweighted} 
by associating a family of independent sets in $K^t[Z]$ (like in Lemma~\ref{lem:sizeinvariant_3}) and find 
$\widehat{\mathcal{S}}'_{t}[Z,i]\subseteq\widehat{\mathcal{S}}_{t}[Z,i]$ satisfies correctness and size invariant 
in time $\cO(2^k{k\choose i}^{w-1})$.
\paragraph{Forget node $t$ with child $t'$.} Here we know $X_{t}\subset X_{t'}$, $|X_t|=|X_{t'}|-1$ and $G_t=G_{t'}$. 
Let $X_t'\setminus X_t=\{v\}$. So for any $Z\subseteq X_{t}$ we have 
$\widehat{\cal S}_{t}[Z]=\widehat{\cal S}_{t'}[Z]\cup \widehat{\cal S}_{t'}[Z\cup\{v\}]$. 
The number of elements in $\widehat{\cal S}_{t}[Z]$ with $i$ number of connected components intersecting with 
$X_t$ is upper bounded by ${k+1 \choose i}+{k+1 \choose i+1}\leq {k+2 \choose i}$. 
Again by applying Theorem~\ref{thm:repsetlovaszweighted} we can find 
$\widehat{\mathcal{S}}'_{t}[Z,i]\subseteq\widehat{\mathcal{S}}_{t}[Z,i]$ satisfies correctness and size invariant 
in time $\cO({k+2 \choose i}{k\choose i}^{w-1})$.

\paragraph{Join node $t$ with two children $t_{1}$ and $t_{2}$.} Here, we know that  $X_{t}=X_{t_{1}}=X_{t_{2}}$. 
The natural way to get a family of partial solutions for $X_t$ is the union of vertex sets of two families stored 
at node $t_1$ and $t_2$ which form a forest, i.e,
\begin{eqnarray*}
\widehat{\cal S}_t[Z]&=&\{U_1\cup U_2~|~ U_1\in\widehat{\cal S}_{t_1}[Z], U_2\in\widehat{\cal S}_{t_2}[Z], G[U_1\cup U_2] \mbox{ is a forest}\}\\
 &=&\widehat{\cal S}_{t_1}[Z]\otimes \widehat{\cal S}_{t_2}[Z]
\end{eqnarray*}

Now we show that $\widehat{\mathcal{S}}_t$ maintains the invariant. Let $L\in \mathscr{S}$. 
Let $L_t=V(G_t)\cap L, L_{t_1}=V(G_{t_1})\cap L, L_{t_2}=V(G_{t_2})\cap L$  and $L_R=L\setminus L_t$. 
Let $Z=L\cap X_t$ 
Now observe that 
 \begin{eqnarray*}
 L\in \mathscr{S}  & \iff &  L_{t_1} \cup  L_{t_2} \cup L_R \in \mathscr{S}  \\
                   & \iff &  \hat{L}_{t_1}  \cup L_{t_2} \cup L_R \in \mathscr{S}~~~~\mbox{(by the property of $\widehat{\mathcal{S}}_{t_1}$ we have 
                                     that  $\hat{L}_{t_1}\in \widehat{\mathcal{S}}_{t_1}[Z]$)}\\
                   & \iff &  \hat{L}_{t_1} \cup \hat{L}_{t_2} \cup L_R \in \mathscr{S}~~~~\mbox{(by the property of $\widehat{\mathcal{S}}_{t_2}$ 
                             we have  that  $\hat{L}_{t_2}\in \widehat{\mathcal{S}}_{t_2}[Z]$)}
 \end{eqnarray*}
 \noindent 
{We put    $\hat{L}_t=\hat{L}_{t_1} \cup \hat{L}_{t_2}$.}  
By the definition of  $\widehat{\mathcal{S}}_t[Z]$,  we have that 
$\hat{L}_{t_1} \cup \hat{L}_{t_2}\in \widehat{\mathcal{S}}_t[Z]$. 
The above inequalities also show that 
$\hat{L}=\hat{L}_t\cup L_R \in \mathscr{S}$. 
Note that $(\hat{L}_t\cup L_R)\cap X_t=Z$
This concludes the proof of correctness invariant.

We apply Lemma~\ref{lem:sizeinvariant_3} and find 
$\widehat{\mathcal{S}}'_{t}[Z]\subseteq\widehat{\mathcal{S}}_{t}[Z]$ satisfies correctness and size invariant 
in time $\cO\left( k^{\omega}\left(2^{\omega}+2\right)^k n + k^\omega 2^{k(\omega-1)} 3^k n \right)$.
\paragraph{Root node $r$.} Here, $X_{r}=\emptyset$. We go through all the solution in 
$\widehat{\mathcal{S}}_r[\emptyset]$ and output the one with the 
maximum weight.

In worst case, in every tree node $t$, for all subset $Z\subseteq X_t$, we apply Lemma~\ref{lem:sizeinvariant_3}. So 
by doing the same run time analysis as in the case of Steiner Tree, the total running time will be upper bounded by
$\cO\left(\left( \left(2^{\omega}+3\right)^{\tw}+ \left( 1+ 2^{\omega-1}\cdot 3\right)^{\tw}\right) {\tw}^{O(1)}n \right).$
\end{proof}

%\subsection{Multilinear Monomial Testing}
%\input{monomial_testing.tex}
% \subsection{k-Minimum Spanning Tree}
% \input{k_MIST.tex}

%%%%%%%
\section{$k$-Path}
%!TEX root = main.tex
In this section we outline a parameterized algorithm for the $k$-{\sc Path} problem with running time $2.619^kn^{O(1)}$. The complete details of a $2.619^kn\log^2 n$ time algorithm will appear in the full version of~\cite{FominLS13}. The algorithm is basically an adaptation of the $k$-{\sc Path} algorithm of Fomin et al.~\cite{FominLS13}, but using generalized separating collections, rather than separating collections, in order to make a trade-off between the size of computed representative families and the time it takes to compute them. We start by giving a brief recolloection of the algorithm of Fomin et al.~\cite{FominLS13}.

Given as input a graph $G$ and integer $k$  we add a source vertex $s$ and make $s$ adjacent to all vertices in the input graph $G$, call the resulting graph $G'$. Every path of length $k$ in $G$ corresponds to a path rooted at $s$ of length $k+1$ in $G'$, and vice versa. Thus we look for such a path in $G'$. For a vertex $v \in V(G)$ define
\begin{eqnarray*}
 {\cal P}_{v}^i& = & \Big\{X~\Big|~X\subseteq V(G'),~v,s \in X, ~|X|=i \mbox{ and there is a path from $s$ to $v$ of length $i$} \\ 
   & & \hspace{1cm} \mbox{ in $G'$ with all the vertices belonging to $X$}. \Big\}
\end{eqnarray*}
It is easy to see that the following recurrence holds for the sets ${\cal P}_{v}^i$:
$${\cal P}_{v}^i = \bigcup_{u \in N_G(v)} \left[ {\cal P}_{u}^{i-1} \bullet \{v\} \right].$$
The correctness of this recurrence is formally proved in~\cite{FominLS13}. The aim now is to compute, for every $v \in V(G)$ and $i \leq k+1$ a $(k+1-i)$-representative family $\widehat{{\cal P}}_{v}^{i} \subseteq {\cal P}_{v}^i$. Fomin et al.~\cite{FominLS13} show that if for every $v$, $\widehat{{\cal P}}_{v}^{i-1}$ is a $(k+2-i)$-representative family of ${\cal P}_{v}^{i-1}$ and $\widehat{{\cal P}}_{v}^{i}$ is a $(k+1-i)$-representative family of 
$$\widetilde{\cal P}^i_v  = \bigcup_{u \in N_G(v)} \left[ \widehat{{\cal P}}_{u}^{i-1} \bullet \{v\} \right],$$
then $\widehat{{\cal P}}_{v}^{i}$ is a $(k+1-i)$-representative family of ${\cal P}_{v}^{i}$.
The algorithm first sets $\widehat{\cal P}^2_v  = {\cal P}^2_v = \{\{s,v\}\}$. Then, for each $i \geq 3$ in increasing order, the algorithm first computes $\widetilde{\cal P}^i_v$ using the recurrence above and then computes a $(k+1-i)$ representative family  $\widehat{{\cal P}}_{v}^{i}$ of  $\widetilde{\cal P}^i_v$ of size ${k+1 \choose i}$. Finally it is easy to see that $G'$ has a path of length $k+1$ rooted at $s$ if and only if some family $\widehat{\cal P}^{k+1}_v$ is non-empty.

The dependence on $k$ in the running time is determined by the running time of the step where a representative family  $\widehat{{\cal P}}_{v}^{i}$ of  $\widetilde{\cal P}^i_v$ is computed. This running time, in turn, depends on 
$|\widetilde{{\cal P}}_{v}^{i}|$, which is upper bounded by $n \cdot \max_u |\widehat{{\cal P}}_{u}^{i-1}|$. In the algorithm of Fomin et al~\cite{FominLS13} each family $\widehat{{\cal P}}_{u}^{i-1}$ is a $(k+2-i)$-representative family of size approximately ${k+1 \choose i-1}$. Simple calculus shows that for any $i$ the running time of the algorithm of  Fomin et al~\cite{FominLS13} is upper bounded by $2.851^kn^{O(1)}$.

Our new algorithm proceeds in exactly the same manner, but with one crucial difference. For each $i \leq k$ the algorithm appropriately selects a probability variable $x_i$ between $0$ and $1$. When the algorithm computes a $(k+1-i)$-representative family  $\widehat{{\cal P}}_{v}^{i}$ of  $\widetilde{\cal P}^i_v$, in the place where the algorithm of Fomin et al. constructs a $(n,i,k+1)$-separating collection, our algorithm uses a generalized $(n,i,k+1-i)$-separating collection with constant $x_i$ instead. This has two effects. First, the running time for computing  $\widehat{{\cal P}}_{v}^{i}$ is {\em decreased} to roughly $|\widetilde{\cal P}^i_v| \cdot (1-x_i)^{-(k-i)}$. Second, the {\em size} of the family  $\widehat{{\cal P}}_{v}^{i}$ is {\em increased} to approximately $x_i^{-i}(1-x_i)^{-(k-i)}$. The increase in the size of the output family then affects the running time of the next iteration of the algorithm, since it increases the size of $\widetilde{\cal P}^{i+1}_v$.  However, it is possible to show that one can choose $x_i$ for every $i$ such that the savings in the running time outweigh the loss caused due to the increased size of the representative family. Specifically setting 
$$x_i = \frac{\frac{i}{k+1-i}}{2-\frac{i}{k+1-i}}$$
yields an upper bound of $2.619^kn^{O(1)}$ for the total running time.

% Now using Lemma~\ref{lem:pathrepsetfinder} we compute $\widehat{{\cal P}}_{sv}^{k+2} \subseteq_{rep}^{k+2}  {\cal P}^{k+2}_{sv}$ for all 
% $v\in V(G)\setminus \{s\}$  in time 

% \[ \cO\left( 2^{o(k)}   m \log^2 n  \max_{i\in [k+2]} \left\{{k+2 \choose i-1}   \left( \frac{k+2}{k+2-i} \right)^{k+2-i} \right\} \right).\]
%The maximum is achieved at $i=\alpha k$ where $\alpha=1+\frac{1-\sqrt{1+4e}}{2e}$. Thus, the running time of the algorithm is 
%$\cO(2.8505^{k}\cdot 2^{o(k)}m \log^2 n)=\cO(2.851^{k}\cdot m \log^2 n).$

\section{Conclusion}
%!TEX root = main.tex
In this paper we gave algorithms for finding representative sets for product families 
that are faster that the naive computation for these families. We showed their applicability by   designing 
the best known deterministic algorithms for {\sc $k$-wMlD}, {\sc $k$-wMMlD} and for ``connectivity problems'' parameterized by treewidth.  One of 
the main technical components of our algorithm is the deterministic construction of 
generalized separating collections. We believe that this pseudo-random object, as well as our algorithms for computing representative sets of product families,  will be useful to accelerate other algorithms.
% similar to the ones obtained using fast subset convolution algorithm~\cite{BjorklundHKK07}.  
 %Finding other algorithmic applications of  the tools developed in this paper remain an interesting open problem.  
We conclude with several  interesting problems. 
\begin{enumerate}
\item What are the other natural set families for which we can find representative sets faster than by directly applying the results of Fomin et al.~\cite{FominLS13}?
\item  Can we find representative sets for a uniform matroid in time linear in the input size? 
\item  Does there exist a deterministic algorithm for {\sc $k$-wMlD} running in time $2^k n^{\cO(1)} \log W$?
% \item Could we use representative sets to find a deterministic algorithm for {\sc Hamiltonian Path} faster than 
 %$2^nn^{\cO(1)}$. 
\end{enumerate}

\end{document}